\documentclass[a4paper, 11pt]{article}

\usepackage[utf8]{inputenc}
\usepackage[T1]{fontenc}
\usepackage{geometry}
\usepackage[francais]{layout}
\usepackage{amsthm}
\usepackage{amsmath}
\usepackage{array,booktabs}
\usepackage{multirow}
\usepackage{array}
\usepackage{comment}
\usepackage{listings}
\usepackage{verbatim}
\usepackage[nottoc, notlof, notlot]{tocbibind}
\usepackage{textcomp}
\usepackage{lmodern}
\usepackage{amssymb}
\usepackage{fancyvrb}
\usepackage{graphicx}
\usepackage{enumerate}
\usepackage{varioref}
\usepackage[dvipsnames]{xcolor}
\usepackage{caption}
\usepackage{wrapfig}
\usepackage{soul}
\usepackage{tikz}
\usepackage{hyperref}
\hypersetup{
colorlinks=true,
linkcolor=MidnightBlue,
urlcolor=BrickRed,
citecolor=OliveGreen,
}
\usepackage[capitalise]{cleveref}
\usetikzlibrary{shapes.multipart}
\usetikzlibrary{positioning}
\usetikzlibrary{arrows}
\usetikzlibrary{fit}
\usepackage{subcaption}
\usepackage{comment}
\usepackage{epsfig}
\usepackage[
backend=biber,
style=alphabetic,
citestyle=authoryear,
uniquename=false
]{biblatex}
\usepackage{stackengine}
\usepackage{comment}

    \usepackage{etoolbox}

    \makeatletter
    \patchcmd{\lsthk@SelectCharTable}{%
      \lst@ifbreaklines\lst@Def{`)}{\lst@breakProcessOther)}\fi}{}{}{}

    \makeatother
    \makeatletter
    \newcounter {subsubsubsection}[subsubsection]
    \renewcommand\thesubsubsubsection{\thesubsubsection .\@alph\c@subsubsubsection}
    \newcommand\subsubsubsection{\@startsection{subsubsubsection}{4}{\z@}%
                                         {-3.25ex\@plus -1ex \@minus -.2ex}%
                                         {1.5ex \@plus .2ex}%
                                         {\normalfont\normalsize\bfseries}}
    \newcommand\l@subsubsubsection{\@dottedtocline{3}{10.0em}{4.1em}}
    \newcommand{\subsubsubsectionmark}[1]{}
    \makeatother
\newcommand*{\colorboxed}{}
\def\colorboxed#1#{%
  \colorboxedAux{#1}%
}
\newcommand*{\colorboxedAux}[3]{%
  \begingroup
    \colorlet{cb@saved}{.}%
    \color#1{#2}%
    \boxed{%
      \color{cb@saved}%
      #3%
    }%
  \endgroup
}
\newcommand\R{\mathbb{R}}
\newcommand\N{\mathbb{N}}

\newcommand{\bl}[1]{{\color{Black}{#1}}}

\DeclareMathOperator{\tr}{tr}

\newcounter{exo}
\newcommand{\vertiii}[1]{{\left\vert\kern-0.25ex\left\vert\kern-0.25ex\left\vert #1 
    \right\vert\kern-0.25ex\right\vert\kern-0.25ex\right\vert}}
\renewcommand{\textbf}[1]{\begingroup\bfseries\mathversion{bold}#1\endgroup}

\makeatletter
\def\thmhead@plain#1#2#3{%
  \thmname{#1}\thmnumber{\@ifnotempty{#1}{ }\@upn{#2}}%
  \thmnote{{\the\thm@notefont: \textbf{#3}}}}
\let\thmhead\thmhead@plain
\makeatother

\newtheorem{prop}{Proposition}[section]
\labelformat{prop}{Proposition~#1}

\labelformat{formalprop}{Formal Proposition~#1}

\newtheorem{theorem}{Theorem}[section]

\newtheorem{rem}{Remark}
\newtheorem{workassump}{Working assumption}

\labelformat{cor}{Corollary}
\labelformat{fig}{Figure~#1}
\usepackage{thmtools}

\begin{document}
\title{The best of both worlds: combining population genetic and quantitative genetic models}
\date{\today}
\author{L. Dekens\footnote{Corresponding author: \href{mailto:leonard.dekens@gmail.com}{leonard.dekens@gmail.com},\href{mailto:dekens@math.univ-lyon1.fr}{dekens@math.univ-lyon1.fr}} \footnote{Institut Camille Jordan, UMR5208 UCBL/CNRS, Universit\'e de Lyon, 69100 Villeurbanne, France.} \footnote{Université Paris Cité, CNRS, MAP5, F-75006, Paris, France.}, S.P. Otto\footnote{Department of Zoology, University of British Columbia, Vancouver, BC Canada.} and V. Calvez\footnotemark[2] \footnote{Team-projet Inria Dracula, Lyon, France.}}

\maketitle

\begin{abstract}
	Numerous traits under migration-selection balance are shown to exhibit complex patterns of genetic architecture \bl{with large variance in effect sizes}. \bl{However, the conditions under which such genetic architectures are stable have yet to be investigated}, because studying the influence of a large number of small allelic effects on the maintenance of spatial polymorphism is mathematically challenging, due to the high complexity of the systems that arise. \bl{In particular, in the most simple case of a haploid population in a two-patch environment, while it is known from population genetics that polymorphism at a single major-effect locus is stable \bl{in the symmetric case}, there exists no analytical predictions on how this polymorphism holds when a polygenic background also contributes to the trait}. Here we propose to answer this question by introducing a new eco-evo methodology that allows us to take into account the combined contributions of a major-effect locus and of a quantitative background resulting from small-effect loci, where inheritance is encoded according to an extension to the infinitesimal model. In a regime of small variance contributed by the quantitative loci, we justify that traits are concentrated around the major alleles, according to a normal distribution, using new convex analysis arguments. This \bl{allows a reduction in} the complexity of the system using \bl{a separation of time scales} approach. We predict an undocumented phenomenon of loss of polymorphism at the major-effect locus despite strong selection for local adaptation, \bl{because} the quantitative background slowly disrupts the \bl{rapidly} established polymorphism at the major-effect locus, which is confirmed by individual-based simulations. \bl{Our study highlights how segregation of a quantitative background can greatly impact the dynamics of \bl{major-effect loci by provoking migrational meltdowns}.} \bl{We} also provide a comprehensive toolbox designed to describe how to apply our method to more complex population genetic models.
\end{abstract}
\newpage
\tableofcontents
\section{Introduction}

\paragraph{Biological motivation.}
Many species, if not most, evolve in heterogeneous habitats, where varying selection acts upon phenotypic traits in a manner that causes local adaptation. The genetic architecture that underlies those traits is known to present an array of possibilities, from major responses at one particular gene to diffuse polygenic reponses (\cite{Slate_2005,Walsh_Lynch_2018}). {\color{Black}However, despite the \bl{boom in} of genome sequencing of the last four decades, global conclusions on the conditions leading to a major gene or a polygenic response to local adaptation are yet to be drawn from empirical studies.} For example, as reviewed in \textcite{Walsh_Lynch_2018}, different conclusions on the genetic basis of the evolution of resistance to the insecticide BT toxin have emerged between field and lab experiments. Indeed, in the field, major-effects are more often found to be the main drivers of evolution of resistance, whereas a polygenic response is more commonly found in the lab (\cite{MCKENZIE1994166}), even if intensity of selection might not differ (\cite{Groeters_Tabashnik_2000}). {\color{Black} In more recent studies, divergent conclusions about the genetic basis of pathogen resistance in cattle have been reached in different regions of the world (major-effect in Australia: \cite{Turner_Harrison_Bunch_Neto_Li_Barendse_2010}, polygenic in the tropics: \cite{Porto-Neto_Reverter_Prayaga_Chan_Johnston_Hawken_Fordyce_Garcia_Sonstegard_Bolormaa_2014}). Other empirical studies also highlight cases where the genetic basis of local adaptation has a large variance in effect size, thus combining major and polygenic responses (see e.g. \bl{\cite{Koch_Ravinet_Westram_Johannesson_Butlin_2022} about the genetic architecture of local adaptation in \emph{Littorina saxatilis} and} \cite{Gagnaire_Gaggiotti_2016} for a review for marine species). We are therefore interested in investigating the following biological question: \bl{What are the stability conditions of} either major gene responses or polygenic responses (with either a small or large variance in effect size) underlying species' evolution in patchy environments?}

{\color{Black}From a theoretical point of view, the genetic basis of adaptation has been the subject of an ongoing debate since the early days of evolutionary biology. On the one hand, the field of population genetics explicitly describes and models the dynamics of a few major genes and alleles that have discrete \bl{Mendelian} effects, like eye color. On the other hand, the quantitative genetic field explores the evolution of quantitative and continuous traits, like limb size, which are thought to arise from the combined small effects of many genes. A first theoretical milestone in the relationship between the two fields was reached in 1919, when Fisher proposed the infinitesimal model to formalize how such a polygenic trait can be inherited, using the Mendelian framework, clarifying the connection between the two genetic approaches (\cite{Fisher_1919}). His framework was subsequently made more precise (\cite{Bulmer_1971,Lange_1978}) and recently justified in various situations using a multi-loci model and a central limit theorem approach (\cite{Barton_Etheridge_Veber_2017}). This debate on the genetic basis of adaptation can be illustrated by the tension between the textbook prediction of \textcite{Orr_1998} of an exponential distribution of allelic effect sizes \bl{following} adaptation \bl{in} a homogeneous environment and the review in \textcite{Rockman_2012}, which \bl{presents} several lines of evidence highlighting infinitesimal polygenic basis of quantitative traits. Here, we would like to revisit the classical prediction of an exponential distribution in allelic effects from \textcite{Orr_1998}, not in the context of weak selection in panmictic populations, but rather in the context of \emph{spatial heterogeneity}. Although this has been explored through individual-based simulation studies (see \cite{Yeaman_Whitlock_2011, Yeaman_2022}), we aim at providing analytical predictions \bl{that} can yield mechanistic insights on the distribution of effect sizes \bl{likely to be observed following adaptation in} patchy environments.}

{\color{Black} \paragraph{Aims.} The adaptation of species to heterogeneous environment \bl{at a small number of loci} has been extensively studied in the population genetic field (see \cite{Nagylaki_Lou_2001, Burger_Akerman_2011} for one or two-locus models, \cite{Yeaman_Otto_2011} for a model including the effect of drift, \cite{Geroldinger_Bürger_2014} for a two-deme two-locus model). In particular, we would like to draw attention \bl{to} the predictions from the \textit{simplest one-locus model describing the dynamics of local adaptation of a haploid species to a symmetrical two-deme environment}. In the case-study where two alleles segregate at a single locus, each allele being favoured in one deme and selected against in the other, it can be shown that \textit{polymorphism is always maintained at this locus, independently of the migration \bl{rate or selection strength}} (unless the population goes extinct - see a proof of this result in \ref{prop:stability_S0}). However, it is not clear whether this polymorphism would similarly be maintained if, in addition to this biallelic major-effect locus, local adaptation was also influenced by very small contributions from a large number of unlinked loci. The main aim of this paper is therefore to answer the following question: \textit{Could a polygenic background constituted by very small allelic effects topple the polymorphism at the major-effect locus, even though the latter is a priori beneficial for local adaptation when considered on its own?} 

From the point of view of population genetics, answering this question in heterogeneous environments would require \bl{the analysis of} models whose complexity would quickly grow as the number of small effect loci considered increases (however, note that multi-loci models in heterogeneous environments exist, but either focus on the case where all the alleles have equal effects - see \cite{Lythgoe_1997,Szep_sachedva_barton_2021} - or on panmictic populations - see \cite{Vladar_Barton_2014,Jain_Stephan_2017,Höllinger_Pennings_Hermisson_2019}). In this work, \textit{we propose to circumvent this limitation with a new eco-evo model and methodology. It merges the point of views of population genetics and quantitative genetics and considers the combined contributions of a quantitative background} (summarizing the polygenic background' small effects contributions) \textit{and a major-effect locus} on the focal trait determining local adaptation (note that the latter is typically not considered in quantitative genetic models; see \cite{Ronce_Kirkpatrick_2001,Hendry_Day_Taylor_2001,Debarre_Ronce_Gandon_2013,Mirrahimi_2017,Mirrahimi_Gandon_2020,hamel_roques_lavigne,Dekens_2022}). 

This approach has the immediate benefit that each individual is only described by two variables (major-effect allele and quantitative background) instead of potentially many (for each alleles). The drawback is that how to implement efficiently the inheritance of the quantitative background becomes less obvious, which adds a methodological challenge to our objectives. One way to proceed would be to make the ad-hoc assumption that the quantitative background only adds Gaussian noise around the major-effects. This was employed in \textcite{Lande_1983} in order to investigate the genetic architecture of adaptation to a shifting environment (\bl{via} major-effect allelic sweeps or subtle shifts in the frequency of many small effect alleles). However, our proposed method aims \bl{to avoid} any prior assumption on the distribution of the quantitative background and rather analyze the \bl{distribution} that naturally emerge from the dynamics of adaptation.\bl{Instead, we focus} on the within-family distribution by extending Fisher's infinitesimal model (\cite{Fisher_1919,Bulmer_1971,Lange_1978,Bulmer_1980,Turelli_Barton_1994,Barton_Etheridge_Veber_2017}).}

\paragraph{Contributions.}
We show that our model for composite traits gives new analytical insights on the stability of polymorphism at a major-effect locus underlying local adaptation in a symmetrical heterogeneous environment in the presence of a quantitative background due to a large number of small effect loci. Due to small perturbations induced by the quantitative component of the trait, polymorphism at the major-effect locus is lost both at low and high \bl{strengths} of selection, below a certain level of migration. The first region of loss of polymorphism, at low selection intensities, is intuitively expected, as migration blends more strongly than selection differentiates. More surprising is the lost of polymorphism at high intensities of selection, where one would expect polymorphism at the major-effect locus to be strongly favoured. \bl{To} our knowledge, this phenomenon, where quantitative differences displace polymorphism at a major-effect locus, has not yet been documented. We confirm that our analysis is consistent with individual-based simulations.

This \bl{case study} suggests that the long-term influence of a quantitative polygenic background on the polymorphic equilibrium at major-effect loci can lead to unforeseen phenomena. In this work, we present an integrative framework that is meant to help analytically bridge population genetics and quantitative genetics. Our method goes deeper than previous models (\cite{Lande_1983}) by justifying in a certain regime of small variance that the traits are normally distributed around the major-effect alleles effects, thanks to new arguments of convex analysis. It allows a \bl{separation of time scales}, \bl{which} ultimately leads to the conditions for when the infinitesimal quantitative background slowly disrupts the rapidly established symmetrical polymorphism at the major-effect locus. 

Furthermore, we provide a comprehensive toolbox that describes how to apply our methodology to more general cases in terms of number of major-effect loci, number of patches, and form of selection for haploid or diploid populations (see Appendices \ref{toolbox} and \ref{app:generalization}).

{\color{Black}\section{Methods}
\subsection{Model}
\label{sec:model}
\subsubsection{From a generic quantitative genetic model to a composite model.}}
We consider a haploid population reproducing sexually and characterized by a quantitative trait $\boldsymbol{\zeta}$ in a heterogeneous environment with two habitats connected by constant migration at rate $\boldsymbol{m_{1}}$ (from habitat 1 to habitat 2) and $\boldsymbol{m_{2}}$ (from habitat 2 to habitat 1). Following classical models of quantitative genetics, we model each habitat $i$ selecting toward a different optimum $\boldsymbol{\theta_i}$ with strength $\boldsymbol{g_i}$. Maladaptation and local uniform competition for resources (with intensity $\boldsymbol{\kappa_i}$ in deme $i$) are sources of mortality leading to a per capita decline at rate:
\[-\boldsymbol{g_i}(\boldsymbol{\zeta} - \boldsymbol{\theta_i})^2- \boldsymbol{\kappa_i}\,\boldsymbol{N_i},\]
for individuals of trait $\boldsymbol{\zeta}$ in habitat $i$ ($\boldsymbol{N_i}$ denotes the local population size).
At time $\boldsymbol{t}\geq 0$, let $\boldsymbol{n_1}(\boldsymbol{t},\boldsymbol{\zeta})$ and $\boldsymbol{n_2}(\boldsymbol{t},\boldsymbol{\zeta})$ be the local trait densities in patches 1 and 2, and $\boldsymbol{\mathcal{B}}[\boldsymbol{n_i}](\boldsymbol{t},\boldsymbol{\zeta})$ the number of individuals born with a trait $\boldsymbol\zeta$ in habitat $i$, with reproduction occuring at rate $\boldsymbol{\lambda_i}$. The dynamics of the local trait densities read:

\begin{equation}
\label{systnonstat}
\begin{aligned}
\begin{cases}
\frac{\partial \boldsymbol{n_{1}}}{\partial \boldsymbol{t}}(\boldsymbol{t},\boldsymbol{\zeta}) = \boldsymbol{\lambda_1}\,\boldsymbol{\mathcal{B}}[\boldsymbol{n_{1}}](\boldsymbol{t},\boldsymbol{\zeta}) - \boldsymbol{g_1}\,(\boldsymbol{\zeta}-\boldsymbol{\theta_1})^2\,\boldsymbol{n_{1}}(\boldsymbol{t},\boldsymbol{\zeta}) - \boldsymbol{\kappa_1}\boldsymbol{N_{1}}(\boldsymbol{t})\,\boldsymbol{n_{1}}(\boldsymbol{t},\boldsymbol{\zeta})\\\qquad\qquad\qquad\qquad\qquad\qquad\qquad\qquad\qquad\qquad\qquad\qquad\qquad\;+\boldsymbol{m_{2}}\,\boldsymbol{n_{2}}(\boldsymbol{t},\boldsymbol{\zeta})-\boldsymbol{m_{1}}\,\boldsymbol{n_{1}}(\boldsymbol{t},\boldsymbol{\zeta}), \\
\\
\frac{\partial \boldsymbol{n_{2}}}{\partial \boldsymbol{t}}(\boldsymbol{t},\boldsymbol{\zeta}) = \boldsymbol{\lambda_2}\,\boldsymbol{\mathcal{B}}[\boldsymbol{n_{2}}](\boldsymbol{t},\boldsymbol{\zeta}) - \boldsymbol{g_2}\,(\boldsymbol{\zeta}-\boldsymbol{\theta_2})^2\,\boldsymbol{n_{2}}(\boldsymbol{t},\boldsymbol{\zeta}) - \boldsymbol{\kappa_2} \,\boldsymbol{N_{2}}(\boldsymbol{t})\,\boldsymbol{n_{2}}(\boldsymbol{t},\boldsymbol{\zeta})\\\qquad\qquad\qquad\qquad\qquad\qquad\qquad\qquad\qquad\qquad\qquad\qquad\qquad\;+\boldsymbol{m_{1}}\,\boldsymbol{n_{1}}(\boldsymbol{t},\boldsymbol{\zeta})-\boldsymbol{m_{2}}\,\boldsymbol{n_{2}}(\boldsymbol{t},\boldsymbol{\zeta}).\end{cases}
\end{aligned}
\end{equation}

We can define the trait axis such that: $\boldsymbol{\theta} := \boldsymbol{\theta_2} = - \boldsymbol{\theta_1}>0$ without loss of generality. We next describe the novel aspect of this work, which allows the trait $\boldsymbol{\zeta}$ to be \bl{the} sum of two components, a major-effect locus and a quantitative background $\boldsymbol{z}$. We furthermore describe the sexual reproduction operator used.

\paragraph{major-effect.} The first component comes from a locus where two alleles $A/a$ are segregating. They have a major-effect on the trait: $\boldsymbol{\eta_A}$ and $\boldsymbol{\eta_a}$. Inheritance of this locus is Mendelian.

\paragraph{Quantitative background.} The second component, denoted by $\boldsymbol{z}\in\R$, represents the quantitative background due to infinitesimally small additive contributions to the trait from a large number of unlinked alleles. Although it comes from infinitesimally small contributions, $\boldsymbol{z}$ should not be thought \bl{of} as being necessarily small, due to the large number of alleles contributing to it. We also assume that the major-effect locus is effectively unlinked with the small-effect ones.

\paragraph{Inheritance of the trait: an extension of the infinitesimal model.}

Let us recall that the infinitesimal model, first introduced in \textcite{Fisher_1919}, provides a way to encode efficiently the inheritance of complex traits coming from a large number of alleles, each with small effects. The classical version states that an offspring receives a trait $\mathcal{Z}$ from its parents with traits $\mathcal{Z}_1$ and $\mathcal{Z}_2$, where $\mathcal{Z}$ differs from the mean parental trait $\frac{\mathcal{Z}_1+\mathcal{Z}_2}{2}$ following a centered Gaussian law, with variance $\frac{\sigma^2}{2}$. The latter accounts for the stochasticity of segregation, and therefore the variance is called the segregational variance. Specifically:

\begin{equation*}
 \mathcal{Z}|\mathcal{Z}_1,\mathcal{Z}_2 \sim \frac{\mathcal{Z}_1+\mathcal{Z}_{2}}{2} + \mathcal{Y}, \quad \mathcal{Y}\sim \mathcal{N}\left(0,\frac{\sigma^2}{2}\right),\;\mathcal{Y}\perp \mathcal{Z}_1,\mathcal{Z}_2.
    \label{eq:infmodel}
\end{equation*}

The Mendelian view of the infinitesimal model has been discussed in \textcite{Fisher_1919,Bulmer_1971,Lange_1978}: the common interpretation is that the trait results from a large number of small additive contributions at unlinked loci. 
For a more in depth description, see \textcite{Barton_Etheridge_Veber_2017}. 

Because the trait we are considering is a composite of a major-effect locus inherited according to Mendelian laws and an infinitesimal background, it is natural to use an extension of the infinitesimal model for this composite case. Now, the offspring's trait $(\mathcal{A},\mathcal{Z})$ given their parents $(\mathcal{A}_1, \mathcal{Z}_1)$ and $(\mathcal{A}_2, \mathcal{Z}_2)$ reads:
\begin{equation}
 (\mathcal{A},\mathcal{Z},)\;|\;(\mathcal{A}_1,\mathcal{Z}_1),(\mathcal{A}_2,\mathcal{Z}_2) \sim \left(X\mathcal{A}_1+(1-X)\mathcal{A}_2,\;\frac{\mathcal{Z}_1+\mathcal{Z}_2}{2} +  \mathcal{Y}\right),
    \label{infmodel_ext_intro}
\end{equation}
where $\mathcal{Y} \sim \mathcal{N}\left(0,\frac{\boldsymbol\sigma^2}{2}\right)$ follows a centered Gaussian law of variance $\frac{\boldsymbol\sigma^2}{2}$ and $X\sim B\left(\frac{1}{2}\right)$ follows a Bernoulli law with parameter $\frac{1}{2}$ (assuming fair meiosis). The \bl{random variables} are independent of each other and of $\mathcal{Z}_1,\mathcal{Z}_2,\mathcal{A}_1,\mathcal{A}_2$.

{\color{Black}\paragraph{Modified reproduction operator.}}
Let us translate \cref{infmodel_ext_intro} into a continuous density model. Let $\boldsymbol{n^A_i(z)}$ (respectively $\boldsymbol{n^a_i(z)}$) denote the density of individuals of patch $i$ carrying allele $A$ (respectively $a$) along with an infinitesimal background $\boldsymbol{z}$, therefore having a trait $\boldsymbol{\zeta} = \boldsymbol{\eta^A + z}$ (respectively, $\boldsymbol{\eta^a+z}$). In agreement with \cref{infmodel_ext_intro}, the number of offspring born with the allele $A$ and an infinitesimal contribution $\boldsymbol{z}$ in habitat $i$ then reads:
\begin{align*}
    \boldsymbol{\mathcal{B}^A_\sigma[n^A_i,n^a_i](z)} &= \int_{\R^2}\frac{1}{\sqrt{\pi}\boldsymbol{\sigma}}\exp\left[-\frac{\left(\boldsymbol{z}-\frac{\boldsymbol{z_1}+\boldsymbol{z_2}}{2}\right)^2}{\boldsymbol{\sigma}^2}\right]\times\\
    &\qquad\frac{1}{N_i}\left[\boldsymbol{n^A_i(z_1)}\,{\boldsymbol{n^A_i(z_2)} + \frac{1}{2}\left[\boldsymbol{n^A_i(z_1)}\,\boldsymbol{n^a_i(z_2)}+\boldsymbol{n^a_i(z_1)}\boldsymbol{n^A_i(z_2)}\,\right]}\right]d\boldsymbol{z_1}\,d\boldsymbol{z_2}\\
    &= \int_{\R^2}\frac{1}{\sqrt{\pi}\boldsymbol{\sigma}}\exp\left[-\frac{\left(\boldsymbol{z}-\frac{\boldsymbol{z_1}+\boldsymbol{z_2}}{2}\right)^2}{\boldsymbol{\sigma}^2}\right]\,\boldsymbol{n^A_i(z_1)}\,\frac{\boldsymbol{n^A_i(z_2)} + \boldsymbol{n^a_i(z_2})}{\boldsymbol{N_i}}\,d\boldsymbol{z_1}\,d\boldsymbol{z_2}.
\end{align*}
Similarly, the corresponding number of offspring born with the allele $a$ and an infinitesimal part $\boldsymbol{z}$ reads: 
\begin{align*}
    \boldsymbol{\mathcal{B}^a_\sigma[n^A_i,n^a_i](z)}&= \int_{\R^2}\frac{1}{\sqrt{\pi}\boldsymbol{\sigma}}\exp\left[-\frac{\left(\boldsymbol{z}-\frac{\boldsymbol{z_1}+\boldsymbol{z_2}}{2}\right)^2}{\boldsymbol{\sigma}^2}\right]\times\\
    &\qquad\frac{1}{N_i}\left[\boldsymbol{n^a_i(z_1)}\,{\boldsymbol{n^a_i(z_2)} + \frac{1}{2}\left[\boldsymbol{n^A_i(z_1)}\,\boldsymbol{n^a_i(z_2)}+\boldsymbol{n^a_i(z_1)}\boldsymbol{n^A_i(z_2)}\,\right]}\right]d\boldsymbol{z_1}\,d\boldsymbol{z_2}\\
    &= \int_{\R^2}\frac{1}{\sqrt{\pi}\boldsymbol{\sigma}}\exp\left[-\frac{\left(\boldsymbol{z}-\frac{\boldsymbol{z_1}+\boldsymbol{z_2}}{2}\right)^2}{\boldsymbol{\sigma}^2}\right]\,\boldsymbol{n^a_i(z_1)}\,\frac{\boldsymbol{n^a_i(z_2)} + \boldsymbol{n^A_i(z_2})}{\boldsymbol{N_i}}\,d\boldsymbol{z_1}\,d\boldsymbol{z_2}.
\end{align*}
The operator reproduction $\boldsymbol{\mathcal{B}_\sigma}$ indicates that it is more relevant to model the dynamics of the two local allelic densities $\boldsymbol{n^a_i},\boldsymbol{n^A_i}$, instead of $\boldsymbol{n_i}$ (which is their sum). From now on, we will therefore adopt this point of view.

\bl{\begin{rem}[Bridging a population genetic model and a quantitative genetic model]
    Our model described above bridges the following population genetic and quantitative genetic models: 
\begin{enumerate}
    \item \ul{The one-locus haploid model in a two-patch environment}, which considers two alleles $A$ and $a$ segregating at the same locus, each improving the survival chance in one of the habitats and being deleterious in the other. We recall that with symmetrical migration and selection, \ul{this model predicts that polymorphism at the focal locus is always stable}, whenever the metapopulation persists (see \cref{rem:OLHM_fast_eq} and \ref{prop:stability_S0} for a proof of this fact).
    \item \ul{The quantitative genetic model} from \textcite{Dekens_2022}, which studies the eco-evo dynamics of a quantitative trait in a heterogeneous environment, where the trait is inherited according to the standard version of the infinitesimal model. Our work can be seen as an extension of this model, to which we add the segregation of two major-effect alleles at a single locus. Moreover, one can notice that if one major-effect allele fixes (loss of polymorphism), the two models are equivalent. \bl{Because} \textcite{Dekens_2022} gives a complete analytical description of the outcomes of their system (in the small segregation variance regime), the outcomes for our present study are known given the fixation of a major-effect allele. Therefore, \ul{our study focuses on the description of polymorphism at the major-effect locus and its stability}.
\end{enumerate}
\label{rem:fixation_alleles_system}
\end{rem} }

{\color{Black}\subsubsection{Dimensionless system.}}
Let us rescale \cref{systnonstat} according to:
\[\eta^A := \frac{\boldsymbol{\eta^A}}{\boldsymbol{\theta}},\quad\! z:=\frac{\boldsymbol{z}}{\boldsymbol{\theta}},\quad\! g_i:=\frac{\boldsymbol{g_i}\boldsymbol{\theta}^2}{\boldsymbol{\lambda_1}},\quad\! m_i:=\frac{\boldsymbol{m_i}}{\boldsymbol{\lambda_1}},\quad\!\varepsilon := \frac{\boldsymbol{\sigma}}{\boldsymbol{\theta}},\quad\! t:=\varepsilon^2\boldsymbol{\lambda_1}\boldsymbol{t},,\quad\! \alpha := \frac{\boldsymbol{\kappa_1}}{\boldsymbol{\kappa_2}},\quad\! \lambda:=\frac{\boldsymbol{\lambda_2}}{\boldsymbol{\lambda_1}},\]
and introduce the rescaled trait densities:

\[n^A_{\varepsilon,i}(t,z):=\frac{\boldsymbol{\kappa_i}}{\boldsymbol{\lambda_1}}\,\boldsymbol{n^A_i(\boldsymbol{t},\boldsymbol{z})},\quad n^a_{\varepsilon,i}(t,z):=\frac{\boldsymbol{\kappa_i}}{\boldsymbol{\lambda_1}}\,\boldsymbol{n^a_i(\boldsymbol{t},\boldsymbol{z})}.\]
so that \cref{systnonstat} reads:

\begin{equation}
\label{systscaled}
\begin{aligned}
\begin{cases}
\varepsilon^2\frac{\partial n^A_{\varepsilon,1}}{\partial t}(t,z) = \mathcal{B_\varepsilon}^A(n^A_{\varepsilon,1},n^a_{\varepsilon,1})(t,z) - g_1(z+\eta^A+1)^2\,n^A_{\varepsilon,1}(t,z) -N_{\varepsilon,1}(t)\,n^A_{\varepsilon,1}(t,z)\\
\qquad\qquad\qquad\qquad\qquad\qquad\qquad\qquad\qquad\qquad\qquad\qquad+\alpha\, m_2\,n^A_{\varepsilon,2}(t,z)-m_1\,n^A_{\varepsilon_1}(t,z),  \\
\varepsilon^2\frac{\partial n^a_{\varepsilon,1}}{\partial t}(t,z) = \mathcal{B_\varepsilon}^a(n^a_{\varepsilon,1},n^A_{\varepsilon,1})(t,z) - g_1(z+\eta^a+1)^2\,n^a_{\varepsilon,1}(t,z) -N_{\varepsilon,1}(t)\,n^a_{\varepsilon,1}(t,z)\\\qquad\qquad\qquad\qquad\qquad\qquad\qquad\qquad\qquad\qquad\qquad\qquad+\alpha \,m_2\, n^a_{\varepsilon,2}(t,z)-m_1\,n^a_{\varepsilon_1}(t,z), 
\\
\varepsilon^2\frac{\partial n^A_{\varepsilon,2}}{\partial t}(t,z) = \lambda\,\mathcal{B_\varepsilon}^A(n^A_{\varepsilon,2},n^a_{\varepsilon,2})(t,z) - g_2(z+\eta^A-1)^2\,n^A_{\varepsilon,2}(t,z) -N_{\varepsilon,2}(t)\,n^A_{\varepsilon,2}(t,z)\\\qquad\qquad\qquad\qquad\qquad\qquad\qquad\qquad\qquad\qquad\qquad\qquad+\frac{m_1}{\alpha}\,n^A_{\varepsilon,1}(t,z)-m_2\,n^A_{\varepsilon_2}(t,z),  \\
\varepsilon^2\frac{\partial n^a_{\varepsilon,2}}{\partial t}(t,z) = \lambda\,\mathcal{B_\varepsilon}^a(n^a_{\varepsilon,2},n^A_{\varepsilon,2})(t,z) - g_2(z+\eta^a-1)^2\,n^a_{\varepsilon,2}(t,z) -N_{\varepsilon,2}(t)\,n^a_{\varepsilon,2}(t,z)\\\qquad\qquad\qquad\qquad\qquad\qquad\qquad\qquad\qquad\qquad\qquad\qquad+\frac{m_1}{\alpha}\,n^a_{\varepsilon,1}(t,z)-m_2\,n^a_{\varepsilon_2}(t,z), \end{cases}
\end{aligned}
\end{equation}
where the rescaled reproduction operator is given by:

\begin{equation}
\label{operator_2}
\mathcal{B^A_\varepsilon}(n^A_{\varepsilon,i},n^a_{\varepsilon,i})(t,z) = \frac{1}{\sqrt{\pi}\varepsilon}\int_{\R^2} \exp\left[\frac{-(z-\frac{z_1+z_2}{2})^2}{\varepsilon^2}\right]n^A_{\varepsilon,i}(t,z_1)\frac{n^A_{\varepsilon,i}(t,z_2)+n^a_{\varepsilon,i}(t,z_2)}{N_{\varepsilon,i}(t)}dz_1\, dz_2.
\end{equation}

\subsection{Derivation of a moment-based system in the regime of small variance $\varepsilon^2
\ll 1$}
\label{sec:pertubative_section}
{\color{Black}In this subsection, we explain how we derive a closed moment-based ODE system on which the separation of time scale analysis will be \bl{conducted}, \bl{starting} from the PDE system \eqref{systscaled} based on the trait distributions, in the regime of small variance $\varepsilon^2
\ll 1$. \bl{To do so, we justify that the quantitative background values are approximately normally distributed among bearers of the same major-effect allele. Moreover, the mean of these quantitative background values is the same for individuals in the same patch. This implies in particular that the main driver for trait divergence within each habitat is the major-effect locus.} }

\subsubsection{Gaussian approximations of quantitative background values in the regime of small variance: a formal analysis.}
\label{subsec:formal_analysis}
We choose to place our study in a regime where the amount of diversity introduced by the segregation of the infinitesimal background at each event of reproduction is small in comparison to the difference between the habitats' optima:
\[\frac{\boldsymbol{\sigma^2}}{\boldsymbol\theta^2}\ll 1 \implies \varepsilon^2\ll 1.\]
In this regime of small variance, the trait distributions are expected {\color{Black}to converge to} Dirac masses. Our focus is to {\color{Black}give an accurate description of the distribution near this limit. To do so, we extend a small variance methodology introduced by \textcite{Diekmann_Jabin_Mischler_Perthame_2005} for asexual populations and adapted recently to sexual populations with the standard infinitesimal model \bl{(\cite{Calvez_Garnier_Patout_2019, patout2020cauchy, Garnier_2022})} and develop new convex analysis arguments. Throughout this section, the time dependency will be omitted for the sake of clarity.}

\paragraph{Presentation of the methodology.}

Almost two decades ago, \textcite{Diekmann_Jabin_Mischler_Perthame_2005} introduced a methodology to determine the dynamics of the trait values around which trait distributions get concentrated \bl{as Dirac masses} under the regime of small variance. This methodology has since been used successfully to study several evolutionary questions, initially for asexual models, where the diversity generated by mutations of small variance is modelled by a linear operator translating the distribution of mutational effects (\cite{Perthame_Barles_2008,Barles_Mirrahimi_Perthame_2009,Mirrahimi_2017, Mirrahimi_Gandon_2020}). It has recently been adapted to study sexually reproducing populations with the infinitesimal model operator in homogeneous spaces (\cite{Garnier_2022,Calvez_Garnier_Patout_2019,patout2020cauchy,Dekens_2022}). 

As the analytical crux heavily relates to the singular nature of the trait distributions $n_\varepsilon$ as Dirac masses, the method consists in defining proxies $U_\varepsilon$ from $n_\varepsilon$ through a suitable transformation so that such proxies are regular functions (by comparison to Dirac masses) and their asymptotic analysis is easier. Studying them often induces a reduction in the complexity of the system while still retaining fundamental quantitative information about the distributions, such as around which traits they are concentrated. Here, we follow quantitative genetic studies that use the infinitesimal model according to the same methodology (\cite{Garnier_2022,Calvez_Garnier_Patout_2019,patout2020cauchy,Dekens_2022}) and define the proxies $U^A_{\varepsilon,i}$ (resp. $U^a_{\varepsilon,i}$):
\begin{equation}
\label{eq:HCprel}
    n^A_{\varepsilon,i} = \frac{1}{\sqrt{2\pi}\varepsilon}e^{-\frac{U^A_{\varepsilon,i}}{\varepsilon^2}},\quad n^a_{\varepsilon,i} = \frac{1}{\sqrt{2\pi}\varepsilon}e^{-\frac{U^a_{\varepsilon,i}}{\varepsilon^2}}.
\end{equation}
{\color{Black}A helpful analogy is to take the example of a spiky Gaussian distribution with small variance $\varepsilon^2$ for $n^A_{\varepsilon,i}$. Then $U^A_{\varepsilon,i}$ is a smooth quadratic function (even when $\varepsilon \ll 1$). \Cref{figure:hopf_cole} displays an example of this kind of exponential transformation (called Hopf-Cole transformation in scalar conservation laws). A key observation to deduce the traits around which the distribution concentrates is that it does so at the minima (zero) of $U_\varepsilon$.}
\begin{figure}
    \centering
    \includegraphics[width =\textwidth]{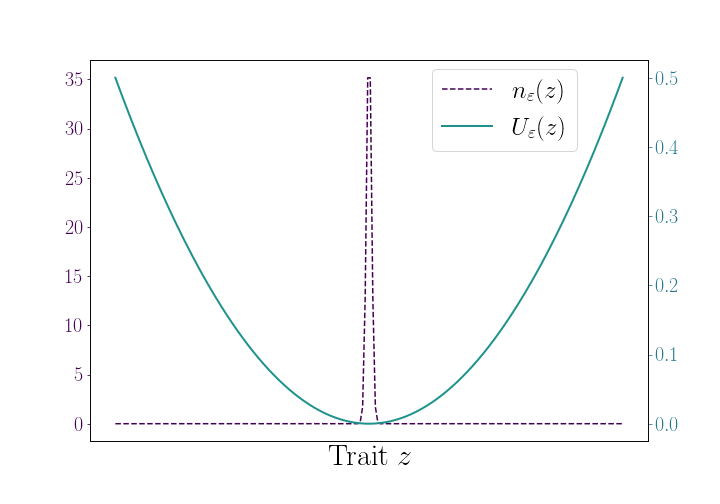}
    \caption{\textbf{Illustration of the Hopf-Cole transform to study concentration phenomena}. This transformation unfolds singular distributions $n_\varepsilon$ close to a Dirac mass (in purple), by defining more regular proxies: $U_\varepsilon$ (in green) such that $n_{\varepsilon} = \frac{1}{\sqrt{2\pi}\varepsilon}e^{-\frac{U_{\varepsilon}}{\varepsilon^2}}$. This figure suggests that, when $\varepsilon$ vanishes, the limit $U$ is regular and positive and cancels at the support of the limit measure $n$.}
    \label{figure:hopf_cole}
\end{figure}
As the proxies $U^A_\varepsilon$ and $U^a_\varepsilon$ are expected to be more regular in the regime of small variance, they are thought to be the right object on which to perform a Taylor expansion series to gain information on the asymptotic distributions in the limit of small variance (see \cite{Calvez_Garnier_Patout_2019}). We therefore define $u^A_{0,i}$ (resp. $u^a_{0,i}$) as the leading term in the Taylor expansion of $U^A_{\varepsilon,i}$ (resp. $U^a_{\varepsilon,i}$) :
\begin{equation}
\label{eq:taylor_u}
U^A_{\varepsilon,i} = u^A_{0,i} + \varepsilon^2\,u^A_{1,i}+\varepsilon^4\,v^A_{\varepsilon,i}, \qquad U^a_{\varepsilon,i} = u^a_{0,i} + \varepsilon^2\,u^a_{1,i}+\varepsilon^4\,v^a_{\varepsilon,i}
\end{equation}
where $u^A_{1,i}$ and $u^a_{1,i}$ are the next term in the Taylor expansion, and $\varepsilon^4v^A_{\varepsilon,i}$ and $\varepsilon^4v^a_{\varepsilon,i}$ are the residues. \textcite{Calvez_Garnier_Patout_2019} provides the tools to control these residues and thus rigorously justify that \eqref{eq:taylor_u} is an admissible Taylor expansion; adapting them is left for future work.
\paragraph{Characterization of the main terms $u_{0,i}^A$ and $u_{0,i}^a$ to justify Gaussian approximations.}
The first step of the analysis in the regime of small variance is the characterization of the main terms $u_{0,i}^A$ and $u_{0,i}^a$. {\color{Black} Indeed, in the regime of small variance, these have to satisfy a strong constraint \bl{that} arises naturally for the contribution of the infinitesimal model reproduction operator term to remain well-balanced within \eqref{systscaled}. In \textcite{Garnier_2022} \bl{and} \textcite{Dekens_2022}, where the standard infinitesimal model operator is used, this constraint yields the analogous main term to be quadratic, which implies that the trait distribution is approximately Gaussian, with a small variance $\varepsilon^2$. However, here, the arguments given in \textcite{Garnier_2022} and used in \textcite{Dekens_2022} are not sufficient, due to the mixing of alleles \bl{between patches} and the discrete nature of Mendelian inheritance. However, we extend the convex analysis to circumvent this limitation (\ref{prop:constraints}) and \bl{identify} $u_{0,i}^A$ and $u_{0,i}^a$ \bl{as} the same quadratic function $z\mapsto\frac{(z-z_i^*)^2}{2}$, where $z^*_i \in \R$ is to be determined later in the analysis. Assuming that \eqref{eq:taylor_u} is an admissible Taylor expansion (which is suggested by the analysis of \textcite{Calvez_Garnier_Patout_2019}), this result is crucial as it justifies the following formal Gaussian approximations of $n^A_{\varepsilon,i}$ and $n^a_{\varepsilon,i}$ ($i\in\{1,2\})$:
\begin{equation}
 n^A_{\varepsilon,i}(z) = \frac{e^{-\frac{-(z-z_i^*)^2}{2\varepsilon^2}}}{\sqrt{2\pi}\varepsilon}
 e^{-u^A_{1,i}(z)+\mathcal{O}(\varepsilon^2)},\qquad n^a_{\varepsilon,i}(z) = \frac{e^{-\frac{-(z-z_i^*)^2}{2\varepsilon^2}}}{\sqrt{2\pi}\varepsilon}e^{-u^a_{1,i}(z)+\mathcal{O}(\varepsilon^2)}.
    \label{eq:n_eps_gaussian}
\end{equation}
\ul{Hence, to the leading order, $n_{\varepsilon,i}^A$ and $n_{\varepsilon,i}^a$ are formally Gaussian, centered at the same quantitative contribution $z_i^*$, with the same variance $\varepsilon^2$.} However, they differ in the \bl{next-order}, which involves the corrector terms $u^A_{1,i}$ and $u^a_{1,i}$, which generate asymmetries in the distributions.

To support \eqref{eq:n_eps_gaussian}, we first derive the following constraints \eqref{constraint_u0} on the main terms $u^A_{0,i}$ and $u^a_{0,i}$. In order for the contribution of both reproduction operators $\mathcal{B}_\varepsilon^A$ and $\mathcal{B}_\varepsilon^a$ to remain well-balanced with the other biological phenomena in the regime of small variance in \eqref{systscaled}, $u^A_{0,i}$ and $u^a_{0,i}$ formally need to satisfy the following (see Appendix \ref{app:constraints} for the details):
\begin{equation}
\tag{C}
\begin{aligned}
\begin{cases}
\forall z \in \R, \quad 
    \max\Big[\underset{z_1,z_2}{\sup}\;u_{0,i}^A(z)&- \left(z-\frac{z_1+z_2}{2}\right)^2-u_{0,i}^A(z_1)-u_{0,i}^A(z_2),\\ &\underset{z_1,z_2}{\sup}\; u_{0,i}^A(z) - \left(z-\frac{z_1+z_2}{2}\right)^2 -u_{0,i}^A(z_1)-u_{0,i}^a(z_2) \Big] = 0,\\
    \forall z \in \R, \quad \max\Big[\underset{z_1,z_2}{\sup}\;u_{0,i}^a(z)& - \left(z-\frac{z_1+z_2}{2}\right)^2-u_{0,i}^a(z_1)-u_{0,i}^a(z_2),\\ &\underset{z_1,z_2}{\sup}\; u_{0,i}^a(z) - \left(z-\frac{z_1+z_2}{2}\right)^2 -u_{0,i}^A(z_1)-u_{0,i}^a(z_2) \Big] = 0.
\end{cases}
\end{aligned}
    \label{constraint_u0}
\end{equation}
We next state the following proposition, which characterizes the main terms $u^A_{0,i}$ and $u^a_{0,i}$ as aforementioned.
\begin{prop}
\label{prop:constraints}
Let $u^A_0$ and $u^a_0$ satisfying \cref{constraint_u0} positive almost everywhere and cancelling somewhere. Then, there exists $z^* \in \R$ such that:
\begin{equation}
\forall z \in \R, \quad 
    u_0^A(z) = u_0^a(z) = \frac{(z-z^*)^2}{2}.
\end{equation}
\end{prop}
The conditions on $u^A_{0,i}$ and $u^a_{0,i}$ in \ref{prop:constraints} (positive everywhere and cancelling somewhere) are explained in Appendix \ref{app:constraints}.  In the \cref{app:generalization}, we actually state and prove a stronger result \ref{prop:constraints_extension}, which generalizes \ref{prop:constraints} to more complex genetic architecture\bl{s}.}

Consequently, assuming that \eqref{eq:taylor_u} is the correct ansatz so that we can control the residues in \eqref{eq:n_eps_gaussian} (which the analysis of \textcite{Calvez_Garnier_Patout_2019} suggests and provides a framework to show), using the result of \ref{prop:constraints} in \eqref{eq:HCprel} and \eqref{eq:taylor_u} leads to \eqref{eq:n_eps_gaussian}.
{\color{Black}\subsubsection{Moment-based system in the regime of small variance}}

{\color{Black}This section follows directly the results of the previous one, where we showed formally that, in each habitat, the two allelic trait distributions $n^A_{\varepsilon,i}$ and $n^a_{\varepsilon,i}$ can be approximated by the same Gaussian distribution. We present here how the latter allows us to close the moment-based system obtained from integrating \eqref{systscaled}.

First, we derive formal expansions of the first moments (population size, mean trait, variance and skew) of $n_{\varepsilon,i}^A$ and $n_{\varepsilon,i}^a$ when $\varepsilon^2\ll 1$, thanks to \eqref{eq:taylor_u} and \eqref{eq:n_eps_gaussian}} (as in \cite{Dekens_2022}):
\begin{equation}
\scriptsize
\label{momentest}
\begin{aligned}
\begin{cases}
N^A_{\varepsilon,i}:= \displaystyle\int_\R n^A_{\varepsilon,i}(z)\,dz &= e^{-u^A_{1,i}(z_i^*)}{\tiny\left[1+\varepsilon^2\left(\frac{\left(\partial_z{u^A_{1,i}}(z_i^*)\right)^2}{2}-\frac{\partial_{zz}u^A_{1,i}(z_i^*)}{2}-v^A_{i,\varepsilon}(z_i^*)\right)\right]}+\mathcal{O}(\varepsilon^4),\\
\overline{z_{\varepsilon,i}^A} := \displaystyle\int_\R z\,\frac{n^A_{\varepsilon,i}(z)}{N^A_{\varepsilon,i}}\,dz &= z_i^*-\varepsilon^2\partial_zu^A_{1,i}(z_i^*)+\mathcal{O}(\varepsilon^4),\\
\left({{\sigma}_{\varepsilon,i}^A}\right)^2 := \displaystyle\int_\R (z-\overline{z^A_{\varepsilon,i}})^2\,\frac{n^A_{\varepsilon,i}(z)}{N^A_{\varepsilon,i}}\,dz &= \varepsilon^2+\mathcal{O}(\varepsilon^4),\\
\left({{\psi}_{\varepsilon,i}^A} \right)^3 := \displaystyle\int_\R (z-\overline{z^A_{\varepsilon,i}})^3\,\frac{n^A_{\varepsilon,i}(z)}{N^A_{\varepsilon,i}}\,dz &= \mathcal{O}(\varepsilon^4).
\end{cases}
\end{aligned}    
\end{equation}

{\color{Black}Using \eqref{momentest} when integrating \eqref{systscaled}, we can close the infinite system of moments in the regime of small variance, producing a system of eight ODEs governing the dynamics of the four allelic subpopulation sizes $N^a_{\varepsilon,1}, N^A_{\varepsilon,1}, N^a_{\varepsilon,2}, N^A_{\varepsilon,2}$ and the four allelic local mean quantitative traits $\overline{z_{\varepsilon,1}^a}, \overline{z_{\varepsilon,1}^A}, \overline{z_{\varepsilon,2}^a}, \overline{z_{\varepsilon,2}^A}$:}
\begin{equation}
\label{eq:syst_moment_small_var}
    \begin{aligned}
\begin{cases}
\varepsilon^2 \frac{d\,N^a_{\varepsilon,1}}{dt} = N^a_{\varepsilon, 1} - \left(N^A_{\varepsilon,1}+N^a_{\varepsilon,1}\right) N^a_{\varepsilon,1} - g_1\left[\overline{z^a_{\varepsilon,1}}+\eta^a+1\right]^2 N^a_{\varepsilon,1}+\alpha\,m_2\,N^a_{\varepsilon,2}-m_1\,N^a_{\varepsilon,1}, \\\qquad\qquad\qquad\qquad\qquad\qquad\qquad\qquad\qquad\qquad\qquad\qquad\qquad\qquad\qquad\qquad\quad +\mathcal{O}(\varepsilon^2),
\\\\
\varepsilon^2 \frac{d\,N^A_{\varepsilon,1}}{dt} = N^A_{\varepsilon,1} - \left(N^A_{\varepsilon,1}+N^a_{\varepsilon,1}\right) N^A_{\varepsilon,1} - g_1\left[\overline{z^A_{\varepsilon,1}}+\eta^A+1\right]^2 N^A_{\varepsilon,1} +\alpha\,m_2\,N^A_{\varepsilon,2}-m_1\,N^A_{\varepsilon,1}\\\qquad\qquad\qquad\qquad\qquad\qquad\qquad\qquad\qquad\qquad\qquad\qquad\qquad\qquad\qquad\qquad\quad +\mathcal{O}(\varepsilon^2),
\\\\
\varepsilon^2 \frac{d\,N^a_{\varepsilon,2}}{dt} = \lambda\,N^a_{\varepsilon,2} - \left(N^A_{\varepsilon,2}+N^a_{\varepsilon,2}\right) N^a_{\varepsilon,2} - g_2\left[\overline{z^a_{\varepsilon,2}}+\eta^a-1\right]^2 N^a_{\varepsilon,2}+\frac{m_1}{\alpha}\,N^a_{\varepsilon,1}-m_2\,N^a_{\varepsilon,2} \\\qquad\qquad\qquad\qquad\qquad\qquad\qquad\qquad\qquad\qquad\qquad\qquad\qquad\qquad\qquad\qquad\quad  +\mathcal{O}(\varepsilon^2),
\\\\
\varepsilon^2 \frac{d\,N^A_{\varepsilon,2}}{dt} = \lambda\,N^A_{\varepsilon,2} - \left(N^A_{\varepsilon,2}+N^a_{\varepsilon,2}\right) N^A_{\varepsilon,2} - g_2\left[\overline{z^A_{\varepsilon,2}}+\eta^A-1\right]^2 N^A_{\varepsilon,2} +\frac{m_1}{\alpha}\,N^a_{\varepsilon,1}-m_2\,N^a_{\varepsilon,2}\\\qquad\qquad\qquad\qquad\qquad\qquad\qquad\qquad\qquad\qquad\qquad\qquad\qquad\qquad\qquad\qquad\quad +\mathcal{O}(\varepsilon^2),
\\\\
\varepsilon^2 \frac{d\,\overline{z_{\varepsilon,1}^a}}{dt} = \varepsilon^22g_1\left[-1-\eta^a-\overline{z^a_{\varepsilon,1}}\right] +\left(\frac{\overline{z^A_{\varepsilon,1}}-\overline{z^a_{\varepsilon,1}}}{2}\right)\frac{N^A_{\varepsilon,1}}{N_{\varepsilon,1}}+\alpha\,m_2\frac{N^a_{\varepsilon,2}}{N^a_{\varepsilon,1}}\left(\overline{z^a_{\varepsilon,2}}-\overline{z^a_{\varepsilon,1}}\right)+\mathcal{O}(\varepsilon^4),\\
\\
\varepsilon^2 \frac{d\,\overline{z_{\varepsilon,1}^A}}{dt} = \varepsilon^22g_1\left[-1-\eta^A-\overline{z^A_{\varepsilon,1}}\right] +\left(\frac{\overline{z^a_{\varepsilon,1}}-\overline{z^A_{\varepsilon,1}}}{2}\right)\frac{N^a_{\varepsilon,1}}{N_{\varepsilon,1}}+\alpha\,m_2\frac{N^A_{\varepsilon,2}}{N^A_{\varepsilon,1}}\left(\overline{z^A_{\varepsilon,2}}-\overline{z^A_{\varepsilon,1}}\right)
+\mathcal{O}(\varepsilon^4),\\
\\
\varepsilon^2 \frac{d\,\overline{z_{\varepsilon,2}^a}}{dt} = \varepsilon^22g_2\left[1-\eta^a-\overline{z^a_{\varepsilon,2}}\right] +\left(\frac{\overline{z^A_{\varepsilon,2}}-\overline{z^a_{\varepsilon,2}}}{2}\right)\frac{N^A_{\varepsilon,2}}{N_{\varepsilon,2}}+\frac{m_1}{\alpha}\frac{N^a_{\varepsilon,1}}{N^a_{\varepsilon,2}}\left(\overline{z^a_{\varepsilon,1}}-\overline{z^a_{\varepsilon,2}}\right)+\mathcal{O}(\varepsilon^4),\\
\\
\varepsilon^2 \frac{d\,\overline{z_{\varepsilon,2}^A}}{dt} = \varepsilon^22g_2\left[1-\eta^A-\overline{z^A_{\varepsilon,2}}\right] +\left(\frac{\overline{z^a_{\varepsilon,2}}-\overline{z^A_{\varepsilon,2}}}{2}\right)\frac{N^a_{\varepsilon,2}}{N_{\varepsilon,2}}+\frac{m_1}{\alpha}\frac{N^A_{\varepsilon,1}}{N^A_{\varepsilon,2}}\left(\overline{z^A_{\varepsilon,1}}-\overline{z^A_{\varepsilon,2}}\right)
+\mathcal{O}(\varepsilon^4).
\end{cases}
\end{aligned}
\end{equation}
{\color{Black}\paragraph{Biological description of the equations of the moment-based system \eqref{eq:syst_moment_small_var}.}
The first four equations encoding the dynamics of the allelic subpopulations sizes involve four terms, that we describe using the first equation \bl{for} $N^a_{\varepsilon, 1}$. The first term $N^a_{\varepsilon, 1}$ is a growth term, the second one $-(N^A_{\varepsilon, 1}+N^a_{\varepsilon, 1})N^a_{\varepsilon, 1}$ is a non-linear negative death term by competition, proportional to the total subpopulation size. The third one $-g_1\left[\overline{z^a_{\varepsilon,1}}+\eta^a+1\right]^2 N^a_{\varepsilon,1}$ is a negative death term by selection (with strength $g_1$), which is more lethal when the allelic local mean trait $\overline{z^a_{\varepsilon, 1}}+\eta^a$ is far from the local optimum $-1$. The last migration term $\alpha\,m_2\,N^a_{\varepsilon,2}-m_1\,N^a_{\varepsilon,1}$ represents the asymmetrical transfer of populations between the two patches.

The last four equations encoding the dynamics of the allelic local mean quantitative traits involve three different terms that we describe, taking for reference the first of these equations \bl{for} $\overline{z^a_{\varepsilon, 1}}$. The first term is the selection gradient that pushes the total mean trait $\overline{z^a_{\varepsilon, 1}}+\eta^a$ towards the local optimum $-1$, with an intensity $2\varepsilon^2 g_1$, proportional to the intensity of selection $g_i$ and the small variance of the quantitative trait $\varepsilon^2$ (in agreement with the Gaussian approximation \eqref{eq:n_eps_gaussian}). The second term $\left(\frac{\overline{z^A_{\varepsilon,1}}-\overline{z^a_{\varepsilon,1}}}{2}\right)\frac{N^A_{\varepsilon,1}}{N_{\varepsilon,1}}$ does not exist in the analogous moment-based system in \textcite{Dekens_2022} (without the major-effect locus), as it originates from the segregation of $A/a$ at the major-effect locus. It describes a force which pushes each allelic mean quantitative component towards one another within the same habitat \bl{due to} the mixing effect of the infinitesimal model. It is consistent with the result provided by \ref{prop:constraints} and the Gaussian approximations \eqref{eq:n_eps_gaussian}, which are centered at the same quantitative component $z^*_i$, close to both $\overline{z^A_{\varepsilon,i}}$ and $\overline{z^a_{\varepsilon,i}}$ according to the second line of the expansions \eqref{momentest}. The last term $\alpha\,m_2\frac{N^a_{\varepsilon,2}}{N^a_{\varepsilon,1}}\left(\overline{z^a_{\varepsilon,2}}-\overline{z^a_{\varepsilon,1}}\right)$ relates to the effect of the transfer of population by migration onto the mean quantitative component: it pushes the local mean quantitative components corresponding to the same major-effect allele $\overline{z^a_{\varepsilon,2}}$ and $\overline{z^a_{\varepsilon,1}}$ towards one another.
\begin{rem}[Selection shifts the allelic local mean quantitative trait slowly]
\label{rem:slow_selection}
In the last four equations of \eqref{eq:syst_moment_small_var}, there is a noticeable difference between the first term, proportional to $\varepsilon^2$, and the other two terms, which are of order 1. This \bl{demonstrates} the fact that, in the regime of small variance, selection shifts the local mean quantitative traits very slowly toward the local optima compared to how fast the other two terms intervene in the equation \bl{(describing selection on the major-effects alleles and migration)}. Notice also that the time scale in which the differential system \eqref{eq:syst_moment_small_var} is written ($\varepsilon^2\frac{d\cdot}{dt}$) is the correct one to capture this slow shift.
\end{rem}
\begin{rem}[Magnitude of the residues in \eqref{eq:syst_moment_small_var}]
In the system \eqref{eq:syst_moment_small_var}, the difference in the system between the residues in the first four equations on the local sizes of population of order $\mathcal{O}(\varepsilon^2)$ and the ones in the last four equations on the mean quantitative components of order $\mathcal{O}(\varepsilon^4)$ is consistent with the analysis of \textcite{patout2020cauchy} (see in particular Theorem 1.4).
\end{rem}}

\subsection{Separation of time scales: slow-fast analysis}
\label{sec:slow_fast}
{\color{Black}As highlighted by \cref{rem:slow_selection}, the shift of allelic local mean quantitative components $\overline{z^A_{\varepsilon, i}}$ and $\overline{z^a_{\varepsilon, i}}$ occurs on a slower time scale \bl{than} growth, death and transfer of populations for the allelic subpopulation sizes (first four equations of \eqref{eq:syst_moment_small_var}) and \bl{than} the two relaxing forces of gene flow and segregation for the allelic local mean quantitative traits (last two terms of the last four equations of \eqref{eq:syst_moment_small_var}). Therefore, in this subsection, we show that the moment-based system \eqref{eq:syst_moment_small_var} has a particular structure (up to a change in variables) \bl{that} allows the possibility to separate two different time scales, which can be interpreted as fast ecological time scale (including selection on the major-effects locus) and slow quantitative evolutionary time scales. 

First, we need to transform \eqref{eq:syst_moment_small_var} into an equivalent system which has a suitable form to prove the separation of time scales. This requires the following change of variables, which is motivated by the formal analysis of \cref{sec:pertubative_section} (especially the results of \ref{prop:constraints}):}
\begin{equation}
    \label{eq:slow_fast_variables}\small{
    \delta^a_\varepsilon = \frac{\overline{z_{\varepsilon,2}^a}-\overline{z_{\varepsilon,1}^a}}{2\varepsilon^2},\quad \delta^A_\varepsilon = \frac{\overline{z_{\varepsilon,2}^A}-\overline{z_{\varepsilon,1}^A}}{2\varepsilon^2},\quad \delta_\varepsilon = \frac{\overline{z_{\varepsilon,1}^A}+\overline{z_{\varepsilon,2}^A}-\overline{z_{\varepsilon,1}^a}-\overline{z_{\varepsilon,2}^a}}{4\varepsilon^2}, \quad Z_\varepsilon = \frac{\overline{z_{\varepsilon,1}^A}+\overline{z_{\varepsilon,2}^A}+\overline{z_{\varepsilon,1}^a}+\overline{z_{\varepsilon,2}^a}}{4}}.
\end{equation}
$Z_\varepsilon$ can be interpreted as the mean infinitesimal part of the metapopulation, $\delta_\varepsilon$ the spatial average of the local difference between the two allelic mean infinitesimal parts, $\delta^A_\varepsilon$ and $\delta^a_\varepsilon$ \bl{the equivalent term among bearers of $A$ and $a$, respectively} (see an illustration of those new variables in \cref{fig:slow_fast_variables}). {\color{Black}The quantities defining $\delta_\varepsilon, \delta^A_\varepsilon, \delta^a_\varepsilon$ are divided by $\varepsilon^2$ because \cref{rem:slow_selection}
 suggests that $\overline{z^A_{\varepsilon, 1}}$, $\overline{z^a_{\varepsilon, 1}}$, $\overline{z^A_{\varepsilon, 2}}$ and $\overline{z^a_{\varepsilon, 2}}$ all relax quickly towards the same value due to the fast action of gene flow and segregation, with an error of order $\varepsilon^2$.}
\begin{figure}
   \begin{tikzpicture}
   \node[] {\includegraphics[width =\linewidth]{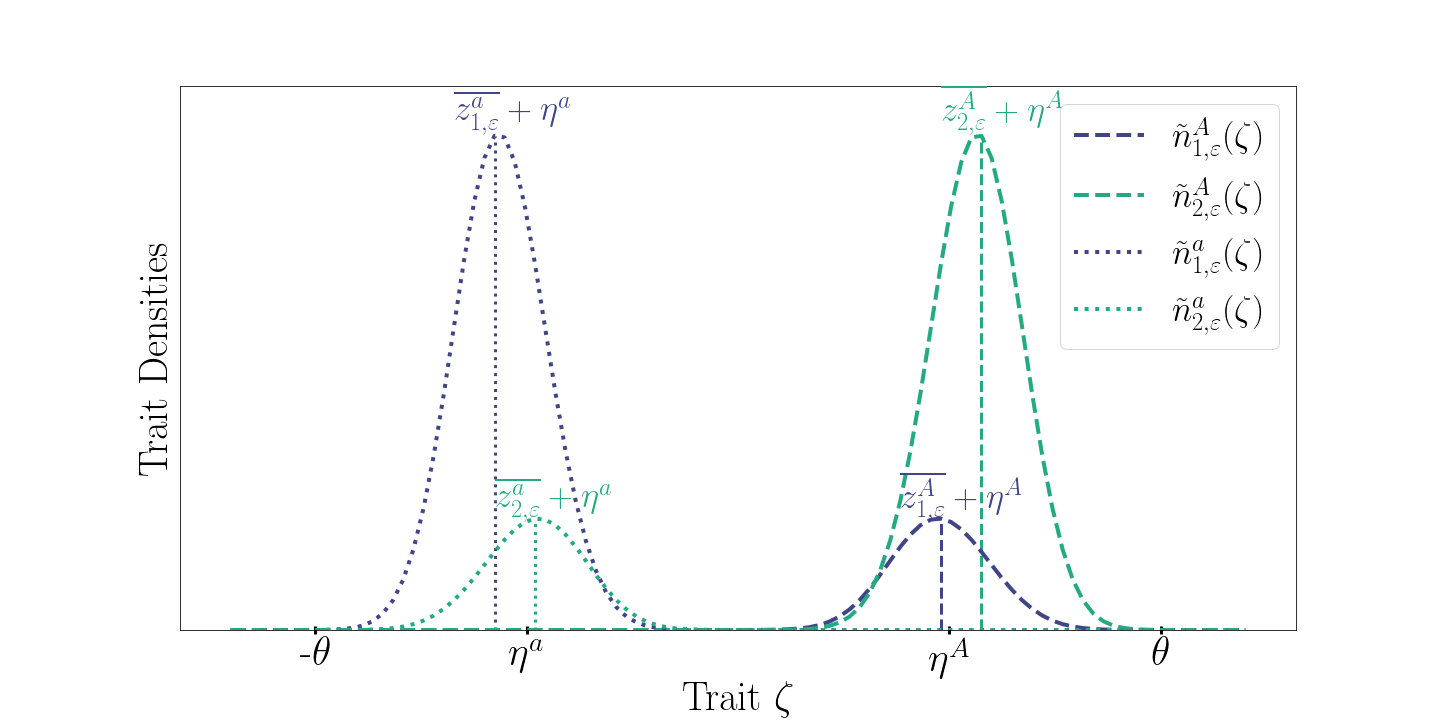}};
   
   \draw[-, color = BrickRed] (-2.4, -4.1) -- (-2.4, -3.9);
   \draw[->, color = BrickRed] (-2.4, -4) -- (-2, -4);
   \node[color = BrickRed] at (-2.2, -4.5) {$\varepsilon^2\,\delta^a_\varepsilon$};
   
   \draw[-, color = BrickRed] (2.45, -4.1) -- (2.45, -3.9);
   \draw[->, color = BrickRed] (2.45, -4) -- (2.85, -4);
   \node[color = BrickRed] at (2.65, -4.5) {$\varepsilon^2\,\delta^A_\varepsilon$};
   
   \draw[-, color = BrickRed] (-2.2, -5.1) -- (-2.2, -4.9);
   \draw[->, color = BrickRed] (-2.2, -5) -- (2.65, -5);
   \node[color = BrickRed] at (0.1, -5.5) {$\varepsilon^2\,\delta_\varepsilon+2\eta$};
   
   \draw[-, color = BrickRed] (0.225, -3.05) -- (0.225, -2.85);
   \node[color = BrickRed] at (0.225, -2.5) {$Z_\varepsilon$};
   \end{tikzpicture}
   \caption{\textbf{Illustration of the slow-fast variables $Z_\varepsilon$, $\delta_\varepsilon$, $\delta^A_\varepsilon$ and $\delta^a_\varepsilon$ (in red)}, introduced in \eqref{eq:slow_fast_variables}. This figure displays a situation where \bl{the two} major-effect alleles are segregating in both habitats in a symmetrical fashion. The graph represents the two local trait densities for each of the two alleles: $\tilde{n}^A_{1,\varepsilon}$, $\tilde{n}^A_{2,\varepsilon}$, $\tilde{n}^a_{1,\varepsilon}$, $\tilde{n}^a_{2,\varepsilon}$ (the same color is for the same deme, and the same linestyle is for the same major-effect allele), as a function of the trait $\zeta = z +\eta^A$ (resp. $z+\eta^a$), where $z$ is the infinitesimal contribution and $\eta^A$ (resp $\eta^a$) is the effect of the major-effect allele. In red, we indicate a visualization of the new variables introduced in \eqref{eq:slow_fast_variables}. $Z_\varepsilon$ is the mean infinitesimal part of the metapopulation, $\delta_\varepsilon$ the spatial average of the local difference between the two allelic mean infinitesimal parts, $\delta^A_\varepsilon$ and $\delta^a_\varepsilon$ the spatial discrepancies in the mean infinitesimal parts per allele. Note the difference in notation between the trait densities $\tilde{n}^A_{i,\varepsilon}$ and the infinitesimal contribution densities $n^A_{i,\varepsilon}$ (which are the ones used in the analysis), which are linked by $n^A_{i,\varepsilon}(z) = \tilde{n}^A_{i,\varepsilon}(z+\eta^A)$ (respectively $z+\eta^a$ for $\tilde{n}^a_{i,\varepsilon}$).}
   \label{fig:slow_fast_variables}
\end{figure}
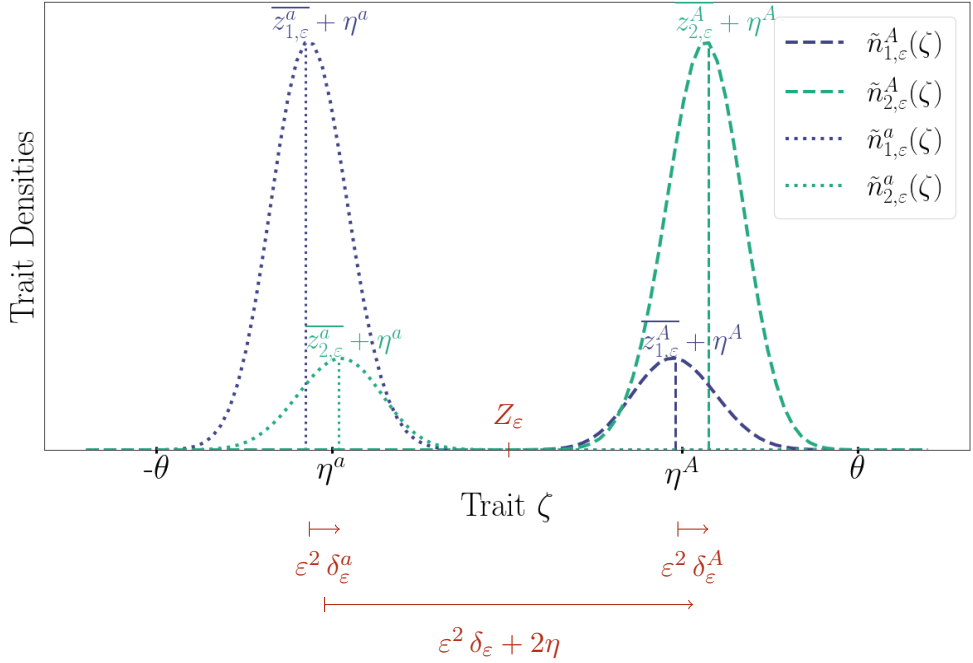

{\color{Black}After the change in variables \eqref{eq:slow_fast_variables}, we obtain a new system (shown in its explicit form \cref{eq:syst_moment_small_var_fast_slow_friendly} in \cref{app:slow_fast_analysis}) which can be written compactly as follows
\begin{equation*}
\label{eq:slowfastvareps}
    \tag{$P_\varepsilon$}
    \begin{aligned}
        \begin{cases}
        \varepsilon^2 \frac{d\Bar{Y_\varepsilon}}{dt} = G(\Bar{Y_\varepsilon},Z_\varepsilon) + \varepsilon^2\,\nu^G_{\varepsilon}(t),\\
        \frac{dZ_\varepsilon}{dt} = -(g_1+g_2)\,Z_\varepsilon+F(\Bar{Y_\varepsilon}) + \varepsilon^2\,\nu^F_{\varepsilon}(t),
        \end{cases}
    \end{aligned}
\end{equation*}
where $\Bar{Y}_\varepsilon := \left(N_{1,\varepsilon}^a, N_{1,\varepsilon}^a,N_{1,\varepsilon}^A,N_{2,\varepsilon}^A,\delta^a_\varepsilon,\delta_\varepsilon^A,\delta_\varepsilon\right)$ denotes the vector of fast variables, located in a set denoted $\Omega := \left(\R_+^*\right)^4\times\R^3$. The two smooth functions $G \in C^\infty(\Omega\times\R)$ and $F \in C^\infty(\Omega)$ encode respectively the fast and slow dynamics. Moreover, the functions $\nu_\varepsilon^G$ and $\nu_\varepsilon^F$ are residues \bl{that} are uniformly bounded w.r.t $\varepsilon$.

\paragraph{Biological interpretation of the different timescales in \eqref{eq:slowfastvareps}.}
As the first four coordinates of the fast variable $\Bar{Y}_\varepsilon$ are the allelic subpopulation sizes, the function $G(\cdot, Z_\varepsilon)$ describes the fast dynamics of growth, death and transfer of populations occurring when the mean quantitative component is at the value $Z_\varepsilon$. The fast timescale of the dynamics of $\Bar{Y}_\varepsilon$ can be interpreted as the \ul{ecological time scale} (including selection on the major-effects locus). On the contrary, the dynamics of the slow variable $Z_\varepsilon$, which is the mean quantitative component, are driven by the shift by selection $-(g_1+g_2)Z_\varepsilon$ and the demographic feedback $F(\bar{Y}_\varepsilon)$, on a slower timescale, which we interpret as the \ul{quantitative evolutionary time scale}. Indeed, notice that the time derivatives are different between the two lines of \ref{eq:slowfastvareps}: the first line involves $\varepsilon^2 \frac{d\cdot}{dt}$, whereas the small factor $\varepsilon^2$ is absent in the second line. 
\paragraph{Convergence to a simplified limit system.}

The slow-fast analysis developed in \cref{app:slow_fast_analysis} is dedicated to show that, when $\varepsilon$ goes to 0, the solutions of \eqref{eq:slowfastvareps} converge to the solutions of the following limit system which separates the ecological and evolutionary time scales
\begin{equation}
\label{eq:slowfastvarlimit}
    \tag{$P_0$}
    \begin{aligned}
        \begin{cases}
        G(\Bar{Y},Z)=0,\\
        \frac{dZ}{dt} = -(g_1+g_2)\,Z+F(\Bar{Y}).
        \end{cases}
    \end{aligned}
\end{equation}
The first line of \eqref{eq:slowfastvarlimit} is an algebraic system defining the slow manifold, constituted by the fast ecological equilibria $\bar{Y}$ corresponding to a value $Z$ of the evolutionary variable (these are formally defined by $\{\bar{Y} \in \Omega, \text{ such that }G(\bar{Y}, Z) =0\}$). The second line describes the dynamics of the slow variable $Z$ constrained to occur on the slow manifold. 

The convergence result linking \eqref{eq:slowfastvareps} to \eqref{eq:slowfastvarlimit} is stated by the following:}
\begin{theorem}
\label{thm:slow_fast_theorem}
For $(\Bar{Y},Z)$ a solution of \eqref{eq:slowfastvarlimit}, there exists $T^*>0$ such that, for $0<\varepsilon<1$, any solution $(\Bar{Y_\varepsilon},Z_\varepsilon)$ of \eqref{eq:slowfastvareps}  on $[0,T^*]$ converges to $(\Bar{Y},Z)$ uniformly on $[0,T^*]$, as $\varepsilon$ goes to 0 and $(\Bar{Y_\varepsilon}(0),Z_\varepsilon(0))$ goes to $(\Bar{Y}(0),Z(0))$.
\end{theorem}
{\color{Black}The proof the \cref{thm:slow_fast_theorem} is detailed in \cref{app:slow_fast_analysis}. The main argument relies crucially on the stability of the fast equilibria at any level defined by a value of the slow variable $Z\in]-1,1[$ (\ref{prop:stability_S0}, \ref{prop:stability_linear}), ensuring that, at the limit, the fast dynamics converge quickly toward the slow manifold and not away from it. The stability argument is completed by the algebraic description of the slow manifold: we show that, for every level $Z\in ]-1,1[$, there exists a single ecological equilibria $\bar{Y}$ satisfying $G(\bar{Y}, Z) = 0$ (\ref{prop:solut_S0}, \ref{prop:linear_syst}). We also summarize in \cref{fig:slow_fast_layout} the links between the different systems, propositions and theorem involved in the slow-fast analysis.}

\begin{rem}[{\color{Black}The one-locus haploid model's equilibrium is part of the fast equilibrium corresponding to the level $Z=0$}]
\label{rem:OLHM_fast_eq}
{\color{Black}The one-locus haploid model is equivalent to the first four differential equations of \eqref{eq:syst_moment_small_var} on the allelic sizes of each subpopulation, with $(\overline{z^a_{\varepsilon, 1}}, \overline{z^a_{\varepsilon, 1}}, \overline{z^a_{\varepsilon, 1}}, \overline{z^a_{\varepsilon, 1}}) = (0, 0, 0, 0)$} (no infinitesimal part - we can obtain from these equations a system describing the allelic frequencies and local population sizes $(p_1, p_2, N_1, N_2) := \left(\frac{N_1^A}{N_1^A+N_1^a},\frac{N_2^A}{N_2^A+N_2^a},N_1^A+N_1^a,N_2^A+N_2^a\right)$, dropping the $\varepsilon$ that is a parameter of the infinitesimal part). Applying \ref{prop:solut_S0} with $Z=0$ gives a unique equilibrium satisfying the first four equations, which is the one found with the one-locus haploid model. One can thus interpret the symmetrical polymorphic equilibrium of the one-locus haploid model as a fast equilibrium in the model presented in this article. It is therefore stable (\ref{prop:stability_S0}) whenever it entails positive population sizes (same condition as in \ref{prop:solut_S0}).
\end{rem}
\begin{rem}[{\color{Black}Degrees of freedom of the slow manifold compared to \cite{Dekens_2022}}]
\label{rem:degrees_of_freedom}
{\color{Black}\ref{prop:solut_S0} states that for every level $Z\in]-1,1[$, there exists a single fast equilibrium $\bar{Y}$ such that $G(\bar{Y}, Z) = 0$. This implies that there are fewer degrees of freedom in the subsystem \eqref{S_0} defining the four allelic subpopulations sizes (see the details in \cref{app:slow_fast_analysis}) than in the analogous system of two equations from the analysis done in \textcite{Dekens_2022}, that can be obtained in the case where one allele has fixed (up to a translation).} Indeed, \textcite{Dekens_2022} shows that the analogous system can have up to three algebraic solutions depending on the parameters. The result of \ref{prop:solut_S0} is thus unexpected, since \ref{S_0} has twice the number of equations and variables.
\end{rem}\section{Results: stability of polymorphism at the major-effect locus in the limit system}
This section follows naturally the separation of timescales shown in \cref{sec:slow_fast} and focuses on the study of the stability of polymorphism at the major-effect locus in the limit system \eqref{eq:slowfastvarlimit}{\color{Black}, in \bl{the} presence of a quantitative background contributing additively to the trait under selection.} To be able to derive analytical conditions, we assume henceforth a symmetrical environment setting (in migration rates, selection strengths, carrying capacities, reproduction rates and major-effect allelic effects):

\[m:=m_1 = m_2,\quad g:=g_1=g_2, \quad\alpha = 1,\quad  \lambda = 1, \quad\eta:=\eta^A = - \eta^a>0.\]

{\color{Black}Under these symmetrical conditions, and in the absence of any quantitative background, we recall that there exists a symmetrical polymorphic equilibrium in the one-locus haploid model, which \ul{is always stable} (see \ref{prop:stability_S0} for a proof). This symmetrical polymorphic equilibrium in the one-locus model corresponds, in our model which considers additionally the additive contribution of a quantitative background on the trait, to the fast equilibrium $Y^*$ associated to the level $Z^* = 0$ ($Z^* = 0$ corresponds to the average quantitative trait between patches cancelling). Because the property of the fast equilibrium does not necessarily transpose to a global equilibrium over multiple timescales, we are therefore interested in the following questions: 

\begin{enumerate}
    \item Does a symmetrical polymorphic equilibrium for the global limit system \eqref{eq:slowfastvarlimit} exist at the level $Z^* = 0$, ie: does the pair of variables $(Z^*, \bar{Y}^*)$ defined above cancel \ul{both} the first line \underline{and} the \bl{right-hand side} of the second line of \eqref{eq:slowfastvarlimit}?
    \item When the symmetrical polymorphic equilibrium exists, is \bl{it} always stable ? Or, in the long-run, \ul{can the slowly evolving infinitesimal background undermine the rapidly established polymorphism at the major-effect locus}, even though the latter is \bl{appears} favored for local adaptation?
    \item If so, can our analysis \ul{predict in which range of parameters} of migration rate $m$, selection strength $g$ and major-effect $\eta$ does that phenomenon occur?
\end{enumerate}

In a first part, we present the results of our analysis \bl{to answer these questions}. We also provide illustrations of the complex patterns \bl{that} can emerge in terms of parameters range, as the studied phenomenon of disturbance of the polymorphism at the major-effect locus by the infinitesimal background exhibits non-monotonic behaviours with regard to each parameter.

In a second part, we confirm the results of the first part thanks to individual-based simulations.}
\subsection{Analytical predictions}
\label{sec:stability}

The results of this section indicate that the unconditional stability of the polymorphism in the OLM can be disturbed by the presence of a quantitative background, for a substantial range of parameters, including, surprisingly, at the strongest selection levels. The interpretation of \cref{rem:OLHM_fast_eq} offers the idea that the infinitesimal background slowly disrupts the rapidly established symmetrical polymorphism at the major-effect locus. 

\paragraph{Existence of a symmetrical polymorphic equilibrium.}
We first show that a symmetrical polymorphic equilibrium can exist under a range of parameters specified in \ref{prop:sym_dim_eq}, as a stationary state of the limit system \eqref{eq:slowfastvarlimit}, hence a solution of the explicit version of the latter:
\begin{equation}
\label{eq:limit_stationary_system}
\begin{aligned}
\begin{cases}
 N^a_1 - \left[N^A_1+N^a_1\right] N^a_1 - g\left[Z-\eta+1\right]^2 N^a_1 +m(N^a_2-N^a_1) &=0,
\\
N^A_1 - \left[N^A_1+N^a_1\right] N^A_1 - g\left[Z+\eta+1\right]^2 N^A_i+m(N^A_2-N^A_1)&=0,
\\N^a_2 - \left[N^A_2+N^a_2\right] N^a_2 - g\left[Z-\eta-1\right]^2 N^a_2 +m(N^a_1-N^a_2) &=0,
\\
N^A_2 - \left[N^A_2+N^a_2\right] N^A_2 - g\left[Z+\eta-1\right]^2 N^A_2+m(N^A_1-N^A_2)&=0,
\\
 2g - m\,\delta^a\left[\frac{N^a_2}{N^a_1}+\frac{N^a_1}{N^a_2}\right]+\frac{\delta^A-\delta^a}{4}\left[\frac{N_2^A}{N_2^a+N_2^A}+\frac{N_1^A}{N_1^a+N_1^A}\right]+\frac{\delta}{2}\left[\frac{N_2^A}{N_2^a+N_2^A}-\frac{N_1^A}{N_1^a+N_1^A}\right]
&=0,\\
2g - m\,\delta^A\left[\frac{N^A_2}{N^A_1}+\frac{N^A_1}{N^A_2}\right]+\frac{\delta^a-\delta^A}{4}\left[\frac{N_2^a}{N_2^a+N_2^A}+\frac{N_1^a}{N_1^a+N_1^A}\right]+\frac{\delta}{2}\left[\frac{N_1^A}{N_2^a+N_2^A}-\frac{N_2^a}{N_1^a+N_1^A}\right]
&=0,\\
- \frac{\delta}{2}-2\,g\,\eta+m\left(\frac{\delta^A}{2}\left[\frac{N^A_2}{N^A_1}-\frac{N^A_1}{N^A_2}\right]-\frac{\delta^a}{2}\left[\frac{N^a_2}{N^a_1}-\frac{N^a_1}{N^a_2}\right] \right)
&=0,\\
\\
-2\,g\,Z +m\left(\frac{\delta^a}{2}\left[\frac{N^a_2}{N^a_1}-\frac{N^a_1}{N^a_2}\right]+\frac{\delta^A}{2}\left[\frac{N^A_2}{N^A_1}-\frac{N^A_1}{N^A_2}\right]\right)+\frac{\delta^A-\delta^a}{4}\left[\frac{N^A_2}{N^A_2+N^a_2}-\frac{N^A_1}{N^A_1+N^a_1}\right]& \\ \qquad\qquad\qquad\qquad\qquad\qquad\qquad\qquad\qquad\qquad\qquad\qquad+\frac{\delta}{2}\left[\frac{N^A_1}{N^A_1+N^a_1}+\frac{N^A_2}{N^A_2+N^a_2}-1\right]& = 0.
\end{cases}
\end{aligned}
\end{equation}
\begin{prop}
\label{prop:sym_dim_eq}
There exists a unique polymorphic equilibrium corresponding to the infinitesimal average $Z=0$ under the condition:
\begin{equation}
    \left[g(\eta^2+1) <1 \right]\lor \left[m < \frac{2\,g^2\,\eta^2}{g(\eta^2+1)-1} - g(\eta^2+1)+1\right].
    \label{eq:condition_viability_0}
\end{equation}
The allelic local population sizes corresponding to this equilibrium satisfy the property:
\[N_{1}^{a,*} = N_{2}^{A,*},\quad N_{2}^{a, *} = N_{1}^{A,*},\]
and for both alleles, the spatial discrepancies between the mean infinitesimal parts of the two patches per allele are the same:
\[\delta^{A,*} = \delta^{a,*}.\]
Therefore, this polymorphic equilibrium at $Z^*=0$ is called symmetrical.
\end{prop}
{\color{Black}The proof uses the results of \ref{prop:solut_S0} and \ref{prop:linear_syst} and is shown in \cref{app:proof_sym_eq}.}

\paragraph{Stability of the symmetrical polymorphic equilibrium.}
Let us recall the limit system \eqref{eq:slowfastvarlimit}:
\begin{equation}
    \begin{aligned}
        \begin{cases}
        G(\Bar{Y},Z)=0,\\
        \frac{dZ}{dt} = -2\,g\,Z+F(\Bar{Y}).
        \end{cases}
    \end{aligned}
\end{equation}
{\color{Black}The stability of the symmetrical polymorphic equilibrium denoted $(0,\bar{Y}^*$) described above in \ref{prop:sym_dim_eq} is studied in the same manner as in \textcite{Dekens_2022}. Because the differential equation on $Z$ in the second line of \eqref{eq:slowfastvarlimit} involves both $Z$ and $\bar{Y}$, it is not sufficient to do a standard linear analysis. The first step is to express the solution to $G(\cdot, Z) = 0$ as a function of $Z$: $\bar{Y}(Z)$. This is possible thanks to the implicit function theorem used in the vicinity of the symmetrical polymorphic equilibrium, because $\left[\partial_{\bar{Y}}G\right]\vert_{Z=0,\bar{Y} = \bar{Y}^*}$ is invertible (thanks to\ref{prop:stability_S0} and \ref{prop:stability_linear}). This step allows us to recast \eqref{eq:slowfastvarlimit} as
\begin{equation}
    \frac{dZ}{dt} = \mathcal{F}(Z) :=-2\,g\,Z+F\left(\bar{Y}(Z)\right),\quad Z \in ]-1,1[.
    \label{eq:autonomous}
\end{equation}
The stability of the symmetrical polymorphic equilibrium $(0, \bar{Y}^*)$ can now be analysed by using the chain rule of differentiation on the \bl{right-hand side} of \eqref{eq:autonomous}.} We obtain that the symmetrical polymorphic equilibrium is asymptotically locally stable if and only if
\[0<2g + \left.\partial_{\bar{Y}}F\cdot\left(\left[\partial_{\bar{Y}}G\right]^{-1}\partial_Z G\right)\right\vert_{Z=0,\bar{Y} = \bar{Y}^*}.\]
Due to the large number of dimensions involved, the explicit formula of the latter is too long to be given here.

{\color{Black}\bl{The patterns resulting from the numerical analysis of the stability of the symmetrical polymorphic equilibrium for four values of the effect of the major-effect locus $\eta \in \{0.5,0.7,1,1.3\}$ are computed} in \cref{fig:stab_region}. For each value of $\eta$, the region of the stability of the polymorphism is indicated in yellow with selection strength ($g$, $x$-axis) and migration rate ($m$, $y$-axis) varying in $[0,3]$. We can first observe that these yellow regions have complex boundaries, and \ul{exhibit non-monotonic behaviours with regard to both migration rate $m$ and selection strength $g$}. These are not predicted by the one-locus \bl{symmetric} model (OLM), which states that polymorphism is maintained everywhere under the dashed yellow line\bl{, which represents the extinction threshold without the quantitative component (computed thanks to the viability condition \eqref{eq:condition_viability_0} stated in \ref{prop:solut_S0} for $Z=0$)}. The latter leads to the conclusion that, when it occurs, \ul{the instability of polymorphism at the major-effect locus shown by our analysis results stems from the presence of the quantitative background due to small-effect loci}.

To describe the non-monotonic behaviour with respect to increasing selection strengths, one can consider holding a constant intermediate migration rate and increase selection (going left to right on a horizontal line in \cref{fig:eta05}, \cref{fig:eta07}, \cref{fig:eta1} and \cref{fig:eta14}). While the polymorphism at the major-effect locus is not stable with weak selection, stability is gained at an intermediate level of selection that depends on the migration rate and subsequently lost at a higher level of selection. \ul{This non-monotonic behaviour when increasing selection levels is quite robust} with regard to different \bl{values of} $\eta$, as shown by the different panels in \cref{fig:stab_region} (even if the effect is attenuated when $\eta = 1$ in \cref{fig:eta1}, which means that the major-effects coincide with the local optima). When selection is weak compared to migration} (left sides of \cref{fig:eta05}, \cref{fig:eta07}, \cref{fig:eta1} and \cref{fig:eta14}), it is expected that the relative blending by migration, which is strong compared to the divergent forces of local selection, provokes the loss of polymorphism. \ul{The loss of polymorphism at the major-effect locus is more surprising and 
counter-intuitive as one would expect that the bonus provided by polymorphism at the major-effect locus}, which helps subpopulations to be locally adapted, \ul{would be even more important at stronger selection levels}, and therefore \bl{maintained}. Unfortunately, the explicit mathematical expression of \eqref{eq:autonomous} is too involved to be truly informative about what is the cause of the loss of polymorphism at the major-effect locus with strong selection. We recommend the reader interested in this to consult the next section presenting the results of individual-based simulations, which provide insights on the origin of this phenomenon.

        \begin{figure}
        \centering
        \begin{subfigure}{.47\textwidth}
            \centering
            \includegraphics[width=\linewidth]{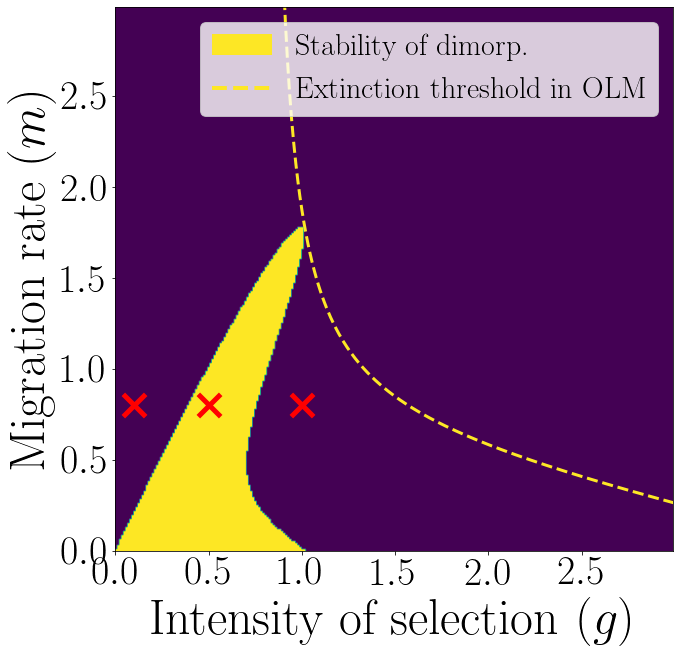}
            \subcaption{$\eta = 0.5$}
            \label{fig:eta05}
        \end{subfigure}
        \begin{subfigure}{.47\textwidth}
            \centering
            \includegraphics[width=\linewidth]{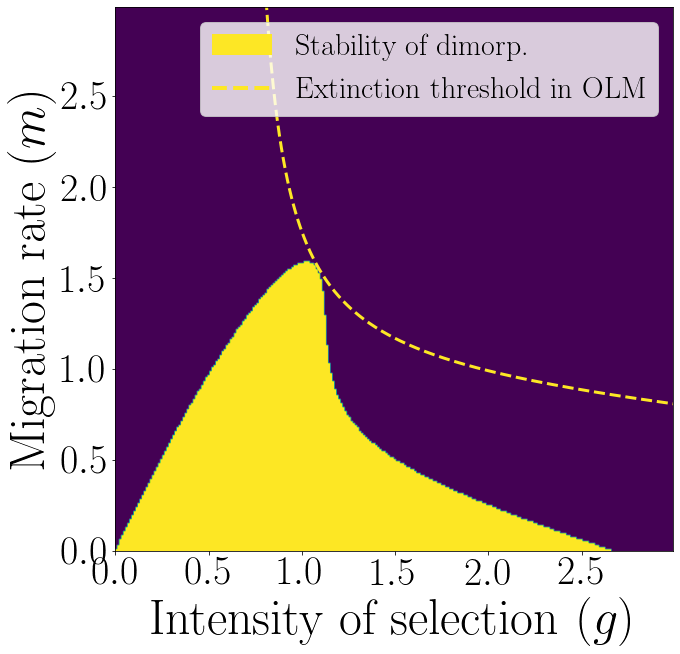}
            \subcaption{$\eta = 0.7$}
            \label{fig:eta07}
        \end{subfigure} \\
        \begin{subfigure}{.47\textwidth}
            \centering
            \includegraphics[width=\linewidth]{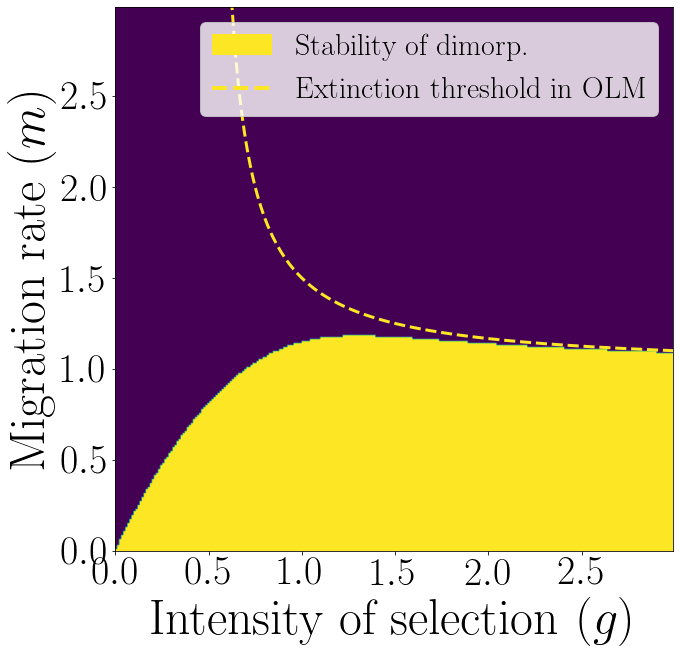}
            \subcaption{$\eta = 1$}
            \label{fig:eta1}
        \end{subfigure}
        \begin{subfigure}{.47\textwidth}
        \centering
            \includegraphics[width=\linewidth]{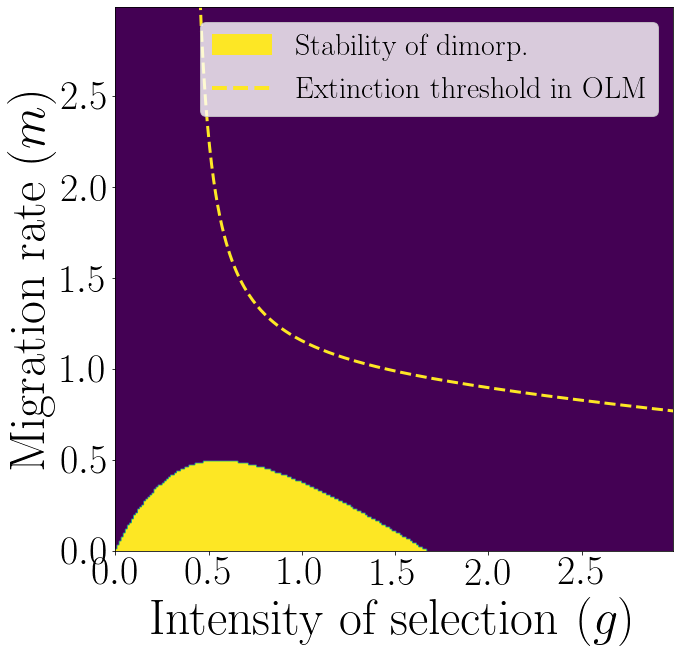}
            \subcaption{$\eta = 1.3$}
            \label{fig:eta14}
        \end{subfigure} 
            \caption{{\textbf{Stability region of the symmetrical polymorphic equilibrium (in yellow), for four major effects $\mathbf{\eta \in \{0.5,0.7,1,1.3\}}$}} (recalling that $\theta = 1$), when $m$ ($y$-axis) and $g$ ($x$-axis) vary in $[0,3]$). This figure highlights the gain and loss of polymorphism with regard to increasing selection, which is not predicted by the one-locus model (abbreviated as OLM in the legend)\bl{, according to which polymorphism is maintained below the extinction threshold represented by the dashed yellow line}. The stable region (in yellow) becomes larger as $\eta$ grows closer to 1, as the major-effect allele effects can then \bl{allow local adaptation to} the two patches on their own, then shrinks again. The red crosses in \cref{fig:eta05} indicate the parameters used for the individual-based simulations (see \cref{sec:individual_based_simulations} and \cref{fig:comp_slim_increasing_selection_sym} and \cref{fig:comp_slim_increasing_selection_asym}).}
            \label{fig:stab_region}
        \end{figure}
\paragraph{Pushing further the numerical analysis of polymorphic equilibria.}

Here, we show a numerical analysis of all the equilibria of the limit system \eqref{eq:slowfastvarlimit} in \cref{fig:dynamic}. To do so, we use the autonomous differential equation \eqref{eq:autonomous} derived previously thanks to the implicit function theorem (used on the whole interval $Z\in ]-1,1[$ thanks to \ref{prop:stability_S0} and \ref{prop:stability_linear}).
From \eqref{eq:autonomous},  $(Z,\bar{Y}(Z))$ is a polymorphic equilibrium if $\mathcal{F}(Z) = 0$, and this equilibrium is locally stable if $\mathcal{F}'(Z) <0$.

Even if the complexity of the limit system is still too great to be analytically solved (due to the implicit nature of the function $\bar{Y}$ defined by the relation $G(\bar{Y}(Z), Z) = 0$), we show in \cref{fig:dynamic} the phase lines corresponding to the limit equation \eqref{eq:autonomous}, when the migration rate and the effect size of the major-effect locus are held constant ($m=0.1, \eta = 0.5$) and the selection strength varies (the lighter the color, the stronger the selection). Solid lines indicate that the system is polymorphic, whereas dotted lines indicate that one major-effect allele has fixed. Every intersection of the zero horizontal line and a solid colored line with a negative slope indicates a locally stable polymorphic equilibrium (conversely, a positive slope indicates an unstable equilibrium). 

This figure is consistent with the analysis of \cref{sec:stability} and \cref{fig:eta05}: at $Z=0$, all the curves return to 0 ($F(Z)$), confirming that a polymorphic equilibrium exists when the mean contribution of the small-effect loci is 0. Their local slope indicates the stability of this equilibrium (stable if negative, \bl{unstable if} positive).
Furthermore, \cref{fig:dynamic} gives insights on the existence of asymmetrical polymorphic equilibria. Particularly, it seems that such equilibria exist for a narrow window of intermediate selection strength: the green curve corresponding to $g = 0.86$ displays two mirrored stable asymmetrical polymorphic equilibria at $Z\approx \pm 0.5$ (indicated by the red arrows), which is hard to predict analytically due to the high orders of polynomials involved. Moreover, such equilibria are presumably quite subtle to catch in individual-based simulations, because the window of selection and the basin of attraction are both narrow. However, this illustrates the new and unsuspected insights that can be obtained from this composite model.
\begin{figure}
    \centering
    \includegraphics[width=\textwidth]{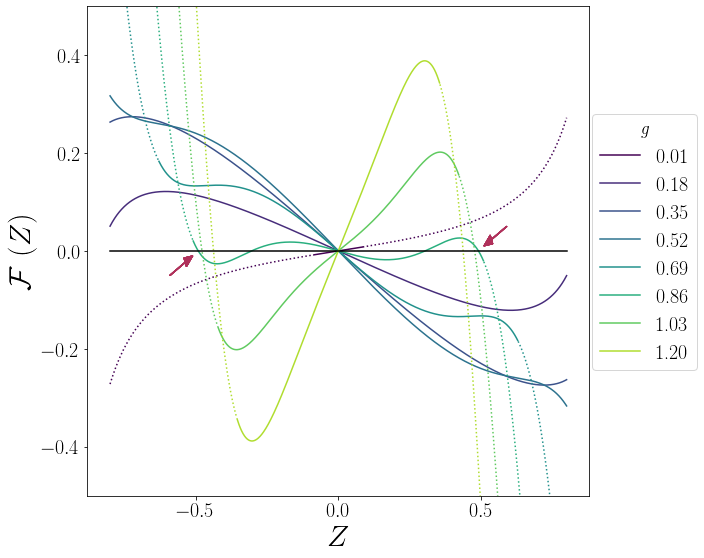}
    \caption{\textbf{Phase lines of the limit equation \eqref{eq:autonomous}, when the migration rate and the strong \bl{allelic} effect are held constant ($m=0.1, \eta = \frac{1}{2}$) and the selection strength varies} (the lighter the color, the stronger the selection). Solid curves indicate that the system is polymorphic, whereas dashed curves indicate that one major-effect allele has fixed. Every intersection of the horizontal black line and a solid colored curve with a negative (resp. positive) slope indicates a locally stable (resp. unstable) polymorphic equilibrium. The darker curve with weak selection $g=0.01$ has a positive slope at $Z=0$ (unstable), the following curves have a negative slope at $Z=0$ (stable for selection between $g=0.18$ and $g=0.86$), and finally the lightest curves have a positive slope at $Z=0$ (unstable for $g\geq 1.03$), which is consistent with \cref{fig:eta05}. Note that there exists additionally two mirrored asymmetrical polymorphic equilibria for $g = 0.86$, for $Z \approx \pm 0.5$ (indicated by the red arrows), which were unsuspected prior to this numerical analysis.}
    \label{fig:dynamic}
\end{figure}

\subsection{Individual-based simulations}
\label{sec:individual_based_simulations}
{\color{Black}In this part, we confirm the results} given by our analysis on the stability of the symmetrical polymorphic equilibrium, using individual-based simulations conducted with the software SLiM (\cite{10.1093/molbev/msy228}). We focus on the gain and loss of polymorphism with regard to increasing selection, when $\eta \neq 1$, for symmetrical and asymmetrical initial conditions. For each set of parameters, we ran 20 replicate simulations. The results \bl{for} the major-effect locus are displayed in \cref{fig:comp_slim_increasing_selection_sym} and \cref{fig:comp_slim_increasing_selection_asym}, both with a quantitative background (left panel) and without (right panel). The simulations confirm that variation is maintained only {\color{Black}for intermediate levels} of selection (as measured by $p(1-p)$, where $p$ is the local frequency of allele $A$). \bl{They also provide some insights regarding the cause of the surprising loss of polymophism at the major-effect locus with strong selection.} The \bl{simulation procedure} is detailed as follows.
\paragraph{Populations and habitats.}

The species is split in two subpopulations living in two different habitats, with local carrying capacity $\boldsymbol{K} = 10^4$. In each habitat, individuals experience selection toward a local trait optimum $\boldsymbol{\theta_i} = (-1)^i$ (for habitat $i$). Initially, the two subpopulations are at $\frac{4}{5}$ of the local carrying capacity. The genetic information of the individuals of the initial population is set as follows. In each subpopulation, all the individuals have, at the major-effect locus, the allele whose effect is the closest to the optimum of the habitat they are in ($\boldsymbol{\eta}$ in habitat 2 and $-\boldsymbol{\eta}$ in habitat 1). The polygenic background is then set randomly and uniformly. 
\paragraph{Genetic architecture.}
{\color{Black}We consider $\boldsymbol{L} = 200$ unlinked loci constituting the polygenic background. At each of these loci, two alleles segregate, having an additive effect on the trait of the individual of value $\frac{\boldsymbol \sigma_{LE}}{\sqrt{\boldsymbol L}}$ or $-\frac{\boldsymbol \sigma_{LE}}{\sqrt{\boldsymbol L}}$, where $\boldsymbol{\sigma_{LE}^2}$ is the variance at linkage equilibrum of the quantitative background. No mutation occurs at those loci. We set the variance at linkage equilibrium to $\boldsymbol{\sigma_{LE}} = 0.1$ small, so that our analysis in a small variance regime is a good approximation. (In \cref{app:50}, we consider the same framework with a smaller number of loci involved in the quantitative background $\boldsymbol{L} = 50)$. }

There is an additional locus of interest, which carries the major-effect alleles $+\boldsymbol{\eta}$ or $-\boldsymbol{\eta}$. This locus is also unlinked to all the others and no mutation occurs at this site. {\color{Black} Note also that the trait range, given by $[-\boldsymbol{\eta} - \boldsymbol{\sigma_{LE}}\sqrt{\boldsymbol{L}}, \boldsymbol{\eta} + \boldsymbol{\sigma_{LE}}\sqrt{\boldsymbol{L}}] = [-\boldsymbol{\eta} -\sqrt{2},\boldsymbol{\eta} + \sqrt{2}]$ extends beyond the local optima (-1,1), even in the absence of major-effects.}

\paragraph{Life cycle.}

The life cycle involves {\color{Black}overlapping generations of small time length $\boldsymbol{\Delta t} = 0.1$}. The life cycle proceeds as follows:
\begin{enumerate}
    \item \emph{reproduction}: each individual of the metapopulation chooses at random one mate within its subpopulation, and {\color{Black}their mating produces an offspring with probability $\boldsymbol{\Delta t}$.}
    \item \emph{selection-competition}: each individual (including offspring generated in the previous step) faces a selection-competition trial according to its trait $\boldsymbol{\zeta}$ and habitat $i$ in which they are currently living. {\color{Black}They survive with probability $\exp\left(-\boldsymbol{g}\boldsymbol{\Delta t}(\boldsymbol{\zeta}-\boldsymbol{\theta_i})^2-\boldsymbol{\Delta t} \frac{\boldsymbol{N_i}}{\boldsymbol{K}}\right)$ and are removed otherwise (here $\boldsymbol{N_i}$ denotes the size of the subpopulation $i$ after reproduction).}
    \item \emph{migration}: at each migration event, within each subpopulation $i$, a number of migrants is drawn, according to a Poisson law with parameter {\color{Black}$\boldsymbol{m\Delta t\;N_i}$ (with a hard cap of $\boldsymbol{N_i}$, which is the number of individuals currently in the subpopulation after the selection-competition step)}. Migrants are uniformly sampled accordingly within the subpopulation and are moved to the other deme. {\color{Black}We stress that a given value of the migration rate $\boldsymbol{m}=0.8$ means that, on average, a fraction $\boldsymbol{m}\boldsymbol{\Delta t} =0.08$ of the population will change deme at each generation}.
\end{enumerate}
Each simulation repeats this life cycle, first without migration for 100 generations of burn in (10 time units) and next with migration for {\color{Black}$\boldsymbol{N_\text{gen}} = 10^4$ generations ($10^3$ time units)}. We model two types of initial events when migration starts: either nothing happens, and the initial state is symmetrical, or we model a sudden catastrophic loss of population in only one of the habitat during the first generation with migration, so that the initial state is asymmetrical (results shown in \cref{app:50}). Precisely, we change the mortality of the uniform competition term only in the habitat 1, by replacing $\exp\left(-\boldsymbol{\Delta t} \frac{\boldsymbol{N_1}}{\boldsymbol{K}}\right)$ by $\exp\left(- \frac{\boldsymbol{N_1}}{\boldsymbol{K}}\right)$ (which is consistent with the interpretation that this catastrophic loss of population is very abrupt). This leads to asymmetrical initial subpopulation sizes.

\paragraph{Qualitative results of the IBS on the stability of polymorphism with regard to increasing selection.}
{\color{Black}{The \bl{solid} lines in the sufigures of \cref{fig:comp_slim_increasing_selection_sym} (symmetrical initial population sizes) and \cref{fig:comp_slim_increasing_selection_asym} (asymmetrical initial populations sizes) represent the median trajectories of the variance at the major-effect locus ($p(1-p)$, where $p$ is the local frequency of allele $A$) in each habitat (gold lines for habitat 1 and navy ones for habitat 2). When the variance $p(1-p)$ is positive, the $A/a$ polymorphism is maintained. In both \cref{fig:comp_slim_increasing_selection_sym} and \cref{fig:comp_slim_increasing_selection_asym}, selection increases from top to bottom and the polygenic background is present in the left panel and absent in the right one. \ul{When there exists a polygenic component contributing to the trait, polymorphism at the major-effect locus is lost after some time with weak selection} ($g=0.1$, \cref{fig:m03g01_sym} and \cref{fig:m03g01_asym}), \ul{is maintained with intermediate selection} ($g=0.5$, \cref{fig:m03g05_sym} and \cref{fig:m03g05_asym}) \ul{and lost again even more quickly with strong selection} ($g=1$, \cref{fig:m03g1_sym} and \cref{fig:m03g1_asym}). \ul{This is qualitatively consistent with the analytical predictions} displayed in \cref{fig:eta05}, where the red crosses indicate the three selection-migration set of parameters chosen for the IBS. Moreover, this phenomenon is robust with regard to initial conditions (\cref{fig:comp_slim_increasing_selection_sym} and \cref{fig:comp_slim_increasing_selection_asym}), although the \bl{loss} of polymophism at the major-effect locus at weak and strong selection is faster when subpopulations sizes are initially asymmetrical (\cref{fig:comp_slim_increasing_selection_asym}).}}

\paragraph{Control case without polygenic background.}

To confirm {\color{Black} that the loss of polymorphism at weak and strong selections is due to the polygenic background and not to genetic drift (although drift is unlikely to have an effect under this time range of $10^3$ time units with a population of order $10^4$)}, we additionally run an equal number of replicates for each set of parameters without any polygenic background $(\boldsymbol{L}=0, \boldsymbol{\sigma_{LE}} = 0)$. Only the major-effect alleles segregates, and this corresponds to the one-locus haploid model. \ul{Results shown in the right panel of} \cref{fig:comp_slim_increasing_selection_sym} and \cref{fig:comp_slim_increasing_selection_asym} \ul{are consistent with the one-locus haploid model analysis}, \bl{which} states that the polymorphism at this major-effect locus is stable at all level of selection (the variance at the major-effect locus remains positive and stable).

\bl{\paragraph{Explanation behind the loss of polymorphism with strong and weak selection.}
The IBS allow us to gain some insights about the cause of the major locus polymorphism's collapse. \bl{We first consider strong selection}. In particular, the dynamics of the subpopulations sizes and the local mean traits reveal that, at one point, stochasticity creates a small shift in the local mean traits. This shift is the same between bearers of $A$ and $a$ and in both patches, because the small segregational variance of the quantitative background binds the quantitative background values to be approximately the same for everyone (in the analysis, this is reflected by the change in variables \eqref{eq:slow_fast_variables} introducing $\delta^a_\varepsilon,\delta^A_\varepsilon$ and $\delta_\varepsilon$). Therefore, this shift, which is toward one of the local optima improves the adaptation of one subpopulation and is deleterious \bl{in} the other one, which causes an asymmetry in subpopulation sizes, \bl{which is} particularly pronounced when selection is strong (see the small figure embedded in \cref{fig:m03g1_sym}, built by selecting one of the two asymmetries for the sake of clarity). Because of this asymmetry, the migrants' flow is also asymmetrical and the larger population then undermines even more the small population by gene flow, which in turn raises the frequency of the major-effect allele favoured in the larger patch, which then further increases the disparity in population sizes among the two patches. This positive feedback loop creates a vortex related to the phenomenon of \textit{migrational meltdown} identified in the quantitative genetic model of \textcite{Ronce_Kirkpatrick_2001}, which eventually leads to the loss of one of the major-effect alleles. 

One crucial feature of this explanation is the dynamics of the (varying) subpopulation sizes, which our eco-evo model allows us to track. To confirm this intuition, we conducted the same IBS, but with adjusting the birth rate to compensate for the deaths at every generation, effectively keeping both subpopulations at a constant size. With constant subpopulations sizes, the polymorphism at the major-effect locus is not lost with strong selection (see \cref{fig:m03g1_fixed_sizes} in \cref{app:fixed_sizes}). This truly highlights the role of the eco-evo framework in which subpopulations sizes are variables that are allowed to vary.

We next consider weak selection, where one might wonder if the loss of polymorphism at the major-effect locus relies on the same mechanism. In this case, random fluctuations at the major-effect locus cause the quantitative trait to shift in the opposite direction, ensuring that the mean trait remains near 0 (the midpoint between the two optima). This is because migration is so high relative to selection that selection favours lineages that survive well in both patches. However, here, the small figure embedded in \cref{fig:m03g01_sym} suggests the role of varying subpopulations sizes in this phenomenon is not as important. This is confirmed by the fact that the loss of polymorphism also occurs in IBS where subpopulations sizes are kept at a constant level (\cref{fig:m03g01_fixed_sizes}). This implies that the loss of polymorphism at the major-effect locus with weak and with strong selection fundamentally differ. With strong selection, randomly generated asymmetries drive the system toward specialization in one patch and fixation of the major allele in that patch, whereas with weak selection, randomly generated asymmetries drive one major allele to fix and the quantitative trait to compensate in such a way that individuals are generalist with a mean trait near the midpoint between the patch optima.
}

\paragraph{Quantitative comparison of IBS with the continuous-time deterministic model \eqref{systnonstat}.}
{\color{Black}\bl{We also} ran deterministic numerical \bl{iterations} of \eqref{systnonstat} to check the quantitative constistency of the stochastic IBS with the deterministic model \eqref{systnonstat}. Two series were run, one for each type of initial condition (symmetrical or asymmetrical initial subpopulations sizes). The median trajectory obtained from these deterministic numerical resolutions of \eqref{systnonstat} for each set of parameters and initial conditions are displayed by the dashed lines in all the subfigures of \cref{fig:comp_slim_increasing_selection_sym} and \cref{fig:comp_slim_increasing_selection_asym}. \ul{These deterministic trajectories are in excellent agreement with the ones obtained from the stochastic IBS and provide good approximations}. We choose to distinguish the two types of initial conditions (asymmetrical or symmetrical) because the deterministic numerical resolutions are unequally sensitive to them. Indeed, since the environment is symmetrical, the symmetrical initial state is at an unstable edge between two symmetrical stable valleys for the interesting range of parameters and thus wanders for some time before choosing a valley to fall into. Therefore, we initialized the deterministic resolutions with the symmetrical initial state according to their respective IBS replicate states at 20\% of the time before the median fixation time (140 times units for \cref{fig:m03g01_sym} and 40 time units for \cref{fig:m03g1_sym}). With the asymmetrical initial conditions (shown in \cref{fig:comp_slim_increasing_selection_asym} in \cref{app:50}), this sensitivity is greatly reduced, and the deterministic resolutions are initialized according to their respective IBS replicate states when migration starts (0 time units). Furthermore, the numerical scheme for the resolution of the deterministic model \eqref{systnonstat} uses a splitting scheme to handle successively migration and the ecological dynamics internal to each habitat. For the latter, we use a discretization of the Duhamel' integral formula on time step of lengths $\boldsymbol{\Delta t}$ for the asymmetrical initial state series and $\frac{\boldsymbol{\Delta t}}{4}$ for the symmetrical initial state series (with an accordingly increased number of time steps).}

\begin{figure}
\begin{subfigure}{.45\linewidth}
\centering

        \stackinset{l}{45pt}{b}{25pt}{\includegraphics[scale=.09]{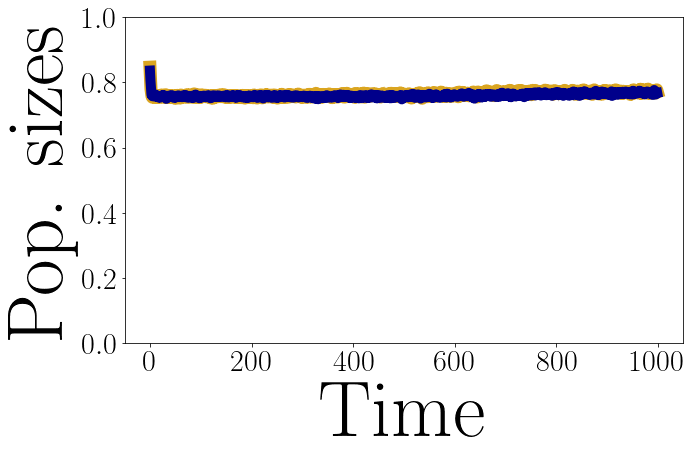}}{\includegraphics[width=\linewidth]{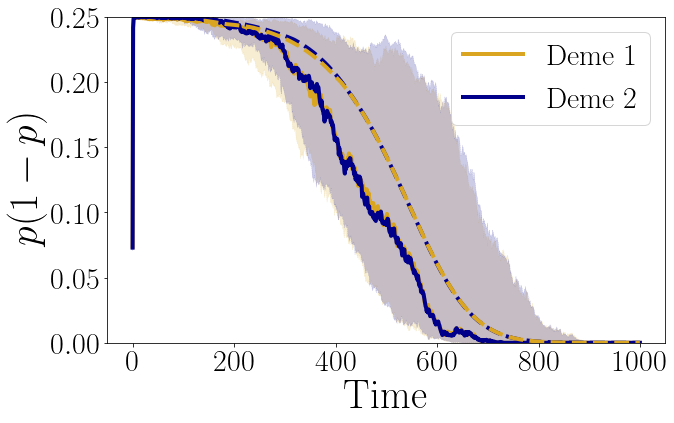}}
    \subcaption{\small{\color{Black}Major-effect locus with polygenic background: weak selection ($g=0.1$)}.}
    \label{fig:m03g01_sym}
\end{subfigure}
\qquad\begin{subfigure}{.45\linewidth}
\centering
    \includegraphics[width=\linewidth]{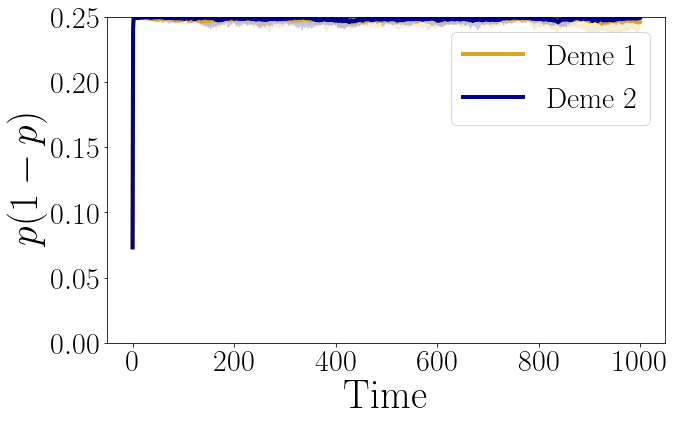}
    \subcaption{\small Control case without polygenic background: weak selection ($g=0.1$).}
    \label{fig:m03g01control_sym}
\end{subfigure}
\\
\begin{subfigure}{.45\linewidth}
\centering
    \stackinset{l}{45pt}{b}{25pt}{\includegraphics[scale=.09]{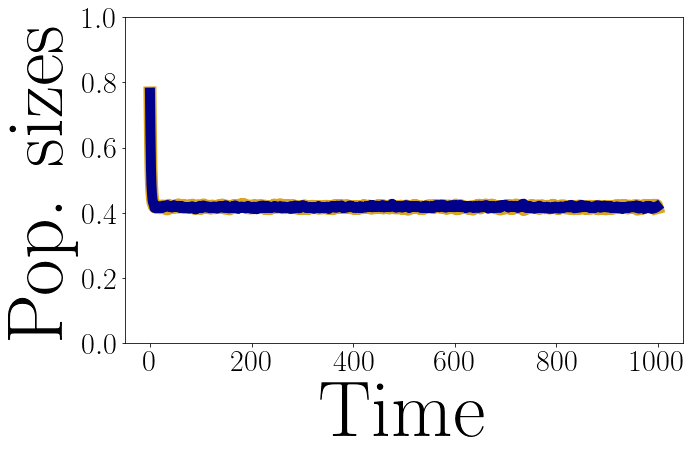}}{\includegraphics[width=\linewidth]{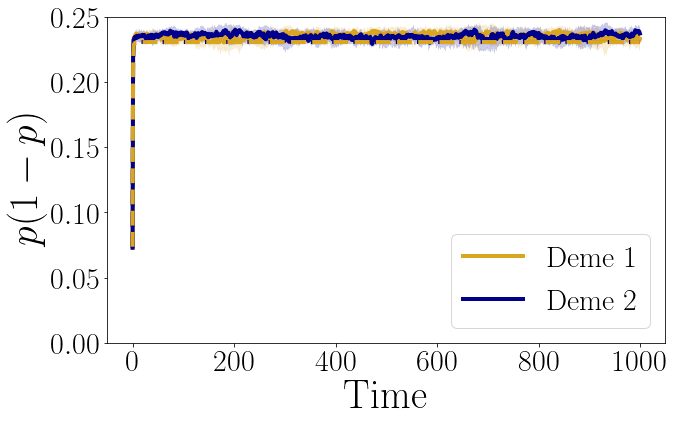}}

    \subcaption{\small Major-effect locus with polygenic background: intermediate selection ($g=0.5$).}
    \label{fig:m03g05_sym}
\end{subfigure}
\qquad\begin{subfigure}{.45\linewidth}
\centering
    \includegraphics[width=\linewidth]{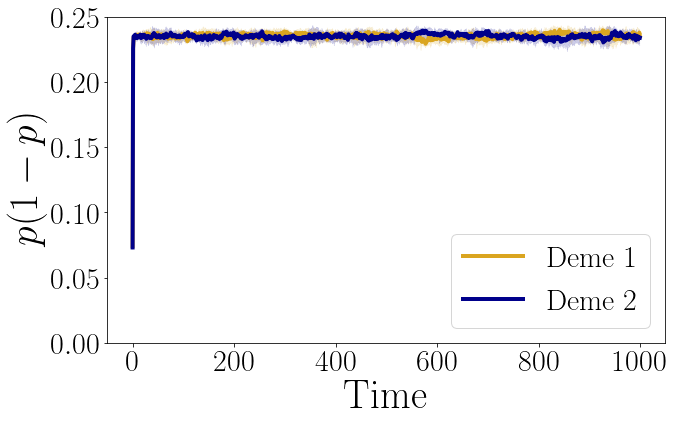}
    \subcaption{\small Control case without polygenic background: intermediate selection ($g=0.5$).}
    \label{fig:m03g05control_sym}
\end{subfigure}
\\
\begin{subfigure}{.45\linewidth}
\centering
    \stackinset{c}{55pt}{b}{25pt}{\includegraphics[scale=.09]{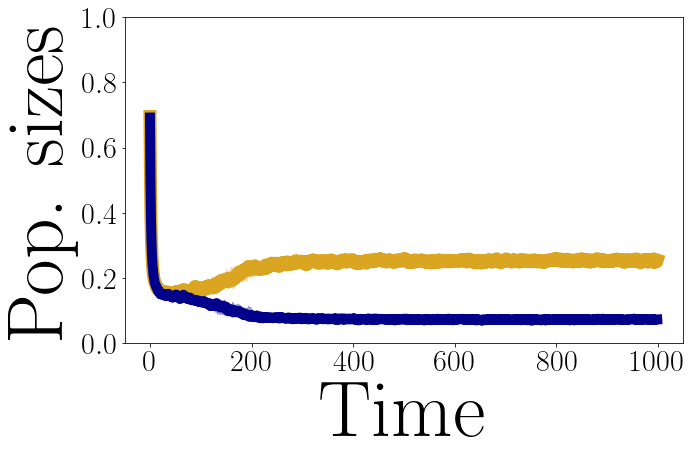}}{\includegraphics[width=\linewidth]{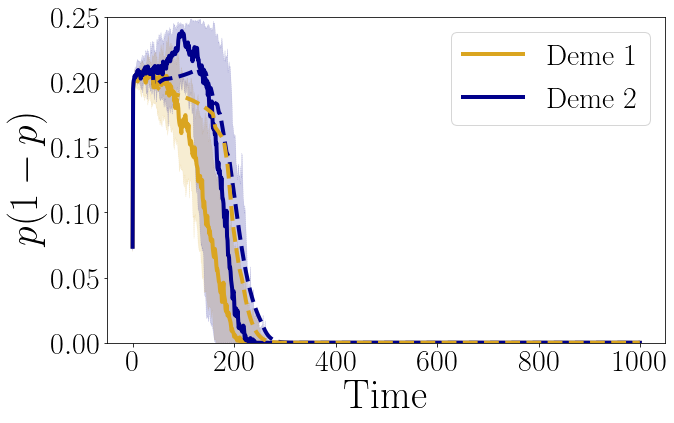}}

    \subcaption{\small{\color{Black}Major-effect locus with polygenic background: strong selection ($g=1$).}}
    \label{fig:m03g1_sym}
\end{subfigure}
\qquad\begin{subfigure}{.45\linewidth}
\centering
    \includegraphics[width=\linewidth]{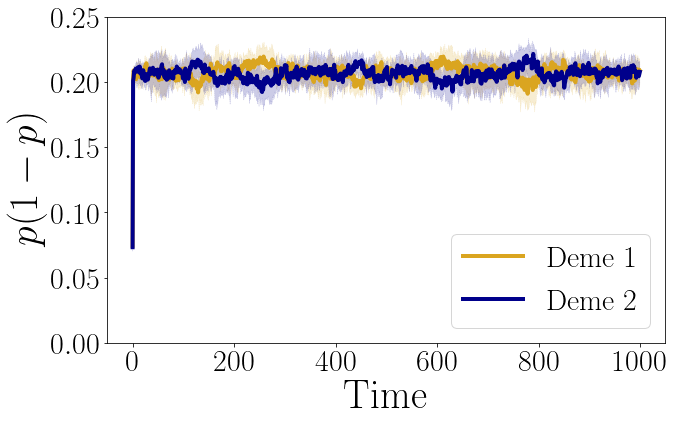}
    \subcaption{\small Control case without polygenic background: strong selection ($g=1$).}
    \label{fig:m03g1control_sym}
\end{subfigure}
\caption{\small\textbf{Variance at the major-effect locus across time for increasing selection (top to bottom: $g = 0.1, \;0.5, \;1$) at a fixed rate of migration ($m = 0.8$), {\color{Black}with symmetrical initial subpopulation sizes}}. $p$ denotes the local frequency of the major-effect allele $A$. The left panel is obtained with both a major-effect locus ($\boldsymbol{\eta} = 1/2$) and {\color{Black}\bl{a} polygenic background of $200$ loci}, whereas only the major-effect locus is present in the right panel. For each subfigure, 20 replicates simulations were run per set of parameters, according to the setting explained in \cref{sec:individual_based_simulations}. In each subfigure, the solid line represents the median trajectory and the shaded area indicates the 0.2 and 0.8 quantiles. {\color{Black}The dashed lines represent the median trajectories of the numerical resolutions of the deterministic model \eqref{systnonstat}}. This figure confirms that polymorphism of the major-effect locus is maintained only when selection is intermediate in strength \bl{(panel c)} in presence of a polygenic background \bl{(left panel)}. \bl{The small figures embedded in each figure of the left panel represent the dynamics of the subpopulation sizes ($N_1$ and $N_2$). They highlight the qualitative difference between the loss of polymorphism at the major-effect locus with weak or strong selection, as fixation occurs without change in subpopulation size with weak selection (\cref{fig:m03g01_sym}) and as subpopulations sizes become asymmetrical with strong selection (\cref{fig:m03g1_sym}).}}
\label{fig:comp_slim_increasing_selection_sym}
\end{figure}

\newpage
\section{Discussion}

\paragraph{Summary.}
In this work, we present a new eco-evo model for selection in a heterogeneous environment that combines a major-effect locus with a quantitative genetic background, without assuming that the latter is normally distributed. \bl{With this model, we aim to examine how the presence of a small quantitative background can disturb the polymorphism at the major-effect locus, which on its own would be favoured in the type of setting we consider.} This model bridges a population genetic model (one-locus haploid model) with a quantitative genetic model recently studied in a heterogeneous environment (\cite{Dekens_2022}). To do so, it introduces a new reproduction operator, inspired by the infinitesimal model, that encodes the inheritance of a major-effect and a quantitative background. The analysis goes deeper than previous studies, by formally justifying that the polygenic component of the trait is normally distributed around the major-effect allelic effects in a regime of small variance\bl{ and hence justifying the Gaussian assumption made in \textcite{Lande_1983} and \textcite{Chevin_Hospital_2008}}. To show this, we find new convex analysis arguments that leads to a separation of time scales, which allows us to study the stability of the polymorphism at the major-effect locus. We show that this polymorphism, which is maintained at intermediate selection, is subsequently lost when selection increases beyond a certain threshold, a phenomenon qualitatively confirmed by individual-based simulations. The separation of time scales' point of view offers the interpretation that the infinitesimal background slowly disrupts the \bl{rapidly} established symmetrical polymorphism at the major-effect locus. Therefore, this phenomenon cannot be predicted by the one-locus haploid model (without the quantitative background). To our knowledge, this phenomenon has not yet been documented. 

\bl{\paragraph{The importance of the eco-evo framework and the influence of small segregational variance.}

In the last section, we provided an explanation for our main biological result, which is the unexpected loss of polymorphism at the major-effect locus with both strong and weak selection. With strong selection, the explanation relies on two factors. First, \ul{the mean quantitative background is constrained to move similarly in both patches and for bearers of $A$ and $a$ because of the small segregational variance}. This implies that any slight shift of the mean quantitative background necessarily increases local adaptation to one patch and decreases local adaptation to the other. Consequently, the latter creates \ul{an asymmetry in subpopulation sizes}, one being better adapted than the other. This asymmetry is significant when local selection is strong. The larger subpopulation then sends relatively more migrants to the other patch, undermining the local adaptation there even more, which contributes to raise the frequency of the allele favoured in the now larger patch everywhere. In turn, the combination of increasing specialization and increasing disparity in population sizes (and therefore migrant production results) in a vortex that can be identified as a \textit{migrational meltdown} (coined in \cite{Ronce_Kirkpatrick_2001}).

Therefore, one can observe that this phenomenon specifically \ul{relies first on our eco-evo framework}, which allows us to track the dynamics of subpopulation sizes. This highlights the importance of considering eco-evo dynamics when dealing with strong selection (which can heavily impact population sizes), as this loss of polymorphism at the major-effect locus with strong selection would not be captured by a more traditional approach in population genetics that considers the population sizes constant. Second, this phenomenon of loss of polymorphism at the major-effect locus also \ul{relies crucially on the small segregational variance of the quantitative background}, which is linked to the very small effect sizes of sufficiently many alleles. It is indeed worth noting that, if the segregational variance of the quantitative background is relatively large, then the first step of our explanation ("the mean quantitative background is constrained to move similarly in both patches and between bearers of $A$ and $a$") does not hold and the mean quantitative background can shift in opposite directions in the two patches, improving local adaptation in both patches. The impact of bimodality in the quantitative trait, with mean trait values in each patch near the local optimum, on the stability of polymorphism at the major-effects locus deserves further attention.
}

\bl{\paragraph{Robustness.} To assess the robustness of the swamping phenomenon that we identified, we performed various individual-based simulations. We found that these are in excellent quantitative agreement with our analysis. They also connect our framework to the evolution of an explicit genetic architecture, which provides a practical translation of the small variance regime \bl{that} underlies our study. This is important, because we have only shown that our results hold in this small variance regime. In particular, they might be different under parameter ranges \bl{that} violate this regime, for example under low or no migration (meaning, at a level of comparable order \bl{as} the small variance). Moreover, because the trajectories of the individual-based simulations are consistent with the deterministic trajectories produced by our model, we can validate essential \bl{assumptions} underlying our model, mostly the constancy of the small segregational variance for the quantitative background. The latter requires enough loci ($\boldsymbol{L}$) with relatively small effects ($\pm\frac{\boldsymbol{\sigma_{LE}}}{\sqrt{\boldsymbol{L}}}$), so that the segregational variance of the quantitative background (lower than $\boldsymbol{\sigma_{LE}^2}$) remains small while the phenotypic range produced by the polygenic background alone $\left[-\boldsymbol{\sigma_{LE}}\sqrt{\boldsymbol{L}}, \boldsymbol{\sigma_{LE}}\sqrt{\boldsymbol{L}}\right]$ spans well beyond the local optimal traits. The last condition is necessary to ensure genotypic redundancy (see also \cite{Yeaman_2022}), so that well adapted mates with similar phenotypes have on average relatively different genotypes, which in turn ensures that the variance of their offspring does not depend too much on their traits. In our simulations presented in \cref{sec:individual_based_simulations}, we showed that $\boldsymbol{L} = 200$ and $\boldsymbol{\sigma_{LE}}=0.1$ produced very similar trajectories to our deterministic model. Furthermore, in \cref{app:50}, we even lowered the number of \bl{loci} to $\boldsymbol{L} = 50$ and increased the segregational variance parameter $\boldsymbol{\sigma_{LE}}=0.2$ to assess the robustness of our conclusions with regard to less favorable parameters, with the same conclusions.}

\paragraph{Complete analytical outcomes.}

The analysis performed in \cref{sec:stability} is centered on the persistence of polymorphism at the major-effect locus. As stated in \cref{rem:fixation_alleles_system}, the loss of polymorphism by fixation would lead to the dynamics of the quantitative background alone, as covered in \textcite{Dekens_2022}. Hence, \cref{fig:outcomes} complements \cref{fig:stab_region} (for $\eta = 0.5$ and varying migration and selection)\bl{. First, it displays} the region of parameters where \bl{where polymorphism at the major-effect locus would be maintained, leading to a bimodal trait distribution in the metapopulation} (in yellow, corresponding to the analogous region in \cref{fig:eta05}). \bl{For the rest of the parameters, either the population goes extinct (purple region, when both migration and selection are too strong), or one major-effect allele fixes and the other one is lost. Therefore, in this case, the equilibria are given by the analysis done in \textcite{Dekens_2022} (as anticipated by \cref{rem:fixation_alleles_system}). These equilibria are monomorphic and are of two types. First, for bounded selection, there exists a critical threshold in the migration rate under which the polymorphism at the major-effect locus is lost due to the strong blending effect of migration. In that case of strong relative migration, the metapopulation occupies equally the two habitats as generalist and its trait distribution concentrates around the midpoint between the two habitats' optima. It is therefore constituted by individuals that do not suffer from a loss in fitness when migrating between the patches. This corresponds to the symmetrical monomorphic equilibrium (see \cite{Ronce_Kirkpatrick_2001,Dekens_2022}), where the population can be qualified as generalist (green region).  Second, for bounded migration rates, the polymorphism at the major-effect locus lost above a certain threshold in selection. In this case of relative strong selection (blue region), the metapopulation becomes specialized in one of the two patches, and deserts the other (at the exception of a few migrants who are too maladapted to establish there). This type asymmetrical equilibrium, highlighted as a source-sink scenario in \textcite{Ronce_Kirkpatrick_2001}, was analytically derived in \textcite{Dekens_2022}.}
\begin{figure}
    \centering
    \includegraphics[width=0.9\textwidth]{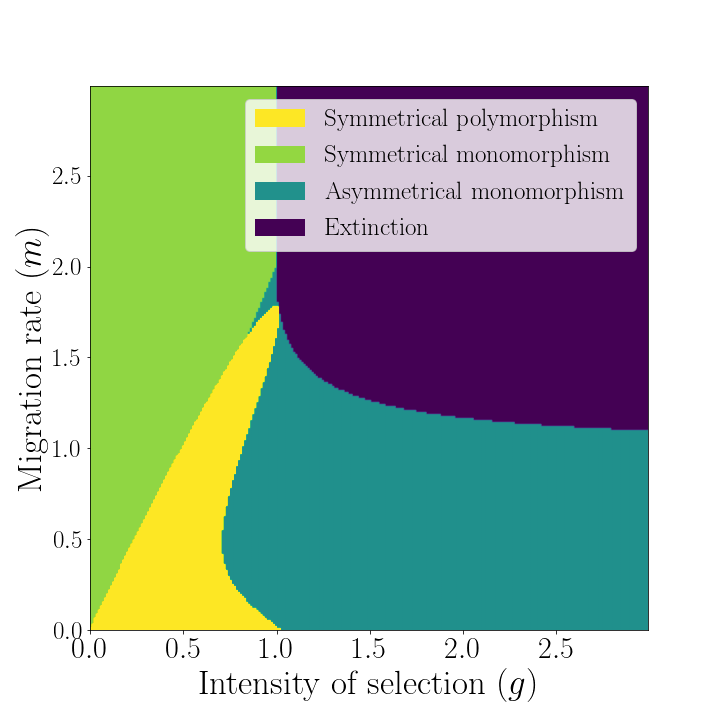}
    \caption{\textbf{Summary of the complete analytical outcomes of the model}, for $\eta = 0.5$ and varying migration ($y$-axis) and selection ($x$-axis). The figure combines the results obtained in \cref{sec:stability} on the stability of the symmetrical polymorphic equilibrium with the results of the model of \textcite{Dekens_2022} (equivalent to this model upon loss of polymorphism). For bounded selection, when migration increases, there is a threshold over which the polymorphism at the major-effect locus is lost due to the blending effect of migration (consistent with \cite{Yeaman_Whitlock_2011}). The population then becomes equally maladapted to both habitats (generalist - symmetrical monomorphism, in the green region). For this specific major-effect allelic effect $\eta = 0.5$, there exists additionally an interval of selection strength ($\approx [0.7, 1]$) for which the major polymorphism might not be stable at all migration rates below the critical threshold. This phenomenon does not seem to hold when the major effect is larger (see \cref{fig:eta07}). For bounded migration (below the threshold rate over which the strong migration blending hampers the major polymorphism), when selection strength increases, the polymorphism at the major-effect locus (yellow region) is lost, and the population becomes adapted to one of the two habitats (specialist - asymmetrical monomorphism, in the blue region).  As this figure is obtained in the small segregational variance regime (which should remain smaller than the other parameters of the system for the analysis to be valid), we warn that the outcomes displayed in the vicinity of the $x$-axis (very small migration rates) might not correspond to the limit when the migration rate is 0 (no migration).}
    \label{fig:outcomes}
\end{figure}

\bl{\paragraph{Filling a methodological gap.} In population genetics, one-locus or two-locus models in heterogeneous environments have been well studied (\cite{Nagylaki_Lou_2001,Burger_Akerman_2011}), with a nuanced picture when including the effect of drift (\cite{Yeaman_Otto_2011}). A two-deme two-locus model is analysed in \textcite{Geroldinger_Bürger_2014}, which in particular shows that a concentrated genetic architecture (a major-effect locus and a tightly linked minor one) maintains polymorphism (full or single-locus) even under high migration rates when selection acts in opposite directions in the two patches. Increasing the number of loci quickly leads to analytical complexity too great for general study. There also exist multi-loci models in heterogeneous environments (\cite{Lythgoe_1997,Szep_sachedva_barton_2021}), but they focus on equal
allelic effects. On the other end of the spectrum, quantitative genetic models do not typically account
for additional discrete major-effect allelic effects on the focal quantitative trait (for sexually reproducing
populations in heterogeneous environment, see Ronce and Kirkpatrick 2001; Hendry, Day, and Taylor
2001; Dekens 2022 and for asexually reproducing populations, see Débarre, Ronce, and Gandon 2013;
Mirrahimi 2017; Mirrahimi and Gandon 2020; Hamel, Lavigne, and Roques 2021).

To our knowledge, the first model that bridges this gap between quantitative traits and discrete loci appears in \textcite{Lande_1983}. In this work, the author considers the dynamics of a single major-effect locus where two alleles segregate along with a polygenic background, in a diploid panmictic population subjected to a sudden change of environment. He models the influence of the polygenic background on the trait by assuming that, among bearers of the same major-effect allele, the trait distribution is Gaussian, centered around the effect of the the major-effect allele on the trait. This study opened the way for more recent work on the genetic architecture of adaptation of panmictic populations in a suddenly changing environment, where the central question is whether this adaptation is due to major-effect allelic sweeps or to subtle shifts in the frequency of many small effect alleles. In \textcite{Chevin_Hospital_2008}, the authors extend the framework of \textcite{Lande_1983} to include less specific selection functions than exponential ones. Subsequent studies (\cite{Vladar_Barton_2014,Jain_Stephan_2017}) explicitly model the short-term dynamics of a polygenic trait \bl{at} mutation-selection balance, following a sudden change of environment. They show that there exists a sharp threshold in allelic sizes \bl{below} which polymorphism remains and \bl{above} which fixation occurs. Lately, in a similar context, \textcite{Höllinger_Pennings_Hermisson_2019} propose an extension to take genetic drift into account on the dynamics of adaptation with a polygenic binary trait under mutation-selection balance. However, all those works from \textcite{Lande_1983} to \textcite{Höllinger_Pennings_Hermisson_2019} study panmicitic populations, without spatial structure, even though spatial heterogeneities are known to generate gene flow, which indirectly shapes genetic architecture through local adaptation (see \textcite{Yeaman_Whitlock_2011}, or below for more details). Moreover, they focus solely on the dynamics of the allelic frequencies without considering their coupling with population size dynamics, assuming it to be constant.

In this paper, we presented a composite framework between population and quantitative genetics aiming at going beyond these methodological limitations. Our model and methodology allows us to study analytically the eco-evo dynamics of a sexually reproducing population characterized by a composite trait resulting from the interplay between a few major-effect loci and a quantitative polygenic background, in spatially heterogeneous environments (migration-selection balance). We want to emphasize that, by "eco-evo dynamics", we mean that we study both the ecological and evolutionary dynamics of the local trait distributions and therefore do not assume that the sizes of the  populations remain constant; rather, they are variables of the system. \bl{This modelling choice is crucial, because the migrational meltdown phenomenon provoking the loss of polymorphism at the major-effect locus with strong selection relies on the building of asymmetrical subpopulation sizes.}}

\paragraph{The role of the Gaussian assumption of quantitative trait values.}

In our work, we \bl{justify} the Gaussian assumption made by \textcite{Lande_1983} and \textcite{Chevin_Hospital_2008} to model the background polygenic effect on the trait via a framework that does not make a priori assumptions on the within-population distribution. Instead, our model relies on an extension of the standard infinitesimal model (\cite{Fisher_1919}) that encodes both the inheritance of the quantitative background and the major-effect alleles. Analytical progress is possible in a regime of small segregational variance for the quantitative component of the trait, despite not specifying the shape of the trait distribution. It relies on the fact that the variance introduced at each event of reproduction by the quantitative background is small compared to the discrete allelic effects at the major-effect locus. This allows us to use a methodology developed by \textcite{Diekmann_Jabin_Mischler_Perthame_2005}, meant to study trait distributions concentrated as Dirac masses, \bl{to justify that assuming Gaussian distributions of quantitative trait values is valid (\cref{sec:model}). Moreover, this Gaussian approximation appears here sufficient to capture the phenomenon of migration meltdown with strong selection that we identify through the rest of our analysis, as higher order moments do not seem to have a significant influence on it}.

\paragraph{Extensions to more complex population genetic models.}
The model and the line of methods that we use in this paper are quite robust. We thus provide a comprehensive toolbox at the end of this manuscript (\cref{toolbox}), to describe how to apply the method more broadly. In particular, the toolbox is meant to indicate how to extend the method to more complex population genetic models by adding a quantitative background. It relies on \ref{prop:constraints_extension}\bl{, which justifies that carrying the analysis under the Gaussian assumption of quantitative trait distribution is valid (in the regime of small variance indicated)}. In \cref{toolbox}, we detail the hypotheses that the population genetic models must satisfy in order to use it (see \cref{app:generalization} for details and examples).

\paragraph{Further prospects.}
The loss of the polymorphism at the major-effect locus \bl{with} strong selection in a symmetrical heterogeneous environment, where one might think that it is most favoured, illustrates the value of our method. However, two natural questions stem from our work:
\begin{enumerate}
    \item Would the stability region of polymorphism at the major-effect locus shrink as much in the presence of a quantitative background when considering asymmetrical levels of selection/migration between the two patches? Our analysis suggest that it should, and this can be investigated through an extension of the last step of our analysis.
    \item Would this phenomenon hold if mutations can accumulate at the major-effect locus? \Cref{fig:stab_region} for example suggests that polymorphism at the major-effect locus would persist over a wider range of parameters if the alleles at the major-effect locus evolve to match the difference in optima. This possibility was indicated by the numerical findings of \textcite{Yeaman_Whitlock_2011}, who found the emergence of tightly linked clusters of major-effect loci underlying local adaptation for intermediate migration rates.
\end{enumerate}

\appendix

\section{Toolbox: How to study the interplay between a quantitative background and a finite number of major-effect loci.}
\label{sec:toolbox}
 The aim is to study the interplay between a quantitative background and a finite number of major-effect loci.
 
 We start with a population genetic model. {\color{Black}Let us consider $K$ different genotypes $\mathcal{A}^{(k)}$ which have genotypic effects on the phenotype $a^{(k)}$ (we use the index $k$ to indicate genotypes)}. For our method to be applied, the genotypes should verify two hypotheses H1 and H2 described in Appendix \ref{app:generalization}. The metapopulation lives in a heterogeneous environement of $I$ patches (we use the index $i$ to indicate location). We denote the population of patch $i$ carrying genotype $k$ by $N^{(k)}_i$. Let us denote the system of equations that describes the dynamics of the genotypic local population sizes: $\frac{d\bar{N}}{dT} = \tilde{G}_{\bar{a}}\left(\bar{N}(T)\right)$ and of a viable stable equilibrium $\bar{N}^*$. \textit{We recall that $\bar{N}^*$ is an equilibrium of the system if  ${\color{Black}\tilde{G}}_{\bar{a}}(\bar{N}^*)$. This equilibrium is viable if all the population sizes are non-negative, and at least one is positive. Its local stability is determined by standard linear analysis (sign of the real parts of the eigenvalues of the system's Jacobian).}

 Let us modify the previous population genetic framework to include the effect of a quantitative background on the trait, generically denoted $z$. While previously, all individuals carrying the same genotype $\mathcal{A}^{(k)}$ shared the same phenotype, now their phenotypes can differ due to the quantitative background they present. Consequently, among individuals of the same patch $k$ carrying the same major genotype $\mathcal{A}^{(k)}$, we distinguish those sharing the same quantitative background $z$, and denote their number $n_i^{(k)}(z)$:
 \begin{align*}
     &\mathcal{A}^{(k)} \quad\leadsto\quad(\mathcal{A}^{(k)}, z)\\
    &a^{(k)} \quad\leadsto \quad a^{(k)} + z\\
    &N^{(k)}_i \quad \leadsto \quad n^{(k)}_i(z).
 \end{align*}
 The PDE system that we obtain on the trait distributions $n_i^{(k)}$ is not easily analysed. That is why we provide a five steps plan in order to guide the analysis when the diversity introduced by the segregation of the quantitative component of the trait is small compared to the variance generated by the major-effect loci (H6 - regime of small variance):
 
 \begin{enumerate}
     \item {\color{Black}First}, we operate a scaling of time according to the regime of small variance. It anticipates on the separation of time scales such that the major-effect allelic frequencies change rapidly, followed by the slow changes of the quantitative components (see step 3).
     \item {\color{Black}In} this regime of small variance, we can justify the Gaussian approximation of the local genotypic distributions $n_i^{(k)}$ centered at the same mean and the same variance $\varepsilon^2$, thanks to \ref{prop:constraints_extension}, as soon as the assumptions (H1) and (H2) are satisfied (see Appendix \ref{app:generalization}) and every genotypic population randomly mates with themselves and every other genotypic population (H3) (the latter excludes for example models that differentiate sexes). This guides the intuition toward which change of variables to perform in order to get a system separating time scales explicitly (see Step 3). We emphasize that \ref{prop:constraints_extension} is crucial to be able to apply this method.
     
     \item  From the PDE system on the distributions $n_i^{(k)}$,  we can deduce the ODE system of their moments. Since we have justified the Gaussian approximation for all local genotypic distributions $n_i^{(k)}$, the new system is closed in the regime of small variance $\varepsilon^2\ll 1$, and only involves the dynamics of the genotypic local sizes of populations $N_i^{(k)}$ and the genotypic local mean quantitative components $z_i^{(k)}$.
     \item This step aims at obtaining a system that explicitly separates time scales, in order to ultimately reduce the complexity of the analysis. It requires a technical change of variables, which is guided by the formal analysis of the step 1 (mean quantitative components roughly the same within patches), and the intuition that migration has a strong blending effect between patches in the small variance regime (which would result in the mean quantitative components roughly being equal between patches). These considerations bring the following new variables replacing the genotypic local mean quantitative component $z_i^{(k)}$:
     \begin{itemize}
         \item[$\diamond$] for each genotype $1\leq k \leq K$, $\delta_{i,\varepsilon}^{(k)}$ is the difference in the mean quantitative component of the genotypic population $k$ between the patch $i+1$ and patch $i$ ($1\leq i \leq I-1$). Dividing by $\varepsilon^2$ comes from the intuition given above.
          \item[$\diamond$] for each genotype $1\leq k \leq K-1$, $\delta_\varepsilon^{(k)}$ is the difference between the mean quantitative component averaged across patches of genotypic population $k+1$ and $k$. Dividing by $\varepsilon^2$ comes from the intuition given in Step 1.
          \item[$\diamond$] $Z_\varepsilon$ is the overall mean quantitative component across patches and major genotypes. It is the slow evolving variable.
     \end{itemize}
     Rewriting the dynamics of the genotypic local population sizes $\bar{N}_\varepsilon$ along these new variables $\bar{\delta}_\varepsilon$ and ${\color{Black}Z_\varepsilon}$ delivers a system in which all the differential equations are multiplied by $\varepsilon^2$ (fast dynamics of $\bar{N}_\varepsilon$ and $\bar{\delta}_\varepsilon$) except the one governing the dynamics of $Z_\varepsilon$ (slow dynamics). 
     
     To finally complete the separation of time scales and obtain the limit system by letting $\varepsilon^2$ vanish, it is sufficient to show that at each value $Z$ of the slow variable, the fast time-scale equilibria $\left(\bar{N},\bar{\delta}\right)$ are stable, for example by using the Routh Hurwitz criterion for linear analysis on the Jacobian $\text{Jac}_{\mathcal{G}_{\bar{a}}}\left(\bar{N},\bar{\delta}\right)$.
     \item The last step to determine the stability of the global equilibria of the full system of the genotypic population with the influence of the quantitative background $\left(\bar{N}^*,\bar{\delta}^*,Z^*\right)$, consists in applying the formula given in the last box (see the next page).
 \end{enumerate}
 
 \newpage
 \label{toolbox}
 \newgeometry{textheight = 32cm, textwidth = 19cm, left = 2cm, tmargin = -1.2cm, inner = 1cm}

\begin{figure}
  \centering
  \begin{tikzpicture}
  \node[draw, rounded corners, align = center, text width = \textwidth] (title) {
  \textbf{\Large{Toolbox: How to study the interplay between a quantitative background\\ and a finite number of major-effect loci dynamics}.}
  };
  \node[below=1.7cm of title] (auxnode1) {};
  \node[draw, rounded corners,left = 3cm of auxnode1, align=left, text width = 5cm] (scenery){
  \textbf{The stadium}: \\$I$ patches $P_i$ ($1\leq i \leq I$)
  };
  \node[draw, rounded corners, right = 1cm of scenery, align=left] (teams){
  \textbf{The teams}:  \\$K$ different genotypes $\mathcal{A}^{(k)}$ $(1\leq k \leq K)$
  \\ Vector of genotypic effects on phenotype: $\bar{a}=(a^{(1)},...,a^{(k)})$
  \\ Matrix of local genotypic population sizes: $\bar{N}=\left(N_i^{(k)}\right)$
  };
  \node[draw, rounded corners, below = 1.5cm of scenery, align= left, text width = 5cm] (popgenmodel) {\textbf{Pop. gen. model}:\\
  $\frac{d\bar{N}}{dT} = \tilde{G}_{\bar{a}}\left(\bar{N}(T)\right)$
  };
  \node[draw, rounded corners, right = 1cm of popgenmodel, text width = 9cm, align= left] (popgenresult) {\textbf{Pop gen. analysis}:
  \begin{enumerate}
      \item[(i)] Viable equilibria: $G_{\bar{a}}\left(\bar{N}^*\right) = 0$ and $\bar{N}^*>\bar{0}$.
      \item[(ii)] Stability: eigenvalues of $\text{Jac}_{G_{\bar{a}}}\left(\bar{N}^*\right)$ in open left plane.
  \end{enumerate}};
  \node[draw, rounded corners, thick, inner xsep=1em, inner ysep=1em, fit=(scenery) (teams) (popgenmodel) (popgenresult)] (popgenanalysis) {};
  \node[fill=white] at (popgenanalysis.north) {\textbf{Population genetic model}};
  
  \node[draw, rounded corners, below = 2.62cm of popgenmodel, align=left, text width = 5.5 cm] (players) {
  \textbf{The new players}: \begin{enumerate}
      \item[(i)]Quantitative background $z$
      \item[(ii)]Individuals carrying genotype $i$ and a quantitative background $z$ have a phenotype $z+a^{(k)}$.
      \item[(iii)]Distribution in patch $k$ : $n_i^{(k)}(z)$
  \end{enumerate}};
  \node[right=0.8cm of players, rounded corners, draw, align = left, text width = 10cm] (hypothesis) {\textbf{Work hypotheses}:
  \begin{enumerate}
        \item[H1] - H2 (reflexivity and irreducible graph - see App.\ref{app:generalization})
      \item[H3] every genotypic population reproduces randomly with themselves and every other in the same patch
      \item[H4] inheritance of the quantitative background in accordance with the infinitesimal model with segregational variance $\sigma^2$.
      \item[H5] the quantitative background is unlinked to $\mathcal{A}^{(k)}$
      \item[H6] $\sigma^2 \ll \min\left|a^{(k)}\right|^2$: small variance regime.
  \end{enumerate}};
  
  \node[draw, rounded corners, thick, inner xsep=1em, inner ysep=1em, fit=(players) (hypothesis)] (hybridmodel) {};
  \node[fill=white] at (hybridmodel.north) {\textbf{\Large{Composite model combining population and quantitative genetics}}};
  \node[below=1.8cm of players, rounded corners, draw, align = left, text width = 6.5cm] (step0) {\textbf{0) Scaling of time} $t := \varepsilon^2\,T$\\ ($\varepsilon^2:=\frac{\sigma^2}{\min\left|a^{(k)}\right|^2}\ll 1\leadsto$ few diversity via inf. model of reproduction)};
   \node[below=0.3cm of step0, rounded corners, draw, align = left, text width = 6.5cm] (step1) {\textbf{1) Formal analysis} (justify Gaussian distributions - \ref{prop:constraints_extension}):
   \begin{enumerate}
       \item[(i)] $n_{i,\varepsilon}^{(k)}(z) \approx N_{i,\varepsilon}^{(k)}\times \text{Gauss}\left(z_{i,\varepsilon}^{(k)}, \varepsilon^2\right)$
       \item[(ii)] $z_{i,\varepsilon}^{(k)}\approx z_{i,\varepsilon}^{(l)}$
          \end{enumerate}};
   \node[below=0.3cm of step1, rounded corners, draw, align = left, text width = 6.5cm] (step2) {\textbf{2) ODE system of moments} ($\bar{z}_\varepsilon := (z^{(k)}_{i,\varepsilon})$):
   \begin{equation*}
       \begin{aligned}
          \begin{cases}
                 \varepsilon^2 \frac{d\bar{N}_\varepsilon}{dt} = G_{\bar{a}}\left(\bar{N}_\varepsilon(t),\bar{z}_\varepsilon(t)\right),\\
                 \varepsilon^2 \frac{d\bar{z}_\varepsilon}{dt} = F_{\bar{a}}\left(\bar{N}_\varepsilon(t),\bar{z}_\varepsilon(t)\right).
          \end{cases}
       \end{aligned}
   \end{equation*}};
   \node[right=0.6cm of step1, rounded corners, draw, align = left, text width = 10.2cm] (step3){\textbf{3) Slow-fast analysis}:\\
   \begin{enumerate}
       \item[(i)] Change in variables: $\delta_{i,\varepsilon}^{(k)} = \frac{z_{i+1,\varepsilon}^{(k)}-z_{i,\varepsilon}^{(k)}}{2\varepsilon^2}$ $\;[K(I-1)]\,$; $\delta^{(k)}_{\varepsilon} = \frac{\sum_i z_{i,\varepsilon}^{(k+1)} -z_{i,\varepsilon}^{(k)} }{2I\varepsilon^2}$ $\;[(K-1)]\,$; $\quad Z_\varepsilon =\frac{\sum_{k,i} z_{i,\varepsilon}^{(k)}}{K\times I} $
       \item[(ii)] Slow-fast system:
       \begin{equation*}
       \begin{aligned}
          \begin{cases}
                 \varepsilon^2 \frac{d\left[\bar{N}_\varepsilon,\bar{\delta}_\varepsilon\right]}{dt} = \mathcal{G}_{\bar{a}}\left(\bar{N}_\varepsilon(t),\bar{\delta}_\varepsilon(t),Z_\varepsilon\right),\\
                 \frac{dZ_\varepsilon}{dt} = \mathcal{F}_{\bar{a}}\left(Z_\varepsilon,\bar{N}_\varepsilon,\bar{\delta}_\varepsilon\right).
          \end{cases}
       \end{aligned}
   \end{equation*}
   \item[(iii)] Separation of time scales (via stability of zeros of $\mathcal{G}_{\bar{a}}$ by Routh-Hurwitz criterion on $\text{Jac}_{\mathcal{G}_{\bar{a}}}(\bar{N},\bar{\delta})$)
   \begin{equation*}
       \begin{aligned}
          \begin{cases}
                 0 = \mathcal{G}_{\bar{a}}\left(\bar{N},\bar{\delta},Z\right),\\
                 \frac{dZ}{dt} = \mathcal{F}_{\bar{a}}\left(Z,\bar{N},\bar{\delta}\right).
          \end{cases}
       \end{aligned}
   \end{equation*}
   \end{enumerate}};
    \node[below=0.3cm of step3, rounded corners, draw, align = left, text width = 10.6cm] (step4){\textbf{4) Analysis of the limit system}:
    \begin{enumerate}
        \item[(i)] Viable equilibria: $\mathcal{G}_{\bar{a}}\left(\bar{N}^*,\bar{\delta}^*,Z^*\right)=\mathcal{F}_{\bar{a}}\left(Z^*,\bar{N}^*,\bar{\delta}^*\right) = 0$, $\bar{N}^*>0$
        \item[(ii)] Stability:$\left.\nabla_{\bar{N},\bar{\delta}}\mathcal{F}_{\bar{a}}\cdot\left(\left[\text{Jac}_{\mathcal{G}_{\bar{a}}}\left(\bar{N},\bar{\delta}\right)\right]^{-1}\partial_Z\mathcal{G}_{\bar{a}}\right)\right|_{Z^*,\bar{N}^*,\bar{\delta^*}}>0$.
    \end{enumerate}};
    \node[draw, rounded corners, thick, inner xsep=1em, inner ysep=1em, fit=(step0) (step1) (step2) (step3) (step4)] (hybridanalysis) {};
     \node[fill=white] at (hybridanalysis.north) {\textbf{\Large{Steps to apply the analysis}}};

  \end{tikzpicture}
\end{figure}
\newpage
\restoregeometry
\section{Generalization of \ref{prop:constraints} for more complex genotypes.}
\label{app:generalization}
To state a generalization of \ref{prop:constraints}, we first need to specify the targeted scope of population genetic models. Let us consider $K$ different genotypes $\mathcal{A}^{(k)}$ that satisfies the following hypotheses relating to how they interact with each other regarding the genotypes of their offspring:
\begin{enumerate}
    \item[\textbf{H1}]{\textbf{Reflexivity}: For all $k \in (1,K)$, the offspring of two parents with the same genotype $\mathcal{A}^{(k)}$ has a positive probability to be of genotype $\mathcal{A}^{(k)}$}.
\end{enumerate}
\textbf{H1} is a natural hypothesis when considering either haploid or diploid populations, even with non-Mendelian processes (genetic linkage/recombination, gene drives), provided that they are not too extreme (lowering the probability of inheriting a certain genotype is fine as long as it does not cancel it). The second hypothesis is more conveniently apprehendable by considering the graph $\mathcal{G}$ whose nodes are the genotypes $\mathcal{A}^{(k)}$. A vertex links two nodes $\mathcal{A}^{(k)}$ and $\mathcal{A}^{(l)}$ if and only if there exists a positive probability that their offspring has genotype $\mathcal{A}^{(k)}$ or $\mathcal{A}^{(l)}$.
\begin{enumerate}
    \item[\textbf{H2}] \textbf{Irreducible graph}: For all $(k,l) \in (1,K)^2$, there exists a path of vertices of $\mathcal{G}$ connecting $\mathcal{A}^{(k)}$ and $\mathcal{A}^{(l)}$.
\end{enumerate}
This last hypothesis is satisfied by any haploid models, regardless of how many loci are considered, because an offspring can inherit all their alleles from only one parent. Consequently, in that case, every node of the graph is connected to every other. In diploid models, where an offspring can have a different genotype from both its parents, which vertices of the graph $\mathcal{G}$ exist is not clear. However, for example, we can show that the graph corresponding to a diploid model, with $L$ loci and two alleles at each loci, is connected according to \textbf{H2}. Indeed, each genotype is directly connected to any other that differs from it from just one allele at one locus. Nevertheless, the interest of \textbf{H2} is that it is very easy to verify whether it is satisfied given any particular model.

To state the proposition that generalizes \eqref{prop:constraints}, we first need to define the index set of couples that can yield an offspring with a particular genotype. For $k\leq K$, we denote it by $C^{(k)}$, where $(l,k) \in C^{(k)}$ if and only if parents with genotypes $\mathcal{A}^{(l)}$ and $\mathcal{A}^{(k)}$ can produce an offspring with genotype $\mathcal{A}^{(k)}$. The following proposition characterizes the genotypic functions $u^{\mathcal{A}^{(k)}}$ that respect the following constraints analogous to \ref{constraint_u0}
\begin{equation}
    \tag{C'}
    \forall k \leq K,\quad \forall z \in \R, \quad 
    \underset{(l,k) \in C^{(k)}}{\max}\left[\underset{z_1,z_2}{\sup}\;u^{\mathcal{A}^{(k)}}(z)- \left(z-\frac{z_1+z_2}{2}\right)^2-u^{\mathcal{A}^{(l)}}(z_1)-u^{\mathcal{A}^{(k)}}(z_2)\right] = 0.
    \label{eq:constraints_extension}
\end{equation}
\begin{prop}
\label{prop:constraints_extension}
Suppose that \textbf{H1} and \textbf{H2} are satisfied. For $k \leq K$, we consider $u^{\mathcal{A}^{(k)}}$ a real valued non-negative function whose zero set is non-empty and of measure 0 (for example, is finite). If $\{u^{\mathcal{A}^{(k)}}, k \leq K\}$ respects \eqref{eq:constraints_extension}, then there exists $z^*\in \R$ such that for all $k \leq K$:
\[\forall z \in \R, \quad u^{\mathcal{A}^{(k)}}(z) = \frac{(z-z^*)^2}{2}.\]
\end{prop}
\begin{proof}

~\paragraph{0) $u^{\mathcal{A}^{(k)}}$ is continuous and has right and left derivatives everywhere.}
For $k \leq K$ and $z \in \R$, we have:
\begin{equation}
        u^{\mathcal{A}^{(k)}}(z) - z^2= \underset{(l,k) \in C^{(k)}}{\min}\underset{z_1,z_2}{\inf}\left[\;-z(z_1+z_2)+\left(\frac{z_1+z_2}{2}\right)^2+u^{\mathcal{A}^{(l)}}(z_1)+u^{\mathcal{A}^{(k)}}(z_2)\right].
\label{eq:eq1_convexity}
\end{equation}

Therefore, $ u^{\mathcal{A}^{(k)}}(z) - z^2$ is concave as infimum of affine functions, and thus continuous and has right and left derivatives. 
~\paragraph{1) $u^{\mathcal{A}^{(k)}}$ cancels only once.}

Let us fix $k \leq K$. Let us suppose that $u^{\mathcal{A}^{(k)}}$ has two zeros $z_1^*\neq z_2^*$. \textbf{H1} implies that $(k,k) \in C^{(k)}$. Then, we deduce from \eqref{eq:constraints_extension} that:
\begin{equation*}
u^{\mathcal{A}^{(k)}}(z)\leq\underset{z_1,z_2}{\inf} \left(z-\frac{z_1+z_2}{2}\right)^2+u^{\mathcal{A}^{(k)}}(z_1)+u^{\mathcal{A}^{(k)}}(z_2).
\end{equation*}
In particular, for $z = \frac{z_1^*+z_2^*}{2}, z_1 = z_1^*, z_2=z_2^*$, we obtain
\begin{equation}
    u^{\mathcal{A}^{(k)}}(\frac{z_1^*+z_2^*}{2})\leq 0.
\end{equation}
As $u^{\mathcal{A}^{(k)}}$ is non-negative, the midpoint between two zeros of $u^{\mathcal{A}^{(k)}}$ is also a zero of $u^{\mathcal{A}^{(k)}}$. $u^{\mathcal{A}^{(k)}}$ is also continuous from the previous point, therefore, we deduce that $u^{\mathcal{A}^{(k)}}$ cancels on $[z_1^*,z_2^*]$. The latter violates the assumption that $u^{\mathcal{A}^{(k)}}$ has a zero set of measure 0. Because it is also not empty, we get that $u^{\mathcal{A}^{(k)}}$ cancels exactly once, in a point that we denote $z_k^*$.
~\paragraph{2) The zero of $u^{\mathcal{A}^{(k)}}$ coincides with the zero of $u^{\mathcal{A}^{(l)}}$: $z^*_k=z^*_l$.}
First, let us consider the case where $(k,l) \in (1,K)^2$ is such that $\mathcal{A}^{(k)}$ and $u^{\mathcal{A}^{(l)}}$ are linked by a vertex in the graph $\mathcal{G}$. Then, we deduce that $(k,l) \in C^{(k)}$ or $(k,l) \in C^{(l)}$. We can assume the first without loss of generality. Similarly as the first part of the proof, we deduce that
\begin{equation*}
u^{\mathcal{A}^{(k)}}(z)\leq\underset{z_1,z_2}{\inf} \left(z-\frac{z_1+z_2}{2}\right)^2+u^{\mathcal{A}^{(k)}}(z_1)+u^{\mathcal{A}^{(l)}}(z_2).
\end{equation*}
Consequently, the midpoint between $z_k^*$ and $z_l^*$ is a zero of $u^{\mathcal{A}^{(k)}}$, which is necessarily $z_k^*$, which implies that $z_k^*= z_l^*$.

Let us now show the same for every couple $(k,l)$ not necessarily linked by a vertex in $\mathcal{G}$. \textbf{H2} implies that there exists a path of vertices between $u^{\mathcal{A}^{(k)}}$ and $u^{\mathcal{A}^{(l)}}$. As we showed that for every pair of nodes connected by a vertex, the zeros of their function is the same point, that property also holds for the extremities of the path of vertices, hence $z_k^*= z_l^*$. We denote $z^*$ the common zero.

\paragraph{3) Convex Legendre conjugates $\hat{u^{\mathcal{A}^{(k)}}}(y) = \underset{z}{\sup}(z-z^*)y-u^{\mathcal{A}^{(k)}}(z)$.}
Let us show that \eqref{eq:constraints_extension} implies that the convex Legendre conjugate satisfies
\begin{equation}
    \forall y \in \R,\quad\hat{u^{\mathcal{A}^{(k)}}}(y) =\frac{y^2}{4} + \underset{(l,k) \in C^{(k)}}{\max}\left[ \hat{u^{\mathcal{A}^{(l)}}}\left(\frac{y}{2}\right)+ \hat{u^{\mathcal{A}^{(k)}}}\left(\frac{y}{2}\right)\right].
\label{eq:conjugates_extension}
\end{equation}
Using \eqref{eq:eq1_convexity} and commuting the $\sup$, we  obtain, for $y\in\R$,
\begin{equation}
\begin{aligned}
   \hat{u^{\mathcal{A}^{(k)}}}(y) &= \underset{z}{\sup}\left[(z-z^*)y-\underset{(l,k) \in C^{(k)}}{\min}\underset{z_1,z_2}{\inf}\left(\;\left(z-\frac{z_1+z_2}{2}\right)^2+u^{\mathcal{A}^{(l)}}(z_1)+u^{\mathcal{A}^{(k)}}(z_2)\right)\right]\\
    &= \underset{(l,k) \in C^{(k)}}{\max}\left[\underset{z_1,z_2}{\sup}\left(-u^{\mathcal{A}^{(l)}}(z_1)-u^{\mathcal{A}^{(k)}}(z_2)+\underset{z}{\sup}\;(z-z^*)y-\left(z-\frac{z_1+z_2}{2}\right)^2\right)\right].
\end{aligned}
\label{eq:eq2_convexity}
\end{equation}
Moreover, a straight-forward calculus shows that the $\underset{z}{\sup}$ is reached at $z = \frac{y+z_1+z_2}{2}$, which leads to
\begin{equation}
\begin{aligned}
    \underset{z}{\sup}\;(z-z^*)y-\left(z-\frac{z_1+z_2}{2}\right)^2 &=  \left(\frac{y+z_1+z_2}{2}-z^*\right)y-\frac{y^2}{4}\\  &=
     \frac{y^2}{4} + \left(z_1-z^*\right)\frac{y}{2} +\left(z_2-z^*\right)\frac{y}{2}.
\end{aligned}
\label{eq:eq3_convexity}
\end{equation}
Combining \eqref{eq:eq3_convexity} and \eqref{eq:eq2_convexity} (the fact that $z^A = z^a=z^*$ plays a crucial part for the crossed term) leads to \eqref{eq:conjugates_extension}.

Moreover, we obtain classically that: \begin{equation}
    \hat{u^{\mathcal{A}^{(k)}}}(y) \geq (z^*-z^*)y-u^{\mathcal{A}^{(k)}}(z^*) = 0 = \hat{u^{\mathcal{A}^{(k)}}}(0)
    \label{eq:sign_hat}
\end{equation}.
\paragraph{4) $\underset{k \leq K}{\max}\;\hat{u^{\mathcal{A}^{(k)}}}:y \mapsto \frac{y^2}{2}$.}

We obtain from \eqref{eq:conjugates_extension} that:
\begin{equation}
    \forall y \in \R,\quad \underset{k \leq K}{\max}\;\hat{u^{\mathcal{A}^{(k)}}}(y) =\frac{y^2}{4} + \underset{k \leq K}{\max}\;\underset{(l,k) \in C^{(k)}}{\max}\left[ \hat{u^{\mathcal{A}^{(l)}}}\left(\frac{y}{2}\right)+ \hat{u^{\mathcal{A}^{(k)}}}\left(\frac{y}{2}\right)\right].
\label{eq:conjugates_max_1}
\end{equation}
For $y\in\R$, let $k_0\leq K$ be such that $\underset{k \leq K}{\max}\;\hat{u^{\mathcal{A}^{(k)}}}\left(\frac{y}{2}\right)=\hat{u_{\mathcal{A}_{k_0}}}\left(\frac{y}{2}\right)$. \textbf{H1} implies in particular $(k_0,k_0) \in C^{(k_0)}$ and therefore, the maximum of the \bl{right-hand side} of \eqref{eq:conjugates_max_1} is reached in $2\hat{u_{\mathcal{A}_{k_0}}}\left(\frac{y}{2}\right)$. Consequently, we deduce that
\begin{equation}
    \forall y \in \R,\quad \underset{k \leq K}{\max}\;\hat{u^{\mathcal{A}^{(k)}}}(y) =\frac{y^2}{4} +2 \underset{k \leq K}{\max}\;\hat{u^{\mathcal{A}^{(k)}}}\left(\frac{y}{2}\right).
\label{eq:conjugates_max}
\end{equation}
Moreover, one can notice that 
\begin{align*}
    \forall y \in \R,\quad \underset{k \leq K}{\max}\;\hat{u^{\mathcal{A}^{(k)}}}(y) &= \underset{k \leq K}{\max}\;\underset{z\in\R}{\max}\; (z-z^*)y - u^{\mathcal{A}^{(k)}}(z) \\
    &= \underset{z\in\R}{\max}\; (z-z^*)y - \underset{k \leq K}{\min}\;u^{\mathcal{A}^{(k)}}(z)\\
    &= \hat{\left(\underset{k \leq K}{\min}\;u^{\mathcal{A}^{(k)}}\right)(y).}
\label{eq:conjugates_max}
\end{align*}
Therefore, $\underset{k \leq K}{\max}\;\hat{u^{\mathcal{A}^{(k)}}}$ is a convex continuous function that has left and right derivative everywhere, in particular in 0. Hence,  iterating \eqref{eq:conjugates_max} implies first that:
\begin{equation}
        \forall y >0\quad (\text{resp.} <0), \quad \underset{k \leq K}{\max}\;\hat{u^{\mathcal{A}^{(k)}}}(y) = \frac{y^2}{2}+ \beta\,y \quad(\text{resp.} \;\alpha\,y),
        \label{eq:max_hat_u_0}
    \end{equation}
where $(\alpha,\beta) = \left(\underset{k \leq K}{\max}\;\hat{u^{\mathcal{A}^{(k)}}}'(0^-),\underset{k \leq K}{\max}\;\hat{u^{\mathcal{A}^{(k)}}}'(0^+)\right)$. From \eqref{eq:sign_hat}, we deduce that the $\alpha\leq0\leq \beta$.  Since $\underset{k \leq K}{\max}\;\hat{u^{\mathcal{A}^{(k)}}}$ is the convex conjugate of $\underset{k \leq K}{\min}\;u^{\mathcal{A}^{(k)}}$, we compute that the convex bi-conjugate of $\underset{k \leq K}{\min}\;u^{\mathcal{A}^{(k)}}$ is 
 \begin{equation}
         z\mapsto
        \begin{aligned}
        \begin{cases}
        \frac{(z-z^*-\alpha)^2}{2}\text{ if }z<z^*+\alpha\\
        \qquad 0\qquad\text{ if }z^*+\alpha\leq z \leq z^*+\beta\\
        \frac{(z-z^*-\beta)^2}{2}\text{ if }z>z^*+\beta.
        \end{cases}
        \end{aligned}
        \label{eq:max_u_0_bi_conjugate}
    \end{equation}
As the convex bi-conjugate of $\underset{k \leq K}{\min}\;u^{\mathcal{A}^{(k)}}$ is the lower convex envelope of $\underset{k \leq K}{\min}\;u^{\mathcal{A}^{(k)}}$, the two of them are equal at the extremal points of its graph, namely for $z = z^*+ \alpha$ and $z=z^*+\beta$. We deduce from \eqref{eq:max_u_0_bi_conjugate} that
\begin{equation*}
    \underset{k \leq K}{\min}\;u^{\mathcal{A}^{(k)}}(z^*+\alpha) = \underset{k \leq K}{\min}\;u^{\mathcal{A}^{(k)}}(z^*+\beta) = 0.
\end{equation*}
Since all the $u^{\mathcal{A}^{(k)}}, k \leq K$ only cancels for $z=z^*$, we obtain that $\alpha = \beta = 0$ and \eqref{eq:max_hat_u_0} yields that $\underset{k \leq K}{\max}\;\hat{u^{\mathcal{A}^{(k)}}}:y\mapsto \frac{y^2}{2}$.

\paragraph{5) $\underset{k \leq K}{\max}\;\hat{u^{\mathcal{A}^{(k)}}} = \underset{k \leq K}{\min}\;\hat{u^{\mathcal{A}^{(k)}}}$.}

First let us state that $\underset{k \leq K}{\min}\;\hat{u^{\mathcal{A}^{(k)}}}$ is continuous as minimum of a finite number of continuous functions and that it is non-negative and reaches its minimum in 0, with $\underset{k \leq K}{\min}\;\hat{u^{\mathcal{A}^{(k)}}}(0)=0$ (from \eqref{eq:sign_hat}). Moreover, \eqref{eq:conjugates_extension} implies that
\begin{equation}
    \forall y \in \R,\quad \underset{k \leq K}{\min}\;\hat{u^{\mathcal{A}^{(k)}}}(y) \leq \frac{y^2}{4} + 2\,\underset{k \leq K}{\max}\;\hat{u^{\mathcal{A}^{(k)}}}\left(\frac{y}{2}\right) = \frac{y^2}{2}.
\label{eq:conjugates_min_1}
\end{equation}
Therefore $\underset{k \leq K}{\min}\;\hat{u^{\mathcal{A}^{(k)}}}$ has left and right derivatives in 0, and 
$\underset{k \leq K}{\min}\;\hat{u^{\mathcal{A}^{(k)}}}'\left(0^+\right) = \underset{k \leq K}{\min}\;\hat{u^{\mathcal{A}^{(k)}}}'\left(0^-\right)=0$.
Furthermore, \eqref{eq:conjugates_extension} also implies that \begin{equation*}
    \forall y \in \R,\quad \underset{k \leq K}{\min}\;\hat{u^{\mathcal{A}^{(k)}}}(y) \geq \frac{y^2}{4} + 2\,\underset{k \leq K}{\min}\;\hat{u^{\mathcal{A}^{(k)}}}\left(\frac{y}{2}\right).
\label{eq:conjugates_min_2}
\end{equation*}
Iterating the last inequality, and knowing that \[\underset{k \leq K}{\min}\;\hat{u^{\mathcal{A}^{(k)}}}\left(0\right)=\underset{k \leq K}{\min}\;\hat{u^{\mathcal{A}^{(k)}}}'\left(0^+\right) = \underset{k \leq K}{\min}\;\hat{u^{\mathcal{A}^{(k)}}}'\left(0^-\right)=0,\]
leads to
\begin{equation*}
    \forall y \in \R,\quad \underset{k \leq K}{\min}\;\hat{u^{\mathcal{A}^{(k)}}}(y) \geq \frac{y^2}{2} =\underset{k \leq K}{\max}\;\hat{u^{\mathcal{A}^{(k)}}}(y).
\label{eq:conjugates_min}
\end{equation*}
Consequently, we deduce that $\underset{k \leq K}{\min}\;\hat{u^{\mathcal{A}^{(k)}}} = \underset{k \leq K}{\max}\;\hat{u^{\mathcal{A}^{(k)}}}$. 
\paragraph{End of proof.}
The last result implies that 
\begin{equation*}
    \forall k \leq K,\quad\forall y \in \R, \quad \hat{u^{\mathcal{A}^{(k)}}}(y) =\underset{k \leq K}{\max}\;\hat{u^{\mathcal{A}^{(k)}}}(y) = \frac{y^2}{2}.
\end{equation*}
From the latter we compute the bi-conjugates $\hat{\hat{u^{\mathcal{A}^{(k)}}}}:z\mapsto \frac{(z-z^*)^2}{2}$. Since $z\mapsto \frac{z-z^*}{2}$ is strictly convex and it is the lower convex envelope of $u^{\mathcal{A}^{(k)}}$, we obtain that
\begin{equation*}
    \forall k \leq K,\quad\forall z \in \R, \quad u^{\mathcal{A}^{(k)}}(z) =\frac{(z-z^*)^2}{2}.
\end{equation*}
\end{proof}

\section{Formal justification of the constraints \eqref{constraint_u0} on the main terms $u_0^A$ and $u_0^a$}
\label{app:constraints}
We drop the index $i$ indicating the habitat and the time dependence $t$ for this appendix for the sake of simpler notations.

Let us first formally justify that $U_0^A$ and $U_0^a$ are positive almost everywhere and cancelling somewhere. As we are interested in the maintenance of the polymorphism at the major-effect locus, we consider that no major-effect allele has yet fixed. Hence, $N^A_\varepsilon$ and $N^a_\varepsilon$ need to remain positive and bounded when $\varepsilon$ vanishes. Using the Hopf-Cole transforms on $n^A_{\varepsilon}$ and $n^a_{\varepsilon}$ \eqref{eq:HCprel} along with the formal Taylor expansions \eqref{eq:taylor_u} on $U^A_{\varepsilon}$ and $U^a_{\varepsilon}$ leads to

\begin{equation}
    N_\varepsilon^A = \displaystyle\int_\R n_\varepsilon^A(z')\,dz' = \displaystyle\int_\R \frac{1}{\sqrt{2\pi}\varepsilon}e^{-\frac{U^A_{\varepsilon}(z')}{\varepsilon^2}}\,dz' = \displaystyle\int_\R \frac{1}{\sqrt{2\pi}\varepsilon}e^{-\frac{u^A_{0}(z')}{\varepsilon^2}}e^{-u_1^A+\varepsilon^2v^A_\varepsilon}\,dz'.
    \label{eq:N_eps_bounded}
\end{equation}
If we assume that the residues $u_1^A$ and $v_\varepsilon^A$ stay bounded when $\varepsilon$ vanishes (as \textcite{Calvez_Garnier_Patout_2019} suggests it), then \eqref{eq:N_eps_bounded} implies that $u^A_{0}$ must be non-negative for $N_\varepsilon^A$ to remain bounded when $\varepsilon$ vanishes. For $N_\varepsilon^A$ not to vanish asymptotically, $u^A_{0}$ must cancel. Moreover, for any interval $I\subset \R$, $u^A_{0}$ cannot cancel on $I$, or we would have:
\[N^A_\varepsilon > \displaystyle\int_I \frac{1}{\sqrt{2\pi}\varepsilon}e^\frac{1}{\varepsilon^2}e^{-u_1^A+\varepsilon^2v^A_\varepsilon}\,dz' \rightarrow +\infty.\]
So $u^A_{0}$ is positive almost everywhere, and cancelling somewhere. The same holds for $u^a_0$.

Now, for determining the constraints $\eqref{constraint_u0}$, let us notice that if we divide the \bl{right-hand side} of the first equality of \eqref{systscaled} by $n^A_{\varepsilon}(z)$, the reproduction term $\frac{\mathcal{B^A_\varepsilon}(n^A_{\varepsilon},n^a_{\varepsilon})(z)}{n^A_{\varepsilon}(z)}$ has to remain positive and bounded for all $z\in \R$ when $\varepsilon$ vanishes for the effect of reproduction to remain well-balanced with selection, migration and competition. We assume henceforth that \eqref{eq:taylor_u} is the correct ansatz (as suggested by \cite{Calvez_Garnier_Patout_2019}). Using the Hopf-Cole transforms on $n^A_{\varepsilon}$ and $n^a_\varepsilon$ \eqref{eq:HCprel} along with the formal Taylor expansions \eqref{eq:taylor_u} on $U^A_{\varepsilon}$ and $U^a_{\varepsilon}$ in \eqref{operator_2} leads to

\begin{equation*}
    \footnotesize
    \begin{aligned}
    &\frac{\mathcal{B^A_\varepsilon}(n^A_{\varepsilon},n^a_{\varepsilon})(t,z)}{n^A_{\varepsilon}(z)}=\frac{\mathcal{B^A_\varepsilon}(n^A_{\varepsilon},n^a_{\varepsilon})(z)}{\frac{1}{\sqrt{2\pi}\varepsilon}e^{-\frac{u^A_{0}(z)}{\varepsilon^2}}e^{-u_1^A(z)+\mathcal{O}(\varepsilon^2)}}\\
    &=  \frac{\sqrt{2}}{N_\varepsilon^A}\times\\
    &\left[\int_{\R^2} \exp\left(\frac{1}{\varepsilon^2}\left[u_0^A(z)-\left(z-\frac{z_1+z_2}{2}\right)^2-u_0^A(z_1)-u_0^A(z_2)\right]\right)\exp\left(u_1^A(z)-u_1^A(z_1)-u_1^A(z_2)+\mathcal{O}(\varepsilon^2)\right)dz_1 dz_2\right.\\
    &\left.+\!\int_{\R^2} \exp\left(\frac{1}{\varepsilon^2}\left[u_0^A(z)-\left(z-\frac{z_1+z_2}{2}\right)^2-u_0^A(z_1)-u_0^a(z_2)\right]\right)\exp\left(u_1^A(z)-u_1^A(z_1)-u_1^a(z_2)+\mathcal{O}(\varepsilon^2)\right)dz_1 dz_2\right]\!\!.
    \end{aligned}
\end{equation*}
As $N^A_\varepsilon$ remains bounded and does not vanish asymptotically, we need the maximum of the two integrals nor to vanish, nor to diverge to infinity when $\varepsilon$ vanishes, for all $z\in\R$. Therefore the maximum of the terms into brackets that are multiplied by $\frac{1}{\varepsilon^2}$ needs to be null for all $z\in\R$:
\begin{align*}
    \forall z \in \R,\quad \max\Big[\underset{z_1,z_2}{\sup}\;u_{0}^A(z)&- \left(z-\frac{z_1+z_2}{2}\right)^2-u_{0}^A(z_1)-u_{0}^A(z_2),\\ &\underset{z_1,z_2}{\sup}\; u_{0}^A(z) - \left(z-\frac{z_1+z_2}{2}\right)^2 -u_{0}^A(z_1)-u_{0}^a(z_2) \Big] = 0,
\end{align*}
which is the first constraint of \eqref{constraint_u0}. The same holds for $\frac{\mathcal{B^a_\varepsilon}(n^a_{\varepsilon},n^A_{\varepsilon})(z)}{n^a_{\varepsilon}(z)}$, which gives the second constraint of \eqref{constraint_u0}.
\section{Slow-fast analysis underlying the separation of time scales}

\label{app:slow_fast_analysis}
Under the change of variable \eqref{eq:slow_fast_variables}, the system \eqref{eq:syst_moment_small_var} is equivalent to the following:

\begin{equation}
\label{eq:syst_moment_small_var_fast_slow_friendly}
    \begin{aligned}
\begin{cases}
\varepsilon^2 \frac{d\,N^a_{\varepsilon,i}}{dt} = \lambda^{i-1}N^a_{\varepsilon,i} - \left[N^A_{\varepsilon,i}+N^a_{\varepsilon,i}\right] N^a_{\varepsilon,i} - g_i\left[Z_\varepsilon+\eta^a-(-1)^i\right]^2 N^a_{\varepsilon,i}\\\qquad\qquad\qquad\qquad\qquad\qquad\qquad\qquad\qquad\qquad\qquad+\alpha^{(-1)^j}m_j\,N^a_{\varepsilon,j}-m_i\,N^a_{\varepsilon,i} +\mathcal{O}(\varepsilon^2),
\\\\
\varepsilon^2 \frac{d\,N^A_{\varepsilon,i}}{dt} =\lambda^{i-1}\,N^A_{\varepsilon,i} - \left[N^A_{\varepsilon,i}+N^a_{\varepsilon,i}\right] N^A_{\varepsilon,i} - g_i\left[Z_\varepsilon+\eta^A-(-1)^i\right]^2 N^A_{\varepsilon,i} \\\qquad\qquad\qquad\qquad\qquad\qquad\qquad\qquad\qquad\qquad\qquad+\alpha^{(-1)^j}m_j\,N^A_{\varepsilon,j}-m_i\,N^A_{\varepsilon,i}+\mathcal{O}(\varepsilon^2),
\\\\
\varepsilon^2 \frac{d\,\delta^a_\varepsilon}{dt} = g_1+g_2 +(g_1-g_2)\,(Z_\varepsilon+\eta^a)+\frac{\delta_\varepsilon}{2}\left[\frac{N_{\varepsilon,2}^A}{N_{\varepsilon,2}^a+N_{\varepsilon,2}^A}-\frac{N_{\varepsilon,1}^A}{N_{\varepsilon,1}^a+N_{\varepsilon,1}^A}\right]
 \\\qquad\qquad\qquad\qquad\qquad- \delta^a_\varepsilon\left[\frac{m_2\alpha\,N^a_{\varepsilon,2}}{N^a_{\varepsilon,1}}+\frac{m_1N^a_{\varepsilon,1}}{\alpha\,N^a_{\varepsilon,2}}\right]+\frac{\delta^A_\varepsilon-\delta^a_\varepsilon}{4}\left[\frac{N_{\varepsilon,2}^A}{N_{\varepsilon,2}^a+N_{\varepsilon,1}^A}+\frac{N_{\varepsilon,1}^A}{N_{\varepsilon,1}^a+N_{\varepsilon,1}^A}\right]+\mathcal{O}(\varepsilon^2),\\
\varepsilon^2 \frac{d\,\delta^A_\varepsilon}{dt} = g_1+g_2+(g_1-g_2)\,(Z_\varepsilon+\eta^A) +\frac{\delta_\varepsilon}{2}\left[\frac{N_{\varepsilon,1}^a}{N_{\varepsilon,1}^a+N_{\varepsilon,1}^A}-\frac{N_{\varepsilon,2}^a}{N_{\varepsilon,2}^a+N_{\varepsilon,2}^A}\right]\\\qquad\qquad\qquad\qquad\qquad- \delta^A_\varepsilon\left[\frac{m_2\alpha N^A_{\varepsilon,2}}{N^A_{\varepsilon,1}}+\frac{m_1N^A_{\varepsilon,1}}{\alpha N^A_{\varepsilon,2}}\right]+\frac{\delta^a_\varepsilon-\delta^A_\varepsilon}{4}\left[\frac{N_{\varepsilon,2}^a}{N_{\varepsilon,2}^a+N_{\varepsilon,2}^A}+\frac{N_{\varepsilon,1}^a}{N_{\varepsilon,1}^a+N_{\varepsilon,1}^A}\right]
+\mathcal{O}(\varepsilon^2),\\

\varepsilon^2 \frac{d\,\delta_\varepsilon}{dt} = - \frac{\delta_\varepsilon}{2}-(g_1+g_2)\frac{\eta^A-\eta^a}{2}+\left(\frac{\delta_\varepsilon^A}{2}\left[\frac{\alpha m_2N^A_{\varepsilon,2}}{N^A_{\varepsilon,1}}-\frac{m_1N^A_{\varepsilon,1}}{\alpha N^A_{\varepsilon,2}}\right]-\frac{\delta_\varepsilon^a}{2}\left[\frac{\alpha m_2N^a_{\varepsilon,2}}{N^a_{\varepsilon,1}}-\frac{m_1N^a_{\varepsilon,1}}{\alpha N^a_{\varepsilon,2}}\right] \right)
\\
\qquad\qquad\qquad\qquad\qquad\qquad\qquad\qquad\qquad\qquad\qquad\qquad\qquad\qquad\qquad\qquad\qquad\;\;+ \mathcal{O}(\varepsilon^2),\\
\\
\frac{dZ_\varepsilon}{dt} = (g_2-g_1) -(g_1+g_2)\,\left(Z_\varepsilon+\frac{\eta^A+\eta^a}{2}\right)+\frac{\delta_\varepsilon}{2}\left[\frac{N^A_{1,\varepsilon}}{N^A_{1,\varepsilon}+N^a_{1,\varepsilon}}+\frac{N^A_{2,\varepsilon}}{N^A_{2,\varepsilon}+N^a_{2,\varepsilon}}-1\right] \\ \qquad\qquad\qquad\qquad\qquad\qquad+\left(\frac{\delta_\varepsilon^a}{2}\left[\frac{\alpha m_2N^a_{\varepsilon,2}}{N^a_{\varepsilon,1}}-\frac{m_1N^a_{\varepsilon,1}}{\alpha N^a_{\varepsilon,2}}\right]+\frac{\delta_\varepsilon^A}{2}\left[\frac{\alpha m_2N^A_{\varepsilon,2}}{N^A_{\varepsilon,1}}-\frac{m_1N^A_{\varepsilon,1}}{\alpha N^A_{\varepsilon,2}}\right]\right)\\
\qquad\qquad\qquad\qquad\qquad\qquad\qquad\qquad\qquad\qquad\quad+\frac{\delta_\varepsilon^A-\delta_\varepsilon^a}{4}\left[\frac{N^A_{2,\varepsilon}}{N^A_{2,\varepsilon}+N^a_{2,\varepsilon}}-\frac{N^A_{1,\varepsilon}}{N^A_{1,\varepsilon}+N^a_{1,\varepsilon}}\right]+\mathcal{O}(\varepsilon^2).
\end{cases}
\end{aligned}
\end{equation}
The system \eqref{eq:syst_moment_small_var_fast_slow_friendly} can be recasted more compactly into \ref{eq:slowfastvareps}. The main slow-fast analysis result is \cref{thm:slow_fast_theorem}, which states the convergence of \ref{eq:slowfastvareps} towards a limit system \ref{eq:slowfastvarlimit} which separates ecological and evolutionary time scales. The arguments of the proof of \cref{thm:slow_fast_theorem} are similar to the analogous theorems proved in \textcite{Levin_1954,Dekens_2022}. The proof requires some preliminaries results, particularly of stability, to which we dedicate the rest of this section. The structure of this section is represented in \cref{fig:slow_fast_layout}.
\begin{figure}
\centering
\begin{tikzpicture}
\centering
 \node[rounded corners=3pt, draw, color = RoyalBlue] at (-5,0) (Peps) {$\boldsymbol{(P_\varepsilon)}$ 
 \begin{math}
     \begin{aligned}
     \begin{cases}
         \text{\textbf{Fast dynamics} (allelic subpopulations sizes}\\\text{and spatial discrepancies between mean infinitesimal parts)},\\
         \\
       \text{\textbf{Slow dynamics} (mean infinitesimal part)}.
         \end{cases}
     \end{aligned}
 \end{math}};
 \node[rounded corners = 3pt,draw] at (-3.7,-3.) (thm) {\cref{thm:slow_fast_theorem}};
 \node[rounded corners = 3pt, below=0.05cm of thm] (cv) {$\varepsilon \rightarrow 0$};
 \node[rounded corners = 3pt, draw, below = 4cm of Peps, color = ForestGreen] (P0) {
 $\boldsymbol{(P_0)}$
\begin{math}\begin{aligned}
    \begin{cases}
    \text{\textbf{Fast ecological equilibria}},\\
         \\
       \text{\textbf{Slow evolutionary dynamics}}.
    \end{cases}
\end{aligned}\end{math}
};
\node[rounded corners = 3pt, draw, color = PineGreen] at (3.3, -4.3) (S0) {\begin{math}
    \begin{aligned}
    &\text{\textbf{\ref{S_0}} (allelic subpopulation sizes)},\\
    &\text{\scriptsize \emph{Resolution} \ref{prop:solut_S0}, \cref{rem:degrees_of_freedom}},\\
    &\text{\scriptsize $\colorboxed{Black}{\text{\emph{Stability}}}$ \ref{prop:stability_S0}}.
\end{aligned}
\end{math}} ;
\node[rounded corners = 3pt, draw, below= 0.2cm of S0, color = PineGreen] (SL) 
{\begin{math}\begin{aligned}
    &\text{\textbf{\ref{linear_syst}} (mean trait discrepancies)},\\
    &\text{\scriptsize \emph{Resolution} \ref{prop:linear_syst}},\\
    &\text{\scriptsize $\colorboxed{Black}{\text{\emph{Stability}}}$ \ref{prop:stability_linear}}.
\end{aligned}\end{math}};

\node[rounded corners = 3pt, draw, above= 3cm of S0, color = Gray] (OLHM) {One-locus haploid model};
\node [rounded corners = 3pt, draw, color = Gray]at  (5.2, -1.8) {\cref{rem:OLHM_fast_eq} ($Z=0$)};
 \draw[-stealth] (Peps) -- (P0);
 \draw[-stealth, dashed, color = PineGreen] (P0.9) -- (S0.170);
 \draw[-stealth, dashed, color = PineGreen] (P0.9) -- (SL.168);
 \draw[stealth-stealth, dashed, color = Gray] (S0) -- (OLHM) ;
\draw[-stealth, dashed, looseness=10] (SL.190) -- (thm.335);
\draw[-stealth, dashed, looseness=10] (S0.190) -- (thm.335);
 \end{tikzpicture}
        \caption{\textbf{Layout of the slow-fast analysis in \cref{app:slow_fast_analysis}}. This figure presents the key elements of the separation of time scales leading from \ref{eq:slowfastvareps} to \ref{eq:slowfastvarlimit}. The stability of the fast equilibria (studied in the two subsystems \ref{S_0} and \ref{linear_syst}) is the crucial argument underlying the convergence result stated in \cref{thm:slow_fast_theorem}. The resolution of \ref{S_0} leads to a uniqueness result that is unexpected with regard to the analogous resolution in \textcite{Dekens_2022} (see \cref{rem:degrees_of_freedom}). }
        \label{fig:slow_fast_layout}
\end{figure}
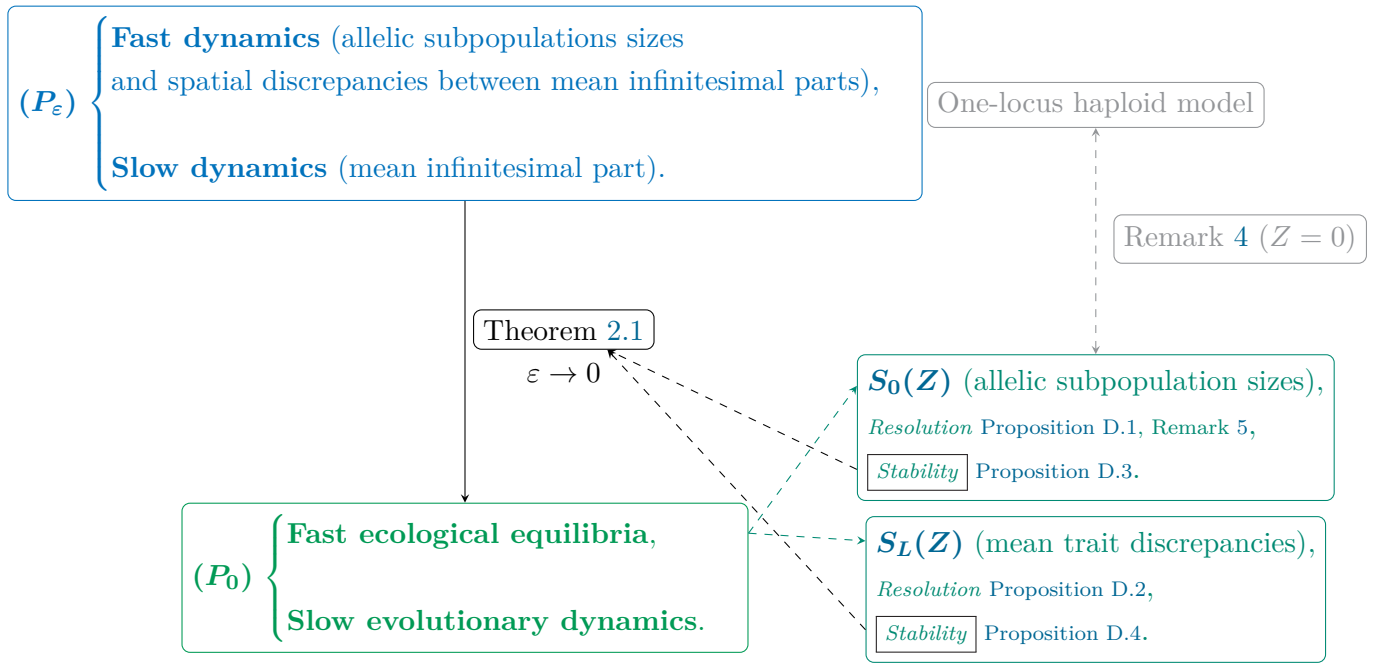
In the rest of this section, we first solve the slow manifold algebraic system $G(\Bar{Y},Z)=0$, showing that there can only exist one instantaneous ecological equilibrium at a given $Z\in]-1,1[$ (\ref{prop:solut_S0} and \ref{prop:linear_syst}). Surprisingly, this resolution is easier than the analogous one in \textcite{Dekens_2022} (see \cref{rem:degrees_of_freedom}). Next, in \cref{subsubsec:stability_slow_manifold}, we show a stability criterion of the slow manifold (\ref{prop:stability_S0} and \ref{prop:stability_linear}), which justifies the separation of time scales approach.
\subsection{Analysis of the fast equilibria.}
\label{subsec:slow_manifold}
The fast equilibria, for $Z\in]-1,1[$, are defined as the solutions $\Bar{Y}$ to the algebraic system $G(\Bar{Y},Z) = 0$, or equivalently seven equations that we group in two subsystems \ref{S_0} and \ref{linear_syst}:
\begin{equation}
    \begin{aligned}
    \begin{cases}
    \alpha \,m_2\,N_2^a -m_1 N_1^a + N_1^a\,\left[1 - (N_1^a + N_1^A) - g_1 \,(Z + \eta^a +1 )^2 \right] &=0,\\
    \alpha \,m_2\,N_2^A -m_1 N_1^A+N_1^A\,\left[1 - (N_1^a + N_1^A) - g_1\, (Z+ \eta^A+1 )^2 \right]&=0,\\
    \frac{m_1}{\alpha}\,N_1^a - m_2\,N_2^a+N_2^a\,\left[\lambda - (N_2^a + N_2^A) - g_2\, (Z + \eta^a -1 )^2  \right] &=0,\\
    \frac{m_1}{\alpha}\,N_1^A - m_2\,N_2^A+N_2^A\,\left[\lambda - (N_2^a + N_2^A) - g_2\, (Z + \eta^A -1  )^2 \right]&=0.\\
    \end{cases}
    \end{aligned}
    \label{S_0}
    \tag{$S_0(Z)$}
\end{equation}
\begin{equation}
J_{S_L} \begin{pmatrix}
        \delta\\\delta^A\\\delta^a
\end{pmatrix}=\begin{pmatrix}
        (g_1+g_2)\frac{\eta^A-\eta^a}{2}\\-(g_1+g_2) +(g_2-g_1)(Z+\eta^A)\\-(g_1+g_2) +(g_2-g_1)(Z+\eta^a)
\end{pmatrix},
\label{linear_syst}
\tag{$S_L(Z)$}
\end{equation}
where:
\begin{equation*}
\tiny{J_{S_L}\!\!:=\!\!\begin{pmatrix}
 
                   -\frac{1}{2}& \frac{1}{2}\left[\frac{\alpha\,m_2\,N_2^A}{N_1^A}-\frac{m_1\,N_1^A}{\alpha\,N_2^A}\right]&-\frac{1}{2}\left[\frac{\alpha\,m_2\,N_2^a}{N_1^a}-\frac{m_1\,N_1^a}{\alpha\,N_2^a}\right]\\\frac{\frac{N_1^a}{N_1^A+N_1^a}-\frac{N_2^a}{N_2^A+N_2^a}}{2}&-\left[\frac{\alpha\,m_2\,N_2^A}{N_1^A}+\frac{m_1\,N_1^A}{\alpha\,N_2^A}\right] - \frac{\frac{N_1^a}{N_1^A+N_1^a}+\frac{N_2^a}{N_2^A+N_2^a}}{4}&\frac{\frac{N_1^a}{N_1^A+N_1^a}+\frac{N_2^a}{N_2^A+N_2^a}}{4}\\
                    \frac{\frac{N_2^A}{N_2^A+N_2^a}-\frac{N_1^A}{N_1^A+N_1^a}}{2}&\frac{\frac{N_1^A}{N_1^A+N_1^a}+\frac{N_2^A}{N_2^A+N_2^a}}{4}&-\left[\frac{\alpha\,m_2\,N_2^a}{N_1^a}+\frac{m_1\,N_1^a}{\alpha\,N_2^a}\right] - \frac{\frac{N_1^A}{N_1^A+N_1^a}+\frac{N_2^A}{N_2^A+N_2^a}}{4}
\end{pmatrix}}
\end{equation*}
\subsubsection{Resolution.}
\label{subsubsec:resolution_slow_manifold}
Following \cref{rem:fixation_alleles_system}, we recall that we assume that no major-effect allele has fixed. Here, we show that there is at most one instantaneous ecological equilibrium at each $Z$-level (for $Z\in]-1-\frac{\eta^A+\eta^a}{2},1-\frac{\eta^A+\eta^a}{2}[$), thanks to \ref{prop:solut_S0} and \ref{prop:linear_syst}.
\begin{prop}
Suppose that no major-effect allele has fixed. Then, for $Z\in]-1-\frac{\eta^A+\eta^a}{2},1-\frac{\eta^A+\eta^a}{2}[$, \ref{S_0} has exactly one solution $(N_1^a, N_1^A, N_2^a, N_2^A) \in (\R^*)^4$, given by:
\begin{equation*}
    N^a_1 = \frac{Y^A\,N_1 - N_2}{Y^A-Y^a}, \quad N^a_2 = Y^a\frac{Y^A\,N_1 - N_2}{Y^A-Y^a}, \quad N^A_1 = \frac{N_2-Y^a\,N_1}{Y^A-Y^a}, \quad N^A_2 = Y^A\frac{N_2-Y^a\,N_1}{Y^A-Y^a},
    \label{eq:solutionS0}
\end{equation*}where the quantities $(Y^A, Y^a, N_1, N_2)$ are defined by \eqref{eq:YaYAN1N2}:
\begin{equation}
\label{eq:YaYAN1N2}
\begin{aligned}
    \begin{cases}
    Y^A = \frac{g_1}{\alpha\,m_2}\,\left(\eta^A+\eta^a+2\,(Z+1)\right)\frac{\eta^A-\eta^a}{2}\left[\sqrt{1+\frac{m_1\,m_2}{4\,g_1\,g_2\,\left(\frac{\eta^A-\eta^a}{2}\right)^2\left(1-\left(\frac{\eta^A+\eta^a}{2}+Z\right)^2\right)}}+1\right],\\
    Y^a = \frac{g_1}{\alpha\,m_2}\,\left(\eta^A+\eta^a+2\,(Z+1)\right)\frac{\eta^A-\eta^a}{2}\left[\sqrt{1+\frac{m_1\,m_2}{4\,g_1\,g_2\,\left(\frac{\eta^A-\eta^a}{2}\right)^2\left(1-\left(\frac{\eta^A+\eta^a}{2}+Z\right)^2\right)}}-1\right],\\
    N_1 = 1-g_1\,(Z +1+ \eta^A )^2-m_1+\alpha\,m_2\,Y^A,\\
    N_2 = \lambda-g_2\,(Z -1+ \eta^A )^2-m_2+\frac{m_1}{\alpha\,Y^A}.
    \end{cases}
    \end{aligned}
\end{equation}
This solution $(N_1^a, N_1^A, N_2^a, N_2^A)$ defines viable numbers of each allele in each sub-populations if and only if:

\begin{equation}
\label{eq:viability}
    [Y^A\,N_1 > N_2] \text{ and } [N_2>Y^a\,N_1].
\end{equation}
\label{prop:solut_S0}
\end{prop}

\begin{proof}
Let us introduce the following change of variables, valid under the assumption that no major-effect allele has fixed:
\[N_1 := N_1^A+N_1^a,\quad N_2 := N_2^A+N_2^a,\quad Y^A := \frac{N_2^A}{N_1^A},\quad Y^a := \frac{N_2^a}{N_1^a}.\]
Then, under the assumptions made in \cref{rem:fixation_alleles_system}, the system \eqref{S_0} is equivalent to:
\begin{equation}
    \begin{aligned}
    \begin{cases}
    \alpha \,m_2\,Y^a -m_1 + \left[1 - N_1 - g_1 \,(Z + \eta^a +1 )^2 \right] &=0,\\
    \alpha \,m_2\,Y^A -m_1 +\left[1 - N_1 - g_1\, (Z+ \eta^A+1 )^2 \right]&=0,\\
    \frac{m_1}{\alpha}\,\frac{1}{Y^a} - m_2+\left[\lambda - N_2 - g_2\, (Z + \eta^a -1 )^2  \right] &=0,\\
    \frac{m_1}{\alpha}\,\frac{1}{Y^A} - m_2+\left[\lambda - N_2 - g_2\, (Z + \eta^A -1  )^2 \right]&=0.\\
    \label{eq:resolution_YA_Ya_N1_N2}
    \end{cases}
    \end{aligned}
\end{equation}
This is equivalent to the following system:
\begin{equation*}
    \begin{aligned}
    \begin{cases}
     \alpha\,m_2\, (Y^a - Y^A) +\,g_1\,\left(\eta^A+\eta^a+2\,(Z+1)\right)(\eta^A-\eta^a)&=0,\\
    \frac{m_1}{\alpha}\, (\frac{1}{Y^a} - \frac{1}{Y^A}) +g_2\,\left(\eta^A+\eta^a+2\,(Z-1)\right)(\eta^A-\eta^a) &=0,\\
    N_1 - \left(1-g_1\,(Z +1+ \eta^A )^2-m_1+\alpha\,m_2\,Y^A\right) &=0,\\
    N_2 - \left(\lambda-g_2\,(Z -1+ \eta^A )^2-m_2+\frac{m_1}{\alpha\,Y^A}\right) &=0.\\
    \end{cases}
    \end{aligned}
\end{equation*}
As $Z\neq 1-\frac{\eta^A+\eta^a}{2}$, the closed subsystem on $(Y^A,Y^a)$ is, in turn, equivalent to:
\begin{equation*}
    \begin{aligned}
    \begin{cases}
     Y^A - Y^a &=A_1(Z):=\frac{g_1}{\alpha\,m_2}\,\left(\eta^A+\eta^a+2\,(Z+1)\right)(\eta^A-\eta^a),\\
    -Y^A\,Y^a &=A_0(Z):= \frac{g_1\,m_1}{\alpha^2\,g_2\,m_2}\frac{\eta^A+\eta^a+2\,(Z+1)}{\eta^A+\eta^a+2\,(Z-1)}.\\
    \end{cases}
    \end{aligned}
\end{equation*}
$Y^A$ and $-Y^a$ are the roots of the polynomial:
\[P(X) = X^2-A_1(Z)\,X+A_0(Z).\]
$P$ has two real roots of opposite signs if and only if:
\[\left[A_0(Z)<0\right],\]
which is equivalent to:
\[-1-\frac{\eta^A+\eta^a}{2}<Z<1-\frac{\eta^A+\eta^a}{2}.\]
Under the last condition on $Z$, $A_1(Z)$ is positive, $A_0(Z)$ is negative and we get:

\begin{equation}
    \begin{aligned}
        \begin{cases}
        Y^A = \frac{A_1(Z)}{2}\,\left[\sqrt{1-\frac{A_0(Z)}{\left(\frac{A_1(Z)}{2}\right)^2}}+1\right],\\
        Y^a = \frac{A_1(Z)}{2}\,\left[\sqrt{1-\frac{A_0(Z)}{\left(\frac{A_1(Z)}{2}\right)^2}}-1\right],
        \end{cases}
    \end{aligned}
\end{equation}which is equivalent to \eqref{eq:YaYAN1N2}.

Inverting the initial change of variables leads to:
\begin{equation*}
    N^a_1 = \frac{Y^A\,N_1 - N_2}{Y^A-Y^a}, \quad N^a_2 = Y^a\frac{Y^A\,N_1 - N_2}{Y^A-Y^a}, \quad N^A_1 = \frac{N_2-Y^a\,N_1}{Y^A-Y^a}, \quad N^A_2 = Y^A\frac{N_2-Y^a\,N_1}{Y^A-Y^a},
\end{equation*}
hence \eqref{eq:YaYAN1N2}. It defines a viable solution with positive entries if and only if $Y^A\,N_1 > N_2$ and $N_2>Y^a\,N_1$.
\end{proof}
\begin{prop}
For all allelic sizes of subpopulations $(N_1^a, N_1^A, N_2^a, N_2^A) \in \left(\R_+^*\right)^4$ and {\color{Black}$Z \in \R$}, there exists a unique solution $(\delta,\delta^A,\delta^a)$ to the system \ref{linear_syst}.
\label{prop:linear_syst}
\end{prop}
\begin{proof}
Using the notation $N_1 := N_1^A+N_1^a$ and $N_2 := N_2^A+N_2^a$, we compute thanks to a symbolic computation tool (Mathematica\copyright):
\begin{align*}
    \det(J_{S_L}) = &-\frac{1}{4}\left[\frac{m_1}{\alpha}\frac{N_1}{N_2}+\alpha\,m_2\frac{N_2}{N_1}+2\frac{m_1^2}{\alpha^2}\frac{N_1^a\,N_1^A}{N_2^a\,N_2^A}+2\alpha^2m_2^2\frac{N_2^a\,N_2^A}{N_1^a\,N_1^A} \right.\\
    &+ \left.2m_1m_2\left(\frac{{N_1^A}^2N_2^a}{N_1^aN_1N_2}+\frac{{N_1^a}^2N_2^A}{N_1^AN_1N_2}+\frac{{N_2^A}^2N_1^a}{N_2^aN_1N_2}+\frac{{N_2^a}^2N_1^A}{N_2^AN_1N_2} + 2\frac{N_2^AN_1^a}{N_1\,N_2}+2\frac{N_1^aN_2^a}{N_1N_2}\right)\right]\\\ &<0.
\end{align*}
\end{proof}
\subsubsection{Stability.}
\label{subsubsec:stability_slow_manifold}

Convergence toward a limit system locally in time in a slow-fast analysis relies essentially on a stability criterion of the fast equilibria which constitute the slow manifold (\cite{Levin_1954,Dekens_2022}). In this subsection, we show that all fast equilibria found in \ref{prop:solut_S0} and \ref{prop:linear_syst} for a level $Z\in]-1-\frac{\eta^A+\eta^a}{2},1-\frac{\eta^A+\eta^a}{2}[$, are stable. Due to the particular shape of the slow manifold, it is sufficient to study separately the Jacobian matrix associated to \ref{S_0} denoted $J_{S_0}$ (\ref{prop:stability_S0}) and the Jacobian matrix associated to the linear system \ref{linear_syst}, which is exactly $J_{S_L}$ (\ref{prop:stability_linear}).
\begin{prop}
Let $Z\in]-1-\frac{\eta^A+\eta^a}{2},1-\frac{\eta^A+\eta^a}{2}[$ such that (\ref{S_0}) has a unique solution $(N_1^a, N_1^A, N_2^a, N_2^A) \in \left(\R_+^*\right)^4$. Let us define the following matrix:
\begin{equation}
    J_{S_0}=\begin{pmatrix}
     
 -\frac{\alpha\,m_2\,N_2^a}{N_1^a}-N_1^a & \alpha\,m_2 & -N_1^a & 0 \\
 \frac{m_1}{\alpha} & -\frac{m_1\,N_1^a}{\alpha\,N_2^a}-N_2^a & 0 & -N_2^a \\
 -N_1^A & 0 & -\frac{\alpha\,m_2\,N_2^A}{N_1^A}-N_1^A & \alpha\,m_2 \\
 0 & -N_2^A & \frac{m_1}{\alpha} & -\frac{m_1\,N_1^A}{\alpha\,N_2^A}-N_2^A \\
    \end{pmatrix}.
\end{equation}
Then:
\begin{enumerate}
    \item $J_{S_0}$ is the Jacobian of \ref{S_0} at $(N_1^a, N_1^A, N_2^a, N_2^A)$.
    \item All the eigenvalues of $J_{S_0}$ are located in the left open half plane.
\end{enumerate}
\label{prop:stability_S0}
\end{prop}
\begin{proof}

1. Let $(N_1^a, N_1^A, N_2^a, N_2^A)$ be solution of \ref{S_0}. One can verify that:

\begin{align*}
    \frac{\partial \left[\alpha \,m_2\,N_2^a -m_1 N_1^a + N_1^a\,\left[1 - (N_1^a + N_1^A) - g_1 \,(Z + \eta^a +1 )^2 \right]\right]}{\partial N_1^a}\\=
    \left[1 - (N_1^a + N_1^A) - g_1 \,(Z + \eta^a +1 )^2-m_1\right]-N_1^a = -\frac{\alpha \,m_2\,N_2^a}{N_1^a}-N_1^a,
\end{align*}
for $(N_1^a, N_1^A, N_2^a, N_2^A)$ solves \ref{S_0}. The same holds for the other diagonal entries.

2. Let:\[\chi_{J_{S_0}} (X) = X^4-\tr\left(J_{S_0}\right)X^3+b\,X^2+c\,X +\det{J_{S_0}},\] be the characteristic polynomial of $J_{S_0}$. Let us verify the Routh-Hurwitz criterion: all the eigenvalues of $J_{S_0}$ are located in the left open half plane if and only if: 
\begin{itemize}
    \item[$(i)$]$\det{J_{S_0}}>0$,
    \item[$(ii)$]$-\tr\left(J_{S_0}\right)>0$,
    \item[$(iii)$]$-\tr\left(J_{S_0}\right)\,b-c>0$,
    \item[$(iv)$]$(-\tr\left(J_{S_0}\right)\,b-c)\,c-\tr\left(J_{S_0}\right)^2\,\det{J_{S_0}}>0$.
\end{itemize}
We have:
\[\det{J_{S_0}} = m_1\,m_2\left(N_1^a\,N_2^a-N_1^A\,N_2^A\right)^2\left(\frac{1}{N_1^a\,N_2^A}+\frac{1}{N_1^A\,N_2^a}\right) >0.\] 
and: \begin{align*}
    -\tr\left(J_{S_0}\right) &= N_1+N_2+
    \sqrt{m_1\,m_2}\left(\frac{\alpha \sqrt{m_2}N_2^a}{\sqrt{m_1}N_1^a}+\frac{\sqrt{m_1}N_1^a}{\alpha \sqrt{m_2}N_2^a}+\frac{\alpha\sqrt{m_2}N_2^A}{\sqrt{m_1}N_1^A}+\frac{\alpha\sqrt{m_2}N_2^A}{\sqrt{m_1}N_1^A}\right)>0.
\end{align*}
With the help of a symbolic computation tool (Mathematica\copyright), we verify that the left hand side of the two last conditions are sums of positive terms, but are too long to be displayed here.
\end{proof}
The Jacobian matrix of the linear system \ref{linear_syst} is exactly $J_{S_L}$ and we also show that $J_{S_L}$ satisfies the Routh-Hurwitz criterion:
\begin{prop}
$J_{S_L}$ has all its eigenvalues located in the left open half plane.
\label{prop:stability_linear}
\end{prop}
\begin{proof}
Let the following be the characteristic polynomial of $J_{S_L}$:\[\chi_{J_{S_L}}(X) = X^3 - \tr(J_{S_L})X^2 -\frac{1}{2}\left(\tr(J_{S_L}^2)-\tr(J_{S_L})^2\right)X-\det(J_{S_L}).\]
We show that $J_{S_L}$ satisfies the Routh-Hurwitz criterion:
\begin{itemize}
    \item[$(i)$] $-\det(J_{S_L})>0$,
    \item[$(ii)$] $-\tr(J_{S_L})>0$,
    \item[$(iii)$] $\frac{1}{2}\left(\tr(J_{S_L}^2)-\tr(J_{S_L})^2\right)\,\tr(J_{S_L})+\det(J_{S_L})>0.$
\end{itemize}
We have $-\det(J_{S_L})>0$ from the proof of \ref{prop:linear_syst} and:
\[-\tr(J_{S_L}) = 1 + \sqrt{m_1\,m_2}\left(\frac{\alpha \sqrt{m_2}N_2^a}{\sqrt{m_1}N_1^a}+\frac{\sqrt{m_1}N_1^a}{\alpha \sqrt{m_2}N_2^a}+\frac{\alpha\sqrt{m_2}N_2^A}{\sqrt{m_1}N_1^A}+\frac{\alpha\sqrt{m_2}N_2^A}{\sqrt{m_1}N_1^A}\right) >1 +4\sqrt{m_1m_2}.\]
We verify that the l.h.s. of the last condition is a sum of positive terms.
\end{proof}
\section{Proof of \ref{prop:sym_dim_eq}.}
\label{app:proof_sym_eq}
Let us define the quantities:
\begin{equation}
\label{eq:YaYAN1N2_0}
\begin{aligned}
    \begin{cases}
    Y^{A,*} = \frac{2\,g\,\eta}{m}\left[\sqrt{1+\frac{m^2}{4\,g^2\,\eta^2}}+1\right],\\
    Y^{a,*} = \frac{2\,g\,\eta}{m}\left[\sqrt{1+\frac{m^2}{4\,g^2\,\eta^2}}-1\right],\\
    N^*_1 = 1-g\,\eta^2-g-m+2\,g\,\eta\,\sqrt{1+\frac{m^2}{4\,g^2\,\eta^2}},\\
    N^*_2 = N^*_1.
    \end{cases}
    \end{aligned}
\end{equation}
\ref{prop:solut_S0} states that the latter defines a solution to \ref{S_0} given that $Z=0$:
\[(N^{a,*}_1,N^{a,*}_2,N^{A,*}_1,N^{A,*}_2) = \left(N^*_1 \frac{Y^{A,*} - 1}{Y^{A,*}-Y^{a,*}},\; N^*_1 \frac{1 - Y^{a,*}}{Y^{A,*}-Y^{a,*}},\;N^*_1 \frac{1 - Y^{a,*}}{Y^{A,*}-Y^{a,*}},\;N^*_1 \frac{Y^{A,*} - 1}{Y^{A,*}-Y^{a,*}}\right).\] Since $Y^{A,*}>1$ and $Y^{A,*}\; Y^{a,*}=1$, this solution is viable under the condition: $N_1>0$, hence requiring :
\[1+\sqrt{4\,g^2\,\eta^2+m^2}>g\,\eta^2+g+m,\]
which in turn is equivalent to \eqref{eq:condition_viability_0}.

\ref{prop:linear_syst} next states that \ref{linear_syst} has a unique solution $(\delta^*,\delta^{A,*},\delta^{a,*})$ for such allelic population sizes $(N^{a,*}_1,N^{a,*}_2,N^{A,*}_1,N^{A,*}_2)$. One can compute that:

\[\delta^{A,*} = \delta^{a,*} = \frac{g\left(1+\eta+Y^{A,*}(1-\eta)\right)}{m(1+Y^{A,*})}, \quad\delta^* = -\frac{2g\left(1+\eta-{Y^{A,*}}^2(1-\eta)\right)}{Y^{A,*}}.\]

Finally, one can verify that $(N^{a,*}_1,N^{a,*}_2,N^{A,*}_1,N^{A,*}_2,\delta^*,\delta^{A,*},\delta^{a,*})$ along with setting $Z^*=0$ is a solution of the last equation of \eqref{eq:limit_stationary_system}.
\section{Supplementary IBS with fixed subpopulations sizes}
\label{app:fixed_sizes}

\bl{In this appendix, we show in \cref{fig:comp_slim_fixed_sizes} the analogous results as those presented in \cref{fig:comp_slim_increasing_selection_sym}, but with a slightly different procedure for the IBS, which adjusts the birth rates to compensate exactly for the deaths by selection at each generation, thus keeping the subpopulations sizes fixed. As mentioned in the main text, the loss of the polymorphism at the major-effect locus still occurs with weak selection, but not with strong selection.
\begin{figure}\centering
\begin{subfigure}{.6\linewidth}
\centering
    \includegraphics[width=\linewidth]{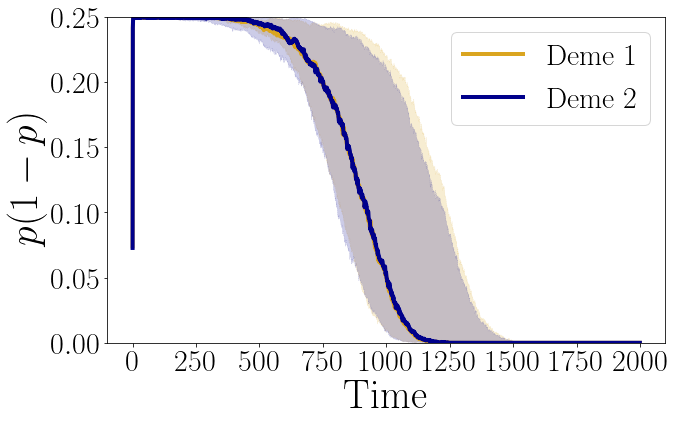}
    \subcaption{\small{\color{Black}Major-effect locus with polygenic background: weak selection ($g=0.1$)}.}
    \label{fig:m03g01_fixed_sizes}
\end{subfigure}
\\
\begin{subfigure}{.6\linewidth}
\centering
\includegraphics[width=\linewidth]{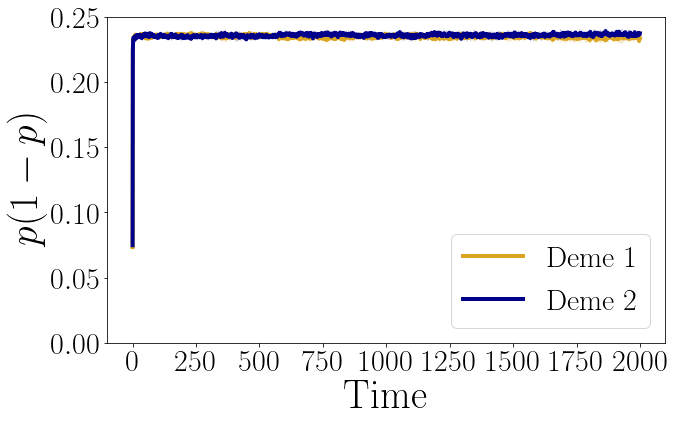}
    \subcaption{\small Major-effect locus with polygenic background: intermediate selection ($g=0.5$).}
    \label{fig:m03g05_fixed_sizes}
\end{subfigure}
\\
\begin{subfigure}{.6\linewidth}
\centering
    \includegraphics[width=\linewidth]{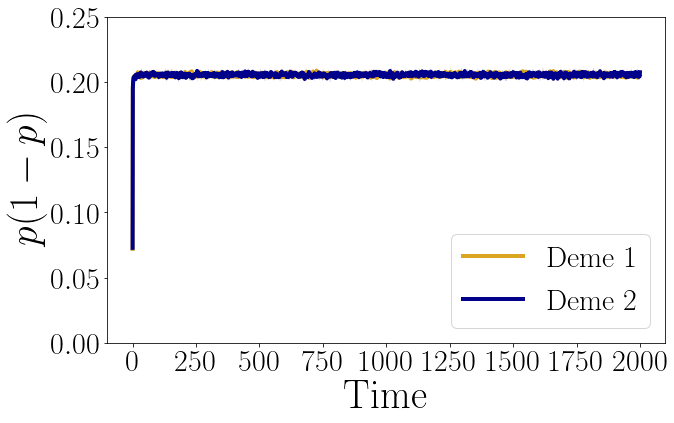}
    \subcaption{\small{\color{Black}Major-effect locus with polygenic background: strong selection ($g=1$).}}
    \label{fig:m03g1_fixed_sizes}
\end{subfigure}

\caption{\small\textbf{Same as the left panel of \cref{fig:comp_slim_increasing_selection_sym}, but with fixed subpopulations sizes}.}   
\label{fig:comp_slim_fixed_sizes}
\end{figure}}

\section{Supplementary IBS with asymmetrical initial conditions or different parameters for the genetic architecture ($\boldsymbol{L} = 50$ loci and $\boldsymbol{\sigma_{LE}}$ = 0.2)}
\label{app:50}
{\color{Black}
In this appendix, \bl{we first show that the phenomenon of loss of polymorphism in the presence of a quantitative background with weak or strong selection is robust to asymmetrical initial population sizes, and even occurs much faster (\cref{fig:comp_slim_increasing_selection_asym}). We emphasize on the excellent agreement of the deterministic iterations with the individual-based simulations}

\begin{figure}
\begin{subfigure}{.45\linewidth}
\centering
    \includegraphics[width=\linewidth]{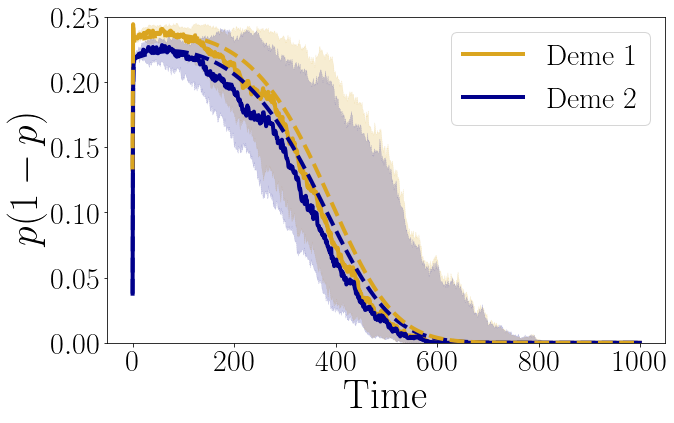}
    \subcaption{\small{\color{Black}Major-effect locus with polygenic background: weak selection ($g=0.1$)}.}
    \label{fig:m03g01_asym}
\end{subfigure}
\qquad\begin{subfigure}{.45\linewidth}
\centering
    \includegraphics[width=\linewidth]{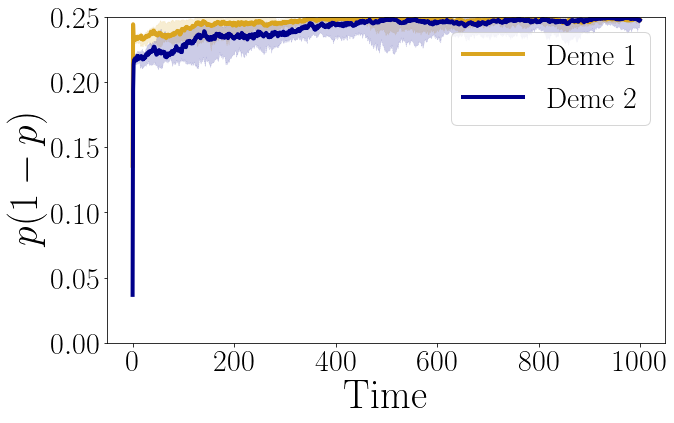}
    \subcaption{\small Control case without polygenic background: weak selection ($g=0.1$).}
    \label{fig:m03g01control_asym}
\end{subfigure}
\\
\begin{subfigure}{.45\linewidth}
\centering
\includegraphics[width=\linewidth]{p1_p2_8_2_all_m=0.80_g=0.50_a=0.1_Nloci=200_sym.png}
    \subcaption{\small major-effect locus with polygenic background: intermediate selection ($g=0.5$).}
    \label{fig:m03g05_asym}
\end{subfigure}
\qquad\begin{subfigure}{.45\linewidth}
\centering
    \includegraphics[width=\linewidth]{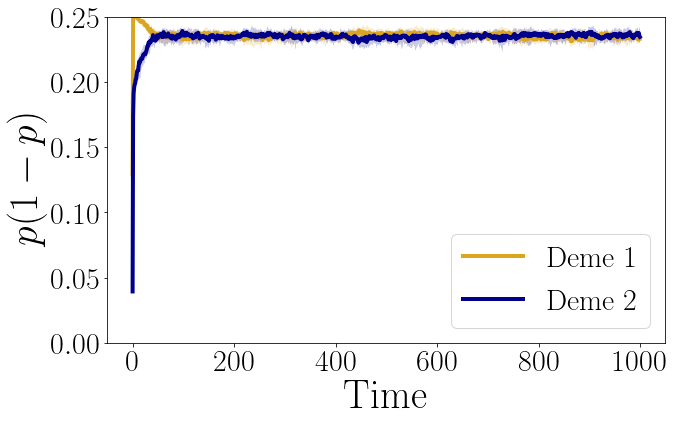}
    \subcaption{\small Control case without polygenic background: intermediate selection ($g=0.5$).}
    \label{fig:m03g05control_asym}
\end{subfigure}
\\
\begin{subfigure}{.45\linewidth}
\centering
    \includegraphics[width=\linewidth]{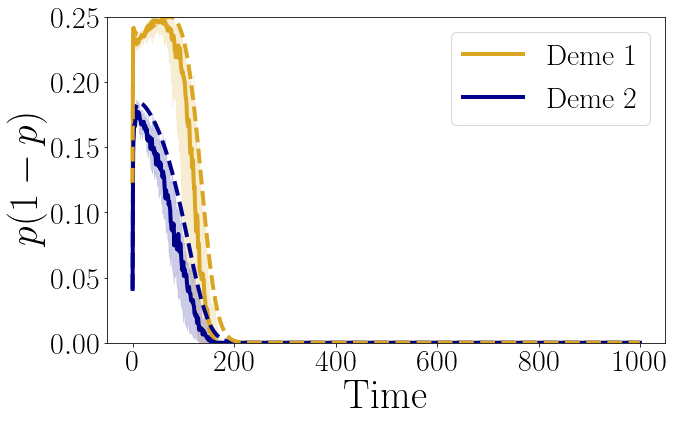}
    \subcaption{\small{\color{Black}Major-effect locus with polygenic background: strong selection ($g=1$).}}
    \label{fig:m03g1_asym}
\end{subfigure}
\qquad\begin{subfigure}{.45\linewidth}
\centering
    \includegraphics[width=\linewidth]{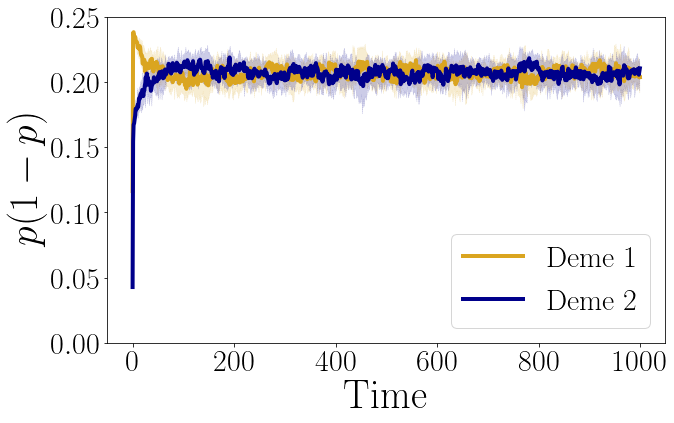}
    \subcaption{\small Control case without polygenic background : strong selection ($g=1$).}
    \label{fig:m03g1control_asym}
\end{subfigure}
\caption{\small\textbf{Same as \cref{fig:comp_slim_increasing_selection_sym}, {\color{Black} but with asymmetrical initial subpopulation sizes}}.}   
\label{fig:comp_slim_increasing_selection_asym}
\end{figure}
\bl{Furthermore, we also show} that our findings hold when considering a smaller number of loci involved in the quantitative background ($\boldsymbol{L} = 50$ instead of 200), with increased relative effect ($\boldsymbol{\sigma_{LE}}=0.2$ instead of 0.1), so that the trait range $[-\boldsymbol{\eta}-\boldsymbol{\sigma_{LE}}\sqrt{\boldsymbol{L}}, \boldsymbol{\eta}+\boldsymbol{\sigma_{LE}}\sqrt{\boldsymbol{L}}] \approx[-\boldsymbol{\eta}-1.4,\boldsymbol{\eta}+1.4]$ extends beyond the local optima (-1 and 1) even in the absence of major effects. We display the results of the IBS with symmetrical initial subpopulation sizes in \cref{fig:comp_slim_increasing_selection_sym_50} and with asymmetrical initial subpopulation sizes in \cref{fig:comp_slim_increasing_selection_asym_50}. Note that the right panel of each figure does not change from \cref{fig:comp_slim_increasing_selection_sym} and \cref{fig:comp_slim_increasing_selection_asym}, because the control case does not depend on the number of loci, but we choose to display it anyway for consistency of comparison. One can notice that the time to fixation at the major-effect locus in \bl{the} presence of \bl{a} quantitative background \bl{under} weak (\cref{fig:m03g1_sym_50}, \cref{fig:m03g01_asym_50}) and strong selection (\cref{fig:m03g1_sym_50}, \cref{fig:m03g1_asym_50}) is reduced compared to when the quantitative background comes from a larger number of loci (\cref{fig:comp_slim_increasing_selection_sym}, \cref{fig:comp_slim_increasing_selection_asym}). Moreover, the sensitivity of the numerical resolutions of \eqref{systnonstat} with regard to symmetrical initial states is more pronounced here.
\begin{figure}
\begin{subfigure}{.45\linewidth}
\centering
    \includegraphics[width=\linewidth]{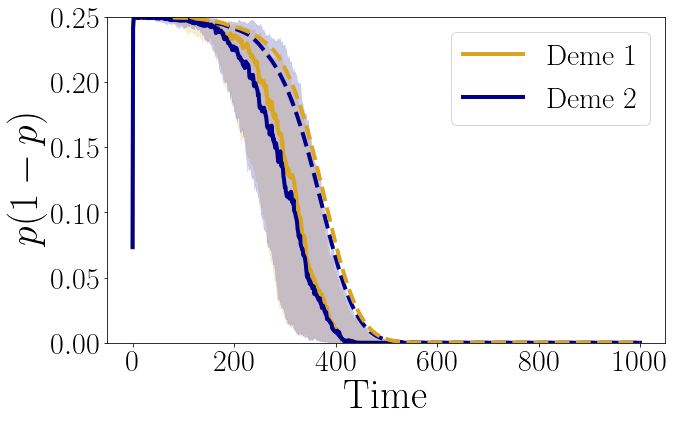}
    \subcaption{\small{\color{Black}Major-effect locus with polygenic background: weak selection ($g=0.1$)}.}
    \label{fig:m03g01_sym_50}
\end{subfigure}
\qquad\begin{subfigure}{.45\linewidth}
\centering
    \includegraphics[width=\linewidth]{p1_p2_8_2_control_m=0.80_g=0.10_sym.png}
    \subcaption{\small Control case without polygenic background: weak selection ($g=0.1$).}
    \label{fig:m03g01control_sym_50}
\end{subfigure}
\\
\begin{subfigure}{.45\linewidth}
\centering
\includegraphics[width=\linewidth]{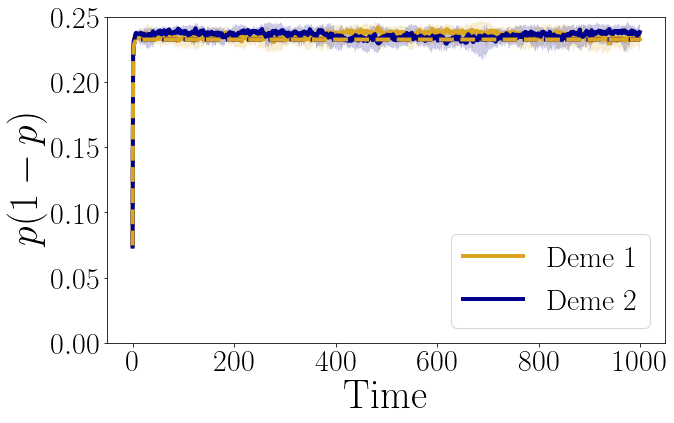}
    \subcaption{\small Major-effect locus with polygenic background: intermediate selection ($g=0.5$).}
    \label{fig:m03g05_sym_50}
\end{subfigure}
\qquad\begin{subfigure}{.45\linewidth}
\centering
    \includegraphics[width=\linewidth]{p1_p2_8_2_control_m=0.80_g=0.50_sym.png}
    \subcaption{\small Control case without polygenic background: intermediate selection ($g=0.5$).}
    \label{fig:m03g05control_sym_50}
\end{subfigure}
\\
\begin{subfigure}{.45\linewidth}
\centering
    \includegraphics[width=\linewidth]{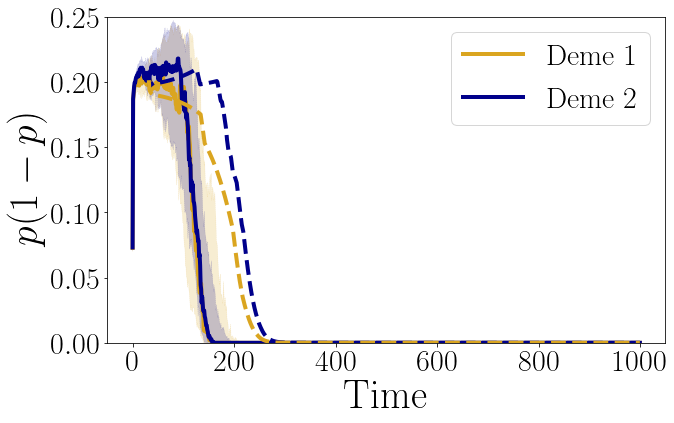}
    \subcaption{\small{\color{Black}Major-effect locus with polygenic background: strong selection ($g=1$)}}
    \label{fig:m03g1_sym_50}
\end{subfigure}
\qquad\begin{subfigure}{.45\linewidth}
\centering
    \includegraphics[width=\linewidth]{p1_p2_8_2_control_m=0.80_g=1.00_sym.png}
    \subcaption{\small Control case without polygenic background : strong selection ($g=1$).}
    \label{fig:m03g1control_sym_50}
\end{subfigure}
\caption{{\color{Black}\small\textbf{Same as \cref{fig:comp_slim_increasing_selection_sym}, but with a quantitative background of $\boldsymbol{L} = 50$ loci and $\boldsymbol{\sigma_{LE}} = 0.2$.}}}
\label{fig:comp_slim_increasing_selection_sym_50}
\end{figure}

\begin{figure}
\begin{subfigure}{.45\linewidth}
\centering
    \includegraphics[width=\linewidth]{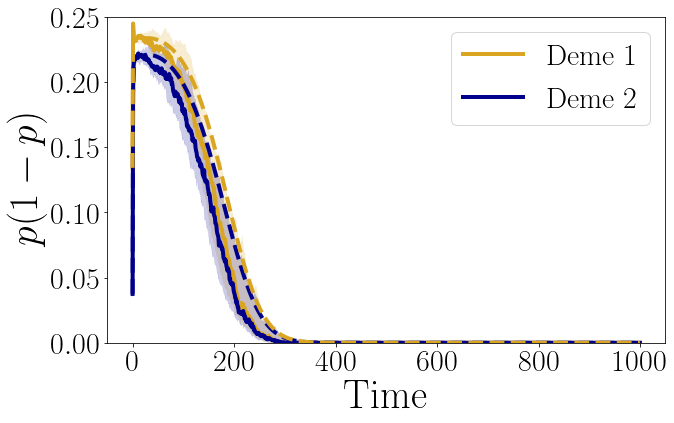}
    \subcaption{\small{\color{Black}Major-effect locus with polygenic background: weak selection ($g=0.1$)}.}
    \label{fig:m03g01_asym_50}
\end{subfigure}
\qquad\begin{subfigure}{.45\linewidth}
\centering
    \includegraphics[width=\linewidth]{p1_p2_8_2_control_m=0.80_g=0.10_asym.png}
    \subcaption{\small Control case without polygenic background: weak selection ($g=0.1$).}
    \label{fig:m03g01control_asym_50}
\end{subfigure}
\\
\begin{subfigure}{.45\linewidth}
\centering
\includegraphics[width=\linewidth]{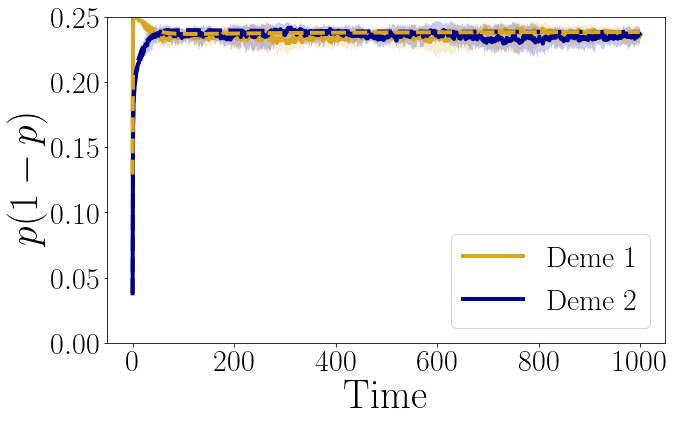}
    \subcaption{\small Major-effect locus with polygenic background: intermediate selection ($g=0.5$).}
    \label{fig:m03g05_asym_50}
\end{subfigure}
\qquad\begin{subfigure}{.45\linewidth}
\centering
    \includegraphics[width=\linewidth]{p1_p2_8_2_control_m=0.80_g=0.50_asym.png}
    \subcaption{\small Control case without polygenic background: intermediate selection ($g=0.5$).}
    \label{fig:m03g05control_asym_50}
\end{subfigure}
\\
\begin{subfigure}{.45\linewidth}
\centering
    \includegraphics[width=\linewidth]{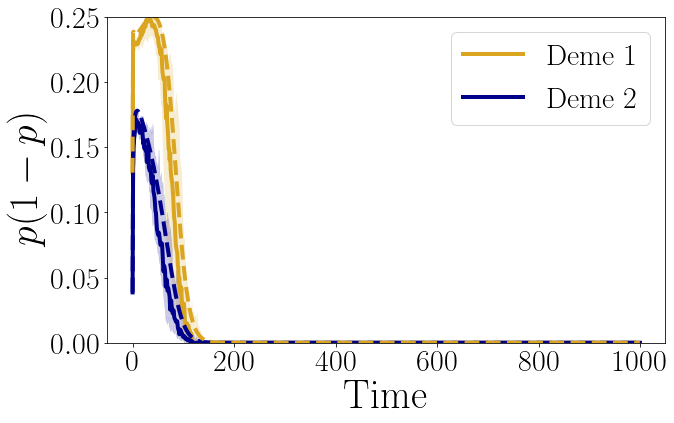}
    \subcaption{\small{\color{Black}Major-effect locus with polygenic background: strong selection ($g=1$).}}
    \label{fig:m03g1_asym_50}
\end{subfigure}
\qquad\begin{subfigure}{.45\linewidth}
\centering
    \includegraphics[width=\linewidth]{p1_p2_8_2_control_m=0.80_g=1.00_asym.png}
    \subcaption{\small Control case without polygenic background : strong selection ($g=1$).}
    \label{fig:m03g1control_asym_50}
\end{subfigure}
\caption{\color{Black}\small\textbf{Same as \cref{fig:comp_slim_increasing_selection_sym}, but with a quantitative background of $\boldsymbol{L} = 50$ loci and $\boldsymbol{\sigma_{LE}} = 0.2$ and asymmetrical initial subpopulations sizes.}}
\label{fig:comp_slim_increasing_selection_asym_50}
\end{figure}
}
\section*{Acknowledgements}

The authors thank John Wakeley and two anonymous reviewers for very constructive and detailed comments that, in their opinion, greatly improved the manuscript. L.D also thanks Sepideh Mirrahimi for valuable discussions and comments, and Florence D\'ebarre for insightful conversations. L.D has received partial funding from the
ANR project DEEV ANR-20-CE40-0011-01 and a Mitacs Globalink Research Award. S.P.O. has received funding from the NSERC Discovery Grant: RGPIN-2016-03711.
This project has received funding from the European Research Council (ERC) under the European Union's Horizon 2020 research and innovation programm (grant agreement No 865711).
\printbibliography
\end{document}